\pdfoutput=1
\documentclass[11pt,a4paper]{article}
\usepackage[nocompress]{cite}
\usepackage[margin=1.25in]{geometry}
\usepackage{amsmath,amssymb,amsthm}
\usepackage{enumitem}
\usepackage[T1]{fontenc}
\usepackage[sc]{mathpazo}
\usepackage{microtype}
\usepackage{flafter}
\usepackage[hidelinks]{hyperref}
\hypersetup{pdftitle={Stability of Matter Interacting with Classical Non-Abelian Gauge Fields},
  pdfauthor={Robin Reuvers, Jan Philip Solovej, David M. A. Stuart}}
\newtheorem{theorem}{Theorem}
\newtheorem{corollary}[theorem]{Corollary}
\newtheorem{proposition}[theorem]{Proposition}
\newtheorem{lemma}[theorem]{Lemma}
\theoremstyle{remark}
\newtheorem{remark}[theorem]{Remark}

\newcommand{\R}{\mathbb{R}}
\newcommand{\C}{\mathbb{C}}
\newcommand{\su}{\mathfrak{su}(2)}
\newcommand{\tr}{\operatorname{tr}}
\newcommand{\Tr}{\operatorname{Tr}}
\newcommand{\1}{\mathbf 1}
\newcommand{\E}{\mathbb E}
\newcommand{\balpha}{\boldsymbol{\alpha}}
\newcommand{\ux}{\underline{x}}
\newcommand{\dist}{\operatorname{dist}}
\newcommand{\curl}{\operatorname{curl}}
\newcommand{\LY}{\mathrm{LY}}
\newcommand{\zhat}{\hat z}

\title{Stability of Matter Interacting with Classical Non-Abelian Gauge Fields}
\author{Robin Reuvers\textsuperscript{1,*}, Jan Philip Solovej\textsuperscript{2} and
David M.\,A.\ Stuart\textsuperscript{3}}
\date{}

\begin{document}
\maketitle
\vspace{-2em}
\begin{center}\small
\textsuperscript{1}Dipartimento di Matematica e Fisica, Universit\`a degli Studi Roma Tre, L.go S.\,L.
Murialdo 1, 00146 Roma, Italy. E-mail: robin.reuvers@uniroma3.it\\
\textsuperscript{2}QMATH, Department of Mathematics, University of Copenhagen,
Universitetsparken 5, DK-2100 Copenhagen \O, Denmark. E-mail: solovej@math.ku.dk\\
\textsuperscript{3}Centre for Mathematical Sciences, University of Cambridge, Wilberforce Road,
Cambridge CB3 0WA, United Kingdom. E-mail: dmas2@cam.ac.uk\\[2pt]
\textsuperscript{*}Corresponding author
\end{center}

\begin{abstract}
Stability of matter is the property that the ground state energy of $N$ charged quantum particles is
bounded below by $-CN$, with $C$ independent of $N$. We ask whether this holds for colour-charged fermions interacting with a classical non-abelian
gauge field. The field is a dynamical variable on the same footing as the matter, and the energy is
minimized over both. The charges are matrices, and every pair of particles has an attractive colour channel. Non-relativistic spinless fermions are stable for every value of the coupling constant $\alpha$. Relativistic spinless fermions and non-relativistic spin-$\frac12$ fermions are stable for small
$\alpha$ and unstable for large $\alpha$. For Dirac fermions the choice
of positive-energy subspace decides: with that of the full Dirac operator they are stable for small
$\alpha$, with the free one unstable for every $\alpha$. Two species of bosons are unstable. In every case the outcome is the same as for electric
charges, although essential steps of the abelian proofs fail. The key tool is an electrostatic inequality for matrix-valued charges, uniform in
the field, which replaces the uniformly charged sphere of Newton's theorem by the surface charge of a
covariant harmonic extension. The stability results hold for every compact gauge group,
and most of them for non-smooth fields of finite energy.
\end{abstract}



\section{Introduction}\label{sec:intro}

If the energy of $N$ charged particles could decrease faster than linearly in $N$, bringing pieces of
matter together would release ever more energy per particle, and bulk matter would collapse. That it
does not is the stability of matter: the ground state energy $\mathcal E_N$ of $N$ charged quantum particles
satisfies $\mathcal E_N\ge-CN$, with $C$ independent of $N$. For electrons and nuclei interacting through the
Coulomb potential, this was first proved by Dyson and Lenard~\cite{DL67,LD68}.
Lieb and Thirring later found a short proof based on a kinetic energy inequality for
fermions~\cite{LT75,LT76} (see~\cite{Li76} for an early review), and the book by Lieb and
Seiringer~\cite{LS10} gives a comprehensive account of the theory. The mechanism is by now well
understood: the Pauli principle keeps the electrons apart, and electrostatics---Newton's theorem in one form or another---reduces the many-body Coulomb
energy to a sum of one-body terms that the kinetic energy can absorb.

In all of this, the charge of a particle is a number. Quarks carry a different kind of charge,
colour, which takes values in a representation of a non-abelian group, and they interact through a
non-abelian gauge field. In this paper we ask whether stability of matter survives this change.

For definiteness we take the gauge group to be $SU(2)$, although our stability results extend to
general compact groups, and to systems in which each species has its own representation and kinetic
energy (Theorem~\ref{thm:groups}). The particles carry $SU(2)$
colour and are minimally coupled to a classical $SU(2)$ gauge field. This is a vector field
$A=(A_1,A_2,A_3)$ on $\R^3$ with values in the Lie algebra $\su$. It acts on the colour space of each
particle through the covariant derivative $\nabla_A=\nabla+A$. Together with its time component $A_0$, also
$\su$-valued, the gauge field has the curvature $F_{\mu\nu}=\partial_\mu A_\nu-\partial_\nu A_\mu+[A_\mu,A_\nu]$, where
$\mu,\nu=0,\dots,3$ and $\partial_0=\partial_t$. The colour-magnetic field is $B_l=\frac12\varepsilon_{ljk}F_{jk}$, and the
colour-electric field is $E_j=F_{j0}=\partial_jA_0-\partial_tA_j+[A_j,A_0]$.

The field $E$ satisfies the Gauss law $\nabla_A\cdot E=4\pi\alpha\rho$, where $A$ acts on $E$
by $[A,\cdot\,]$ (the adjoint representation), $\alpha$ is the coupling
constant (the analogue of the fine structure constant) and $\rho$ is the colour charge density of the
particles. In electrodynamics, the Coulomb
interaction is what remains of the electric field once its longitudinal part has been eliminated with
the Gauss law. Doing the same here, by writing $E=E_\perp-\nabla_A\varphi$ with $\nabla_A\cdot E_\perp=0$
(Section~\ref{sec:colour} and Appendix~\ref{app:derivation}), produces an instantaneous interaction
$\alpha V_A$ with
\begin{equation}\label{eq:VAintro}
V_A=4\pi\sum_{1\le i<j\le N}e_ie_j\sum_{a,b=1}^3K_A^{ab}(x_i,x_j)\,\tau^a_{(i)}\tau^b_{(j)} .
\end{equation}
Here $x_i$ and $e_i$ are the position and
charge of particle $i$, the Pauli matrices $\tau^a_{(i)}$ act on its two-dimensional colour space,
and $K_A$ is the Green function of the covariant Laplacian $-\Delta_A=-\nabla_A\cdot\nabla_A$ in the adjoint representation.
The factor $4\pi$ makes $V_A$ reduce to a Coulomb potential, with matrix-valued charges, when $A=0$.

The potential~\eqref{eq:VAintro} differs from the Coulomb potential of electrodynamics in two ways.
First, the kernel $K_A$ depends on the gauge field. Second, the charge of particle $i$ is not a
number but the triple of matrices $e_i\tau^a_{(i)}$. At $A=0$ the kernel is
$\delta^{ab}/(4\pi|x_i-x_j|)$, and for two particles
\begin{equation}\label{eq:twobody}
V_0=\frac{e_1e_2}{|x_1-x_2|}\sum_{a=1}^3\tau^a\otimes\tau^a
=\frac{e_1e_2}{|x_1-x_2|}\times\begin{cases}-3&\text{on the colour singlet},\\
\phantom{-}1&\text{on the colour triplet}.\end{cases}
\end{equation}
Whatever the sign of $e_1e_2$, one of the two channels is attractive.

The gauge field is a dynamical variable, on the same footing as the particles. The state of the
system consists of a normalized fermionic wave function $\Psi$ and two classical fields, $A$ and $E_\perp$, which are independent variables in the Hamiltonian
description (Section~\ref{sec:hamiltonians}). Its energy is $\langle\Psi,H_{\rm tot}\Psi\rangle$, with
\begin{equation}\label{eq:Htot}
H_{\rm tot}=\sum_{i=1}^NT_{A,i}+\alpha V_A+\frac1{8\pi\alpha}\int_{\R^3}|B|^2+\frac1{8\pi\alpha}\int_{\R^3}|E_\perp|^2 ,
\end{equation}
where $T_{A,i}$ is the kinetic energy of particle $i$, of mass $m_i$:
$p_{A,i}^2/2m_i$ in the Schr\"odinger model, with $p_A=-i\nabla_A$ the covariant momentum, and
$(\sigma\cdot p_A)_i^2/2m_i$ in the Pauli model, with $\sigma$ the Pauli matrices acting on spin
(Section~\ref{sec:model}). In the relativistic models $T_{A,i}$ is $\sqrt{p_{A,i}^2+m_i^2}$ or the
Dirac operator (Section~\ref{sec:relmodels}). For $e_i=\frac12$ and the Schr\"odinger, Pauli and Dirac
kinetic energies, $H_{\rm tot}$ is the energy obtained from the Yang--Mills Lagrangian
(Appendix~\ref{app:derivation}). Formally, it generates the coupled time evolution of $\Psi$, $A$ and
$E_\perp$, which is that of a classical Yang--Mills field interacting with quantum matter. We use only the energy.
The ground state energy is now
\[
\mathcal E_N=\inf_{A,E_\perp}\inf_\Psi\langle\Psi,H_{\rm tot}\Psi\rangle ,
\]
the energy of matter with a self-generated field. Here $A$ ranges over smooth fields, and $\Psi$ over smooth,
compactly supported normalized states in the fermionic $N$-particle space of
Section~\ref{sec:hamiltonians}; for the Dirac model see Section~\ref{sec:relmodels}. Since $E_\perp$ enters only through a nonnegative
term, the infimum over $E_\perp$ is attained at $E_\perp=0$, and we drop it.

The question is therefore whether $\mathcal E_N\ge-CN$. For the two non-relativistic models the answer
is yes: for fermions,
\begin{align*}
\mathcal E_N&\ge-C_1\,M\alpha^2Q^4N\quad\text{in the Schr\"odinger model, for every }\alpha,\\
\mathcal E_N&\ge-C_2\,M\alpha^2Q^4N\quad\text{in the Pauli model, for $\alpha Q$ below an explicit threshold},
\end{align*}
where $M=\max_im_i$ and $Q=\max_i|e_i|$, and $C_1$ and $C_2$ are
explicit (Theorems~\ref{thm:spinless} and~\ref{thm:pauli}; for one species of spin-$\frac12$
particles the threshold is $\alpha^2Q^2\le9.00\times10^{-4}$). In the Schr\"odinger model the bound holds even if
the energy of the field is left out. These bounds have the right form: far-apart singlet pairs, each a hydrogen-like atom, have energy of
the same order (Remark~\ref{rem:sharp}). 

We are not aware of earlier results on stability of matter with non-abelian Coulomb interactions.
Three features of the abelian proofs are lost:
\begin{enumerate}[leftmargin=2em,itemsep=2pt]
\item \emph{There are no like charges.} In electrodynamics like charges repel, and stability
proofs exploit this at every turn. Here every pair of particles has an attractive colour channel,
by~\eqref{eq:twobody}. (Proposition~\ref{prop:electro} nevertheless bounds the Coulomb energy from
below.)
\item \emph{Newton's theorem fails.} In the abelian case, the potential of a uniformly charged sphere
equals that of a point charge outside the sphere. This is what allows one to smear out point charges
at no cost in the interaction energy. For the matrix-valued kernel $K_A$ there
is no such statement, and no obvious candidate for a `uniformly charged' sphere.
(Lemma~\ref{lem:newton} provides a substitute.)
\item \emph{Colour-charged particles cannot be fixed nuclei.} Two particles of mass $m$ in the colour
singlet form a hydrogen-like atom with reduced mass $m/2$ and coupling $3\alpha e^2$, so their energy
is $-\frac94m\alpha^2e^4$, which tends to $-\infty$ with the mass. So particles carrying colour cannot
play the role of infinitely heavy, mutually repelling nuclei. All colour-charged particles are
dynamical, and because of the attraction in~\eqref{eq:twobody} we take all of them to be fermions. The random splitting of Section~\ref{sec:kin} provides substitute nuclei. Static nuclei without
colour charge, coupled to the particles through an additional abelian field, can be included in
the Schr\"odinger model, as sketched in Remarks~\ref{rem:nuclei} and~\ref{rem:nucleiproof}.
\end{enumerate}

Newton's theorem also fails on a curved Riemannian manifold, and with it the usual proof of stability for
ordinary abelian matter. There,
smearing each charge with the harmonic measure on a level set of the Green function replaces the uniform
smearing over spheres. This was introduced in~\cite{Enc11} for complete manifolds with nonnegative Ricci
curvature, and extended in~\cite{Stu26} to asymptotically Euclidean manifolds, with bounds on the curvature
but no condition on its sign. Lemma~\ref{lem:newton} plays the same role for $K_A$, with spheres in place of
level sets: it replaces a point charge by a surface charge on a sphere whose potential outside the sphere is that of the
point charge.

\subsection{Summary of results}\label{sec:mainintro}

Section~\ref{sec:results} states our results precisely, and Table~\ref{tab:constants} collects them.
As in the abelian theory, whether stability holds depends on the kinetic energy.

For spin-$\frac12$ fermions stability fails for large $\alpha$, so the threshold in Theorem~\ref{thm:pauli}
cannot be removed. The Pauli operator has zero modes, and in a suitable field two particles attract
each other with no kinetic energy to oppose them. The field energy is then the only obstacle. It
prevails for small $\alpha$, and for large $\alpha$ the two particles collapse (Theorem~\ref{thm:collapse}).

For relativistic spinless fermions, stability holds for small $\alpha$ and fails for large $\alpha$
(Theorems~\ref{thm:rel} and~\ref{thm:relcollapse}), as Lieb and Yau showed in the abelian
case~\cite{LY88}.

For the Dirac operator the outcome depends on the positive-energy subspace to which the particles
are confined. If it is the positive spectral subspace of the Dirac operator with the gauge field, the
model is stable for small $\alpha$ (Theorem~\ref{thm:dirac}). If it is the free positive-energy
subspace, the model is unstable for every $\alpha$ (Theorem~\ref{thm:free}).

We also show that two species of bosons are unstable (Theorem~\ref{thm:bosons}), extend the
stability results to compact gauge groups and to mixtures of species with different representations
and kinetic energies (Theorem~\ref{thm:groups}), and extend the spinless
results to gauge fields in $L^2_{\rm loc}$ and the Pauli results to all such fields of finite energy
(Propositions~\ref{prop:rough} and~\ref{prop:roughpauli}).

\subsection{The abelian case}

The known results for electrons and nuclei follow the same pattern. Without spin, stability holds uniformly in the magnetic field, by the diamagnetic
inequality~\cite{LL01}, \cite[Sect.~4.4]{LS10}.

With spin, the field energy must be included. The Pauli operator has zero modes~\cite{LY86}, and
without the field energy even a one-electron atom has energy unbounded below as the field grows: in a
constant field $B$ the Zeeman term lowers the lowest Landau level to zero, and the binding energy below
it grows like $(\ln B)^2$~\cite[Theorem~2.5]{AHS81}. With the field energy, and with $Z$ the largest nuclear charge, stability of matter holds for small
$\alpha$ and $Z\alpha^2$~\cite{F95,LLS95} and fails for large values~\cite{FLL86,LL86}; the
one-electron atom was treated in~\cite{FLL86}, and the many-electron atom and the one-electron molecule
in~\cite{LL86}. Stability also holds with a quantized radiation field with an ultraviolet cutoff~\cite{BFG96,FFG97}.

With relativistic kinetic energy, stability holds for $Z\alpha\le2/\pi$ and small $\alpha$~\cite{LY88}
(see also~\cite{Co84,FdlL86}), also in magnetic fields~\cite{FLS07}, and fails for large
$\alpha$~\cite{LY88}. For the Dirac operator, the choice of positive-energy subspace matters. With the
free subspace, the energy is unbounded below for every $\alpha>0$~\cite{LSS97,GT99}. With the subspace
of the Dirac operator including the magnetic field, stability holds for small $\alpha$ and
$Z\alpha$~\cite{LSS97}, also for a quantized field with an ultraviolet cutoff~\cite{LLo02}.

For bosons stability fails: the energy of a neutral system of charged bosons behaves like
$-N^{7/5}$~\cite{Dy67,CLY88,LSo04,So06}.

In every case, then, the answer for colour charges is the same as for electric charges. What is new is
the proof, since the three features of the abelian proofs listed above are lost. The two theories also
differ in how the interaction is derived. In electrodynamics, eliminating the longitudinal electric
field gives the Coulomb potential with no choices to make. In Yang--Mills theory, the standard
Hamiltonian treatment uses the Coulomb gauge, where the kernel involves the Faddeev--Popov
operator~\cite[Sect.~IV]{CL80}. We instead split the electric field in a gauge-covariant way, which makes
the kernel the Green function of the covariant Laplacian in the adjoint representation.
Appendix~\ref{app:derivation} gives the derivation and explains why we make no claim about the
Coulomb-gauge Hamiltonian.

\subsection{Ideas of the proof}\label{sec:ideas}

All stability proofs in this paper follow the same route, so we describe it once and compare it with
the abelian argument that it generalizes. There are three steps: an electrostatic inequality, a
reduction to a one-body problem, and Lieb--Thirring-type bounds for that one-body problem. In the
Pauli and relativistic cases a fourth step absorbs an error from the centres.

\emph{Electrostatics.} Consider classical point charges $e_1,\dots,e_N$ at positions
$x_1,\dots,x_N$, and let $d_i$ be the distance from $x_i$ to its nearest neighbour. The Coulomb
energy satisfies
\begin{equation}\label{eq:abelianNN}
\sum_{i<j}\frac{e_ie_j}{|x_i-x_j|}\ \ge\ -\sum_i\frac{e_i^2}{d_i}.
\end{equation}
The proof is short. Spread each charge uniformly over the sphere of radius $d_i/2$ around
$x_i$. Each sphere lies outside the balls bounded by the others, so
by Newton's theorem the interaction energy between any two of them is unchanged. The total Coulomb
energy of the smeared charges is nonnegative, and it equals the left side of~\eqref{eq:abelianNN}
plus the self-energies $e_i^2/d_i$ of the spheres. Inequalities of this type go back to
Onsager; see~\cite[Chap.~5 and Sect.~6.4]{LS10}.

Proposition~\ref{prop:electro} is the non-abelian version of~\eqref{eq:abelianNN}: for every gauge
field,
\[
V_A\ \ge\ -3\sum_i\frac{e_i^2}{d_i},
\]
where the factor $3=\sum_a(\tau^a)^2$ comes from the three Pauli matrices. Two things in the abelian proof have to be
replaced. The first is the uniformly charged sphere. Take the potential $K_A(\cdot,y)$ of a point
charge at $y$, keep it outside the ball of radius $R$ around $y$, and replace it inside the ball by
the solution of $\Delta_Au=0$ with the same boundary values. The result is the potential of a
surface charge on the sphere, and outside the ball this potential is that of the point charge by
construction (Lemma~\ref{lem:newton}). The self-energy of the surface charge is controlled by the maximum principle, which
holds for solutions of $\Delta_Au=0$ because $|u|^2$ is subharmonic. The surface charge need not be
uniform or parallel to the point charge, but nothing in the argument requires this. The second thing to replace is
the positivity of the total energy, since the charges are now matrices. This is done by writing the
positive matrix of mutual energies of the surface charges as $\mathcal M^{\mathsf T}\mathcal M$ and
observing that $\sum_k(\sum_I\mathcal M_{kI}\hat e_I)^2\ge0$ also when the $\hat e_I$ are Hermitian
matrices (here $\hat e_I=e_i\tau^a_{(i)}$ for $I=(i,a)$). In the self-energy terms the matrices drop out, because Pauli matrices anticommute:
$\sum_{a,b}\mathcal K^{ab}\tau^a\tau^b=(\tr\mathcal K)\1$ for every symmetric matrix
$\mathcal K$.

\emph{Reduction to fixed centres.} After Proposition~\ref{prop:electro} we need a lower bound on
\begin{equation}\label{eq:NNproblem}
\sum_iT_i-\kappa\sum_i\frac1{d_i},
\end{equation}
where $T_i$ is the kinetic energy of particle $i$ without its mass factor, such as $p_{A,i}^2$, and $\kappa>0$ is proportional to $\alpha Q^2$ ($\kappa=6M\alpha Q^2$ in the non-relativistic models; see below). The difficulty is that the nearest-neighbour potential is not a one-body operator. Following Lieb and
Thirring~\cite{LT84} and Lieb and Yau~\cite{LY88}, we split the particles at
random into `centres' and `satellites', each particle being a centre with probability~$\frac12$.
With probability $\frac14$, particle $i$ is a satellite and its nearest neighbour is a centre, and
then $d_i$ is the distance from $x_i$ to the set of centres. So $\sum_id_i^{-1}$ is bounded by $4$
times the expected nearest-centre potential at the satellites, and $\sum_iT_i$ is
$2$ times the expected kinetic energy of the satellites, since each particle is a satellite half of
the time. Once the centres are fixed, the satellites see a fixed potential
$Z/\dist(x,\{\text{centres}\})$ with $Z=2\kappa$: the centres play the role of the nuclei. Other centre probabilities change only the constants, and Propositions~\ref{prop:kin0}
and~\ref{prop:relkin} use this freedom to improve them.

\emph{One-body bounds.} What remains is a one-body problem for fermions in the potential of the fixed centres. This is exactly the setting of the abelian proofs. The tool is the Lieb--Thirring
inequality, which carries over to $-\Delta_A$ because the heat kernel of $\Delta_A$ is dominated by
that of $\Delta$ (Lemmas~\ref{lem:dom} and~\ref{lem:count}). For the Pauli operator we follow the
argument of Lieb, Loss and Solovej~\cite{LLS95}, in which the Zeeman term is paid for by the field
energy. For relativistic kinetic energy we localize $|p_A|$ using a superadditivity property
(Lemma~\ref{lem:superadd}) together with Lieb and Yau's relativistic Hardy inequality in a
ball~\cite[Theorem~7]{LY88}. For the Dirac operator we use the inequality of Birman, Koplienko and Solomyak~\cite{BKS75}, as
in~\cite{LSS97}.

\emph{Absorption.} In the Pauli and relativistic cases, the abelian argument also uses the repulsion
between nuclei, which is not available here. Near a nucleus, the one-body bounds produce an error
proportional to $1/D$, with $D$ the distance to the nearest other nucleus, and in the abelian case the
nuclear repulsion $Z^2/D$ pays for it. Our centres are particles, and since there are no like charges,
they need not repel each other. The error is, however, of higher order in the coupling: it is of order
$\kappa^4\sum_jD_j^{-1}$, with $D_j$ half the distance from centre $j$ to the nearest other centre,
against the Coulomb term $\kappa\sum_jd_j^{-1}$. Since $D_j\ge d_j/2$, it can be absorbed into the
Coulomb term itself when $\kappa$ is small. This is the mechanism of Lieb and Yau~\cite{LY88}. In the
relativistic case, where kinetic and Coulomb energies scale in the same way, it forces the coupling to
be small, and by Theorem~\ref{thm:relcollapse} this cannot be avoided.

\emph{Putting the steps together.} For non-relativistic spinless fermions these steps already prove
Theorem~\ref{thm:spinless}: the electrostatic inequality reduces the Coulomb energy to the
nearest-neighbour potential, with strength $3\alpha Q^2$; the kinetic energy then absorbs this
potential at a cost of order $\kappa^2/M$ per particle, where $\kappa=6M\alpha Q^2$ and $M$ is the
largest mass (Proposition~\ref{prop:kin0}). Together they give the bound $-C_1\,M\alpha^2Q^4N$.

\paragraph{Organization of the paper.} Section~\ref{sec:model} defines the model, and
Section~\ref{sec:results} states the results, with remarks, a discussion and open problems.
Section~\ref{sec:tools} collects the facts about the covariant Laplacian that we use.
Section~\ref{sec:electro} proves the electrostatic inequality. Section~\ref{sec:kin} proves the kinetic
energy inequalities and, with them, Theorems~\ref{thm:spinless} and~\ref{thm:pauli} and
Proposition~\ref{prop:firstkind}. Section~\ref{sec:collapse} proves the collapse of
Theorem~\ref{thm:collapse} and the instability of bosons, and Section~\ref{sec:rel} the relativistic and
Dirac results. Section~\ref{sec:groups} proves the general Theorem~\ref{thm:groups}. Appendix~\ref{app:derivation}
derives the model from the Yang--Mills Lagrangian, Appendix~\ref{sec:rough} treats non-smooth
fields, and Appendix~\ref{sec:clr} proves the Cwikel--Lieb--Rozenblum bound for $p_A^2$.

\section{The model}\label{sec:model}

Sections~\ref{sec:colour}--\ref{sec:domains} define the $SU(2)$ model: colour and the Coulomb
potential, the particles and the non-relativistic Hamiltonians, the relativistic Hamiltonians, and
their domains. Section~\ref{sec:general} describes the general setting of Theorem~\ref{thm:groups}, and
Table~\ref{tab:constants} lists which theorem treats which $SU(2)$ model.

\subsection{Colour, the gauge field and the Coulomb potential}\label{sec:colour}

A particle with $SU(2)$ colour has a two-dimensional colour space $\C^2_{\rm col}$, on which the
Pauli matrices $\tau^1,\tau^2,\tau^3$ act. We write $|{\uparrow}\rangle$ and $|{\downarrow}\rangle$ for the
eigenvectors of $\tau^3$ with eigenvalues $1$ and $-1$, and call the corresponding states colour-up and
colour-down. Spin, when present, is described by a second copy of $\C^2$, with its own Pauli matrices
$\sigma_1,\sigma_2,\sigma_3$ (by $\C^4$ for the Dirac operator; Section~\ref{sec:relmodels}). The Lie algebra $\su$ of
anti-Hermitian traceless $2\times2$ matrices consists of the matrices
$X=-\frac i2\sum_aX^a\tau^a$ with $X^a\in\R$, and we identify $X$ with the vector
$(X^1,X^2,X^3)\in\R^3$, using the same letter for both. Under this identification the inner product
$\langle X,Y\rangle=2\tr(X^\dagger Y)$ is the Euclidean one, and $[X,Y]=X\times Y$. This is the
normalization of quantum chromodynamics, with generators $\tau^a/2$; in it, $|X|$ is $\sqrt2$~times the Hilbert--Schmidt norm of $X$. Throughout, an unsubscripted $\|\cdot\|$ is the norm of $L^2$ for
functions and the operator norm for matrices and bounded operators, and $|\cdot|$ the Euclidean norm of vectors in $\C^n$, $\R^3$ or $\su$; for self-adjoint
$M$, $|M|=(M^2)^{1/2}$.

The gauge field $A$ of Section~\ref{sec:intro} is smooth unless stated otherwise. It acts in two ways: on functions with values in the colour space $\C^2_{\rm col}$ by
\[
\nabla_A=\nabla+A
\]
(the fundamental representation), and on $\R^3$-valued functions by $\nabla_A=\nabla+[A,\cdot\,]$,
that is, $\nabla_A\phi=\nabla\phi+A\times\phi$ (the adjoint representation). In either case we write
\[
p_A=-i\nabla_A,\qquad p_A^2=-\Delta_A\ge0 ,
\]
and call $-\Delta_A$ the covariant Laplacian. The colour-magnetic field $B$ has components $B_l^a$, and
$|B|^2=\sum_{l,a}(B_l^a)^2$. We write $K_A(x,y)$ for the integral kernel of $(-\Delta_A)^{-1}$ on
$L^2(\R^3;\R^3)$ in the adjoint representation (Lemma~\ref{lem:dom}); it is a real $3\times3$ matrix, with entries $K_A^{ab}(x,y)$. At $A=0$ this is the ordinary Coulomb kernel,
$K_0^{ab}(x,y)=\delta^{ab}/(4\pi|x-y|)$.

The potential~\eqref{eq:VAintro} arises as follows. In electrodynamics, with the Gauss law
$\nabla\cdot E=4\pi\alpha\rho$, the energy of the longitudinal part $-\nabla\varphi$ of the electric field
is $\frac1{8\pi\alpha}\int|\nabla\varphi|^2=\frac\alpha2\iint\rho(x)\rho(y)|x-y|^{-1}\,dx\,dy$: the Coulomb
energy. In the Yang--Mills Gauss law of Section~\ref{sec:intro}, the colour charge density of a matter
field $\psi$ is $\rho=(\rho^1,\rho^2,\rho^3)$ with $\rho^a=\frac12\psi^\dagger\tau^a\psi$, and the
decomposition $E=E_\perp-\nabla_A\varphi$ gives, by the same computation, $\varphi=4\pi\alpha K_A\rho$ and
\begin{equation}\label{eq:coulombenergy}
\frac1{8\pi\alpha}\int|\nabla_A\varphi|^2=\frac\alpha2\sum_{a,b=1}^3\iint\rho^a(x)\,4\pi K_A^{ab}(x,y)\,\rho^b(y)\,dx\,dy .
\end{equation}
After quantization and normal ordering, in the sector with $N$ particles, this becomes $\alpha V_A$, with
$V_A$ as in~\eqref{eq:VAintro}. For each configuration $\ux=(x_1,\dots,x_N)$ of distinct points,
$V_A(\ux)$ is a Hermitian matrix on colour space; see Appendix~\ref{app:derivation} for the details. Throughout, $\alpha=g^2/4\pi$, where $g$ is the Yang--Mills coupling constant, scaled into the gauge
field so that the field energy is $\frac1{8\pi\alpha}\int|B|^2$ and the Coulomb energy carries the
factor $\alpha$. In the abelian theory the vector potential enters as $\sqrt\alpha
A$~\cite[Sect.~9.2]{LS10}; renaming $\sqrt\alpha A$ as $A$ gives the same conventions.

For particles in the fundamental representation the derivation gives $e_i=\frac12$, and
by~\eqref{eq:twobody} the singlet pair potential is $-\frac34\alpha/|x|$, with the familiar $SU(2)$
Casimir $C_F=3/4$, that is, $j(j+1)$ for the spin-$j$ representation with $j=\frac12$. We nevertheless allow general charges $e_i\in\R$, since the proofs do not use their
value. The charge $e_i$ enters only the Coulomb term, not the minimal coupling in the kinetic energy;
for $e_i\ne\frac12$ the models are therefore not consequences of the Lagrangian, but they make sense
in their own right.

\subsection{Particles and Hamiltonians}\label{sec:hamiltonians}

Particle $i$ has a mass $m_i>0$ ($m_i\ge0$ in the relativistic models of Section~\ref{sec:relmodels}) and a
charge $e_i\in\R$. The particles are divided into species.
Particles of one species are identical: if $i$ and $j$ are of the same species, then $m_i=m_j$, $e_i=e_j$ and
the two have the same internal space, so the Hamiltonians below are symmetric under permutations within each species. Particles of different
species are distinguishable, even when some or all of these parameters coincide. The internal space of a species $k$ is
\[
\mathfrak h_k=\C^{n_{{\rm s},k}}\otimes\C^{n_{{\rm c},k}}\otimes\C^{n_{{\rm f},k}},
\]
where $n_{{\rm s},k}$ is the number of spin states, $n_{{\rm c},k}$ the number of colours and $n_{{\rm f},k}\ge1$ the
number of flavours: labels on which nothing acts. We first consider the simpler case in which all
species carry the fundamental representation of $SU(2)$ and have the same type of kinetic energy, so
that $n_{{\rm c},k}=2$ and $n_{{\rm s},k}=n_{\rm s}$ for all $k$. In the Pauli model $n_{\rm s}=2$ and in the Dirac model
$n_{\rm s}=4$. In the spinless models the kinetic energy does not act on spin, so $n_{\rm s}=1$ and any
spin states are counted as flavours. The general case, with other groups and with representations
and kinetic energies that depend on the species, is described in Section~\ref{sec:general}. Each species is fermionic: with $N_k$ particles of species $k$, the $N$-particle
space is
\[
\mathcal H_N=\bigotimes_k\bigwedge^{N_k}L^2(\R^3;\mathfrak h_k),\qquad N=\sum_kN_k .
\]
For the bosons of Theorem~\ref{thm:bosons}, $\bigwedge^{N_k}$ is replaced by the symmetric tensor power.
The bounds below hold for every choice of the $N_k$, which we suppress from the notation. (Two species that differ only in their numbers of flavours amount to one species with more flavours,
restricted to a sector of fixed flavour numbers; the Hamiltonians preserve such sectors.) We write
\[
M=\max_im_i,\qquad Q=\max_i|e_i|,\qquad q=\sum_k\dim\mathfrak h_k=\sum_kn_{{\rm s},k}\,n_{{\rm c},k}\,n_{{\rm f},k} ,
\]
where $M$ and $Q$ are the largest mass and the largest absolute charge among the species present, and
$q$ is the total number of internal states. Only this total enters the results, because, for the one-body bounds we use, several fermionic
species behave like a single species with $q$ internal states (Lemma~\ref{lem:slice}); this is the convention of~\cite[Sect.~3.1.2]{LS10}, where all internal states
are counted as spin states. The smallest values are $q=2$ (one species, spinless), $q=4$ (with spin)
and $q=8$ (Dirac); more species or flavours increase $q$ accordingly.

In~\eqref{eq:Htot}, the norm of $E_\perp$ is given by $|E_\perp|^2=\sum_{l,a}(E_{\perp,l}^a)^2$, as for $|B|$ in Section~\ref{sec:colour}. In the Hamiltonian formulation of
Appendix~\ref{app:derivation}, $-E_\perp/(4\pi\alpha)$ is the part of the momentum conjugate to $A$ that
the Gauss law leaves free.

The magnetic energy in~\eqref{eq:Htot} matters for the true ground state: the field changes the
Coulomb kernel $K_A$ and couples to the currents of the particles, so the infimum over $A$ need not
be attained at $A=0$. For a lower bound without spin, however, the magnetic energy can be dropped,
since our bounds hold even without it. With spin it must be kept, as explained below. We therefore
study, for a given field $A$, the Schr\"odinger and Pauli Hamiltonians
\begin{align}
H_S&=\sum_{i=1}^N\frac{1}{2m_i}\,p_{A,i}^2+\alpha V_A, \label{eq:H0}\\
H_P&=\sum_{i=1}^N\frac{1}{2m_i}(\sigma\cdot p_A)^2_i+\alpha V_A+\frac{1}{8\pi\alpha}\int_{\R^3}|B|^2 .
\label{eq:HP}
\end{align}
Thus $\mathcal E_N\ge\inf_A\inf_\Psi\langle\Psi,H_S\Psi\rangle$ in the Schr\"odinger model, and
$\mathcal E_N=\inf_A\inf_\Psi\langle\Psi,H_P\Psi\rangle$ in the Pauli model, with $\Psi$ again over smooth,
compactly supported normalized states in $\mathcal H_N$. Since $\sigma_j\sigma_k=i\sum_l\varepsilon_{jkl}\sigma_l$ for $j\ne k$ and
$[\nabla_{A,j},\nabla_{A,k}]=F_{jk}$, the non-abelian Pauli operator is given by the flat-space,
gauge-twisted case of the Lichnerowicz formula~\cite{Li63}
\begin{equation}\label{eq:pauli}
(\sigma\cdot p_A)^2=p_A^2-i\sum_l\sigma_l\otimes B_l
=p_A^2-\frac12\sum_{l,a}B_l^a\,\sigma_l\otimes\tau^a .
\end{equation}
The second term is the colour-magnetic (Zeeman) coupling, with gyromagnetic factor $2$. It is the
leading spin-dependent term of non-relativistic quantum chromodynamics~\cite{CaL86,BBL95}. It is also
the reason why the field energy appears in $H_P$ but not in $H_S$: without the field energy, zero modes
of the Pauli operator~\cite{LY86} would make the energy unbounded below (Section~\ref{sec:collapse}), as
in the abelian case~\cite{AHS81,FLL86}.

\subsection{Relativistic models}\label{sec:relmodels}

For spinless particles we also consider the relativistic Hamiltonian
\begin{equation}\label{eq:HR}
H_R=\sum_{i=1}^N\sqrt{p_{A,i}^2+m_i^2}+\alpha V_A .
\end{equation}
As for $H_S$, the field energy is dropped, so $\mathcal E_N\ge\inf_A\inf_\Psi\langle\Psi,H_R\Psi\rangle$. We keep
the rest mass in $H_R$. The abelian literature subtracts it and writes
$\sqrt{p^2+m^2}-m$~\cite{LY88}, \cite[Chap.~8]{LS10}; in that convention the bound $H_R\ge0$ of
Theorem~\ref{thm:rel} reads $H_R-\sum_im_i\ge-\sum_im_i$, a lower bound linear in $N$. None of the
relativistic results depends on the masses, which is why they may vanish here and in the Dirac model
below.

For the Dirac model, with $n_{\rm s}=4$, the
Dirac operator of a particle of mass $m$ is
\[
D_A=\balpha\cdot p_A+\beta m ,
\]
where $\balpha=(\alpha_1,\alpha_2,\alpha_3)$ and $\beta$ are the Dirac matrices in the Dirac
representation. (The bold $\balpha$ is not to be confused with the coupling constant $\alpha$; this
clash of notation is standard~\cite{LSS97}.) For smooth $A$, $D_A$ is essentially self-adjoint on
$C_c^\infty$, since it has finite propagation speed~\cite{Ch73} (see also~\cite[Theorem~4.3]{Th92}). Let $\Lambda_A^+=\1(D_A>0)$ be the
projection onto the positive spectral subspace of $D_A$. The Dirac Hamiltonian
\begin{equation}\label{eq:HD}
H_D=\sum_{i=1}^N(D_A)_i+\alpha V_A+\frac1{8\pi\alpha}\int|B|^2
\end{equation}
acts on the fermionic $N$-particle space 
\[
\mathcal H_N^+=\bigotimes_k\bigwedge^{N_k}\bigl(\Lambda_{A,k}^+L^2(\R^3;\mathfrak h_k)\bigr),
\]
that is, the space $\mathcal H_N$ of Section~\ref{sec:hamiltonians} with each one-particle space
$L^2(\R^3;\mathfrak h_k)$ replaced by the range of $\Lambda_{A,k}^+$, the projection $\Lambda_A^+$ for the mass of
species $k$. The ground state energy is $\mathcal E_N=\inf_A\inf_\Psi\langle\Psi,H_D\Psi\rangle$, with $\Psi$
over normalized states in $\mathcal H_N^+$ that lie in the form domain of $\sum_i|D_A|_i$ (Section~\ref{sec:domains}). Since
$\alpha_j\alpha_k=i\sum_l\varepsilon_{jkl}\Sigma_l$ for $j\ne k$, where $\Sigma_l=\bigl(\begin{smallmatrix}\sigma_l&0\\0&\sigma_l\end{smallmatrix}\bigr)$,
\begin{equation}\label{eq:D2}
D_A^2=p_A^2-i\sum_l\Sigma_l\otimes B_l+m^2 ,
\end{equation}
which has the same Zeeman structure as~\eqref{eq:pauli}. 

The projection $\Lambda_A^+$ is gauge
covariant. The free projection $\Lambda_0^+=\1(D_0>0)$, used in the model of Brown and
Ravenhall~\cite{BR51,EPS96}, is not gauge covariant, already under local abelian gauge transformations
(Remark~\ref{rem:abelian}), and with it the model is unstable for every $\alpha$
(Theorem~\ref{thm:free}).

\subsection{Domains}\label{sec:domains}

This subsection records the functional-analytic setting and can be skipped on a first reading.
On smooth, compactly supported fermionic states, $H_S$, $H_P$ and $H_R$ are symmetric operators. The
quadratic form of $H_S$ is bounded below for every $\alpha$ (Theorem~\ref{thm:spinless}). That of
$H_P$ is bounded below for every $\alpha$ when $\int|B|^2<\infty$ (Proposition~\ref{prop:firstkind});
otherwise the field energy is infinite. That of $H_R$ is bounded below under the hypothesis of
Theorem~\ref{thm:rel}. In these cases we take the Friedrichs extensions~\cite{RS75}. In
Theorems~\ref{thm:relcollapse} and~\ref{thm:groups}(ii), where no lower bound holds, the statements
refer to the quadratic form of $H_R$ on smooth, compactly supported fermionic states. Likewise,
Theorem~\ref{thm:free} refers to the quadratic form of the Dirac Hamiltonian with the free projection,
on fermionic states in the range of $\Lambda_0^+$ with compactly supported Fourier transforms. Non-smooth fields are treated in Appendix~\ref{sec:rough}, where the Hamiltonians are defined by quadratic forms.

For $H_D$ we may again assume $\int|B|^2<\infty$. Under the hypotheses of Theorem~\ref{thm:dirac}, the
quadratic form of $H_D$ on admissible fermionic states, that is, on states in $\mathcal H_N^+$, is bounded above and below, up to additive constants, by
multiples of that of $\sum_i|D_A|_i$ (Section~\ref{sec:rel}). It is therefore closed on the intersection of $\mathcal H_N^+$ with the form domain
of $\sum_i|D_A|_i$, and $H_D$ is the self-adjoint operator on $\mathcal H_N^+$ associated with it~\cite{RS75}.

\subsection{The general setting}\label{sec:general}

Theorem~\ref{thm:groups} extends the stability results beyond the $SU(2)$ model of
Sections~\ref{sec:colour}--\ref{sec:domains}. Let $G$ be a compact Lie group with Lie algebra $\mathfrak g$, and
fix an orthonormal basis $\{t_a\}$ of $\mathfrak g$ for an invariant inner product, written $X\cdot Y$. Each
species $k$ now carries its own unitary representation $\mathcal R_k$ of $G$ on $\C^{n_{{\rm c},k}}$, with Hermitian
generators $\mathcal T^a_k=i\mathcal R_k(t_a)$ (we also write $\mathcal R_k$ for the induced representation of
$\mathfrak g$). Its kinetic energy may be any of $p_A^2/2m$, $(\sigma\cdot p_A)^2/2m$, $\sqrt{p_A^2+m^2}$ or the
Dirac operator, with Dirac particles confined to the positive spectral subspace of $D_A$
(Section~\ref{sec:relmodels}). Its internal space $\mathfrak h_k$ is as in Section~\ref{sec:hamiltonians}, with $n_{{\rm c},k}$
the dimension of $\mathcal R_k$ and $n_{{\rm s},k}=1$, $2$ or $4$ according to its kinetic energy (spinless,
Pauli or Dirac), and the total number of internal states is again
$q=\sum_kn_{{\rm s},k}\,n_{{\rm c},k}\,n_{{\rm f},k}$. The gauge field is $\mathfrak g$-valued. It acts on the colour factor of species
$k$ through $\mathcal R_k$, and on $\mathfrak g$-valued functions by $[A,\cdot\,]$, which defines $K_A$. The
potential is the analogue of~\eqref{eq:VAintro}, with the generators of each particle's representation in
place of the Pauli matrices:
\[
V_A=4\pi\sum_{1\le i<j\le N}e_ie_j\sum_{a,b=1}^{\dim\mathfrak g}K_A^{ab}(x_i,x_j)\,\mathcal T^a_{(i)}\mathcal T^b_{(j)},\qquad
K_A^{ab}(x,y)=t_a\cdot K_A(x,y)\,t_b .
\]
Here $K_A^{ab}$ are the components of $K_A$ in the basis $\{t_a\}$, and $\mathcal T^a_{(i)}$ is the generator $\mathcal T^a_k$ of
the species $k$ of particle $i$, acting on its colour factor. The Hamiltonian is the analogue of~\eqref{eq:Htot}
without $E_\perp$,
\[
H=\sum_{i=1}^NT_{A,i}+\alpha V_A+\frac1{8\pi\alpha}\int_{\R^3}|B|^2,\qquad |B|^2=\sum_lB_l\cdot B_l ,
\]
with $T_{A,i}$ the kinetic energy of particle $i$. When all species have the same kinetic energy, $H_S$, $H_P$, $H_R$
and $H_D$ denote the corresponding Hamiltonians of Sections~\ref{sec:hamiltonians} and~\ref{sec:relmodels}
with this $V_A$; as there, $H_P$ and $H_D$ include the field energy, and $H_S$ and $H_R$ do not.

The two constants of
Theorem~\ref{thm:groups} are
\[
C_{\mathcal R}=\max_k\Bigl\|\sum_a(\mathcal T^a_k)^2\Bigr\|,\qquad
\tau_{\mathcal R}=\max_k\sup_{|\zeta|=1}\Bigl\|\sum_a\zeta_a\mathcal T^a_k\Bigr\|
\]
(see Section~\ref{sec:groups}). In this normalization the physical charge is $e_i=1$. For $SU(2)$ with
the fundamental representation, $\mathcal T^a=\tau^a/2$, so $C_{\mathcal R}=\frac34$ and
$\tau_{\mathcal R}=\frac12$, and a charge $e$ in the $SU(2)$ results corresponds to the charge $2e$ here. By the Cauchy--Schwarz inequality,
$\|\sum_a\zeta_a\mathcal T^a_kv\|\le|\zeta|\,\langle v,\sum_a(\mathcal T^a_k)^2v\rangle^{1/2}$ for
$v\in\C^{n_{{\rm c},k}}$, so
$\tau_{\mathcal R}\le C_{\mathcal R}^{1/2}$. All conditions and constants could therefore be expressed through
$C_{\mathcal R}$ alone, at the cost of weaker bounds; for $SU(2)$, $C_{\mathcal R}^{1/2}=\frac{\sqrt3}2$ against
$\tau_{\mathcal R}=\frac12$.

Theorem~\ref{thm:groups}(iii) uses weights. Fix a maximal commuting subalgebra
$\mathfrak t\subset\mathfrak g$ (a Cartan subalgebra). A unit vector $v\in\C^{n_{{\rm c},k}}$ is a weight vector of
$\mathcal R_k$ if it is an eigenvector of every $\mathcal R_k(s)$, $s\in\mathfrak t$. These matrices are
anti-Hermitian, so $\mathcal R_k(s)v=-i(s\cdot w)\,v$ for all $s\in\mathfrak t$, for a unique $w\in\mathfrak t$, the
weight of $v$.

\section{Main results}\label{sec:results}

We state the results for the gauge group $SU(2)$, except in Theorem~\ref{thm:groups}; Table~\ref{tab:constants} summarizes them. In the
theorems, $\inf_\Psi$ is over normalized states of the model (Sections~\ref{sec:hamiltonians}--\ref{sec:general}),
which are fermionic except in Theorem~\ref{thm:bosons}, and $\inf_A$ over smooth $\su$-valued fields ($\mathfrak g$-valued in Theorem~\ref{thm:groups}); where
stated, the bounds also hold over rough fields. As before, $q$ is the total number of internal states,
$M=\max_im_i$ and $Q=\max_i|e_i|$. Numerical values are rounded in the safe direction: thresholds for stability
down, and constants in lower bounds and thresholds for collapse up.

The proofs give bounds that hold for each fixed field and do not depend on it. By
Sections~\ref{sec:hamiltonians} and~\ref{sec:relmodels}, Theorems~\ref{thm:spinless}, \ref{thm:pauli},
\ref{thm:rel} and~\ref{thm:dirac} therefore bound $\mathcal E_N$ from below, and
Theorems~\ref{thm:collapse} and~\ref{thm:relcollapse} give $\mathcal E_2=-\infty$ (in
Theorem~\ref{thm:relcollapse} because the field energy vanishes at $A=0$).

\begin{table}[!ht]
\centering\footnotesize
\renewcommand{\arraystretch}{1.25}
\setlength{\tabcolsep}{4pt}
\begin{tabular}{llll}
\hline
Thm & Model & Condition & Conclusion\\
\hline
\ref{thm:spinless} & Schr\"odinger & none & $H_S\ge-184\,M\alpha^2Q^4N$\\
\ref{thm:pauli} & Pauli + field energy & $\alpha^2Q^2\le9.00\times10^{-4}$ \ ($\alpha\le0.0600$) & $H_P\ge-656\,M\alpha^2Q^4N$\\
\ref{thm:collapse} & & $\alpha^2|e_1e_2|>9\pi^2$ \ ($\alpha>6\pi$) & collapse ($N=2$)\\
\ref{thm:rel} & $\sqrt{p_A^2+m^2}$, no spin & $\alpha Q^2\le2.248\times10^{-2}$ \ ($\alpha\le0.0899$) & $H_R\ge0$\\
\ref{thm:relcollapse} & & $\alpha e_1e_2>4/(3\pi)$ \ ($\alpha>1.70$) & collapse ($N=2$)\\
\ref{thm:dirac} & Dirac, projection $\Lambda_A^+$ & $10.32\,\alpha\le(1-4983\,\alpha^2Q^4)^{3/2}$ \ ($\alpha\le0.0384$) & $H_D\ge0$\\
\ref{thm:free} & Dirac, projection $\Lambda_0^+$ & every $\alpha>0$ & unstable\\
\ref{thm:bosons} & two bosonic species & every $\alpha>0$ & $\inf H_S\le-c\,m_{\min}\alpha^2e^4N^{7/5}$\\
\hline
\end{tabular}
\caption{Summary of the results for the gauge group $SU(2)$. The numbers are for the smallest $q$ of
each model: $q=2$ without spin, $q=4$ with spin and $q=8$ for the Dirac operator. They follow from the
theorems at these values of $q$; the effect of optimizing them is discussed in Remark~\ref{rem:optimal}. The values
of $\alpha$ in parentheses are for the physical charge $e=\frac12$. The bounds of Theorems~\ref{thm:spinless}
and~\ref{thm:rel} hold for all $A\in L^2_{\rm loc}$, and that of Theorem~\ref{thm:pauli} for all
$A\in L^2_{\rm loc}$ with $\int|B|^2<\infty$ (Propositions~\ref{prop:rough} and~\ref{prop:roughpauli}).}
\label{tab:constants}
\end{table}

\subsection{Non-relativistic models}

\begin{theorem}[Spinless fermions: stability for all couplings]\label{thm:spinless}
For every $N$,
\[
\inf_A\inf_\Psi\langle\Psi,H_S\Psi\rangle\ \ge\ -C_1\,M\alpha^2Q^4\,N,\qquad C_1=104.5\,(8\pi Lq)^{2/3},
\]%
where $L=0.0464$ is the Lieb--Thirring constant of Lemma~\ref{lem:LT}. For $q=2$, $C_1<184$.
The bound also holds with the infimum over $A\in L^2_{\rm loc}$ (Proposition~\ref{prop:rough}).
\end{theorem}

\begin{theorem}[Pauli operator: stability for small coupling]\label{thm:pauli}
Let
\[
\kappa_P=\min\bigl\{0.431\,q^{-4/3},\ 0.279\,q^{-1}\bigr\}.
\]
If $24\pi\,\alpha^2Q^2\le \kappa_P$, then for every $N$,
\[
\inf_A\inf_\Psi\langle\Psi,H_P\Psi\rangle\ \ge\ -C_2\,M\alpha^2Q^4\,N ,\qquad C_2=260\,q^{2/3} .
\]%
For $q=4$ the hypothesis reads $\alpha^2Q^2\le9.00\times10^{-4}$, and the bound is
$-656\,M\alpha^2Q^4N$. For the physical charge $e=\frac12$ this means $\alpha\le0.0600$. The bound
also holds with the infimum over $A\in L^2_{\rm loc}$ with $\int|B|^2<\infty$
(Proposition~\ref{prop:roughpauli}).
\end{theorem}

\begin{theorem}[Pauli operator: collapse for large coupling]\label{thm:collapse}
Consider two particles with $e_1e_2\ne0$, of the same species or of different species. If
$\alpha^2|e_1e_2|>9\pi^2$, then $\inf_A\inf_\Psi\langle\Psi,H_P\Psi\rangle=-\infty$. For
$e_1=e_2=\frac12$ the condition reads $\alpha>6\pi$.
\end{theorem}

Theorem~\ref{thm:pauli} concerns stability of the second kind, a lower bound on $\mathcal E_N$ that is
linear in $N$. Theorem~\ref{thm:collapse} shows that its restriction on $\alpha$ cannot be removed: for large
$\alpha$, already $\mathcal E_2=-\infty$. This collapse comes from minimizing over the field. For a fixed
field there is stability of the first kind, a lower bound for each $N$:

\begin{proposition}[Pauli operator: stability of the first kind]\label{prop:firstkind}
For a fixed smooth field with $\int|B|^2<\infty$, $H_P$ is bounded below for each $N$ and every
$\alpha$.
\end{proposition}

\subsection{Relativistic models and the Dirac operator}

\begin{theorem}[Relativistic spinless fermions]\label{thm:rel}
Let $\ell_0=0.085$. If $3\alpha Q^2\le\ell_0q^{-1/3}$, then $\inf_A\inf_\Psi\langle\Psi,H_R\Psi\rangle\ge0$
for every $N$, also with the infimum over $A\in L^2_{\rm loc}$ (Proposition~\ref{prop:rough}). For $q=2$ and $e=\frac12$ this holds when $\alpha\le0.0899$.
\end{theorem}

\begin{theorem}[Relativistic collapse]\label{thm:relcollapse}
For two particles with $e_1e_2>0$, of the same species or of different species, and
$3\alpha e_1e_2>4/\pi$, $\inf_\Psi\langle\Psi,H_R\Psi\rangle=-\infty$
already at $A=0$. For $e_1=e_2=\frac12$ this means $\alpha>16/(3\pi)\approx1.70$.
\end{theorem}

\begin{theorem}[Dirac operator]\label{thm:dirac}
Let $\lambda=3\alpha Q^2q^{1/3}/\ell_0$, with $\ell_0$ as in Theorem~\ref{thm:rel}, and let
$L_D$ be the constant of Lemma~\ref{lem:daub}, so that $L_D<0.0684$. If $\lambda<1$ and
$6\pi\alpha\,q\,L_D\le(1-\lambda^2)^{3/2}$, then $\inf_A\inf_\Psi\langle\Psi,H_D\Psi\rangle\ge0$ for every $N$.
For $q=8$ and $e=\frac12$ this holds when $\alpha\le0.0384$.
\end{theorem}

\begin{theorem}[The free Dirac projection]\label{thm:free}
Let $H_D^0$ be the Dirac Hamiltonian~\eqref{eq:HD} on
$\bigotimes_k\bigwedge^{N_k}\bigl(\Lambda_{0,k}^+L^2(\R^3;\mathfrak h_k)\bigr)$, that is, on the space $\mathcal H_N^+$ of
Section~\ref{sec:relmodels} with the free projection $\Lambda_{0,k}^+=\1(D_0>0)$ for the mass of species $k$ in
place of $\Lambda_{A,k}^+$,
as in the model of Brown and Ravenhall~\cite{BR51,EPS96}.
\begin{enumerate}[label=\textup{(\alph*)},leftmargin=2em,itemsep=1pt]
\item Without a gauge condition, $\inf_A\inf_\psi\langle\psi,H_D^0\psi\rangle=-\infty$ already for
$N=1$, with $A$ ranging over smooth, compactly supported pure gauges (so $B\equiv0$ and the field
energy vanishes).
\item In Coulomb gauge ($\nabla\cdot A=0$), consider two species with the same charge $e\ne0$ and
$n^3$ particles each. For every $\alpha>0$ there is $n_0$ (of order $\alpha^{-1/2}$ for small $\alpha$) such that
$\inf_A\inf_\Psi\langle\Psi,H_D^0\Psi\rangle=-\infty$ for all $n\ge n_0$.
\end{enumerate}
\end{theorem}

Part (a) uses only the gauge non-covariance of $\Lambda_0^+$, and holds equally in the abelian case
(Remark~\ref{rem:abelian}). Part (b) is the substantive statement: in Coulomb gauge, where the argument of (a) is not available,
the model is still unstable.

\subsection{The general setting and bosons}\label{sec:resgeneral}

\begin{theorem}[General setting]\label{thm:groups}
In the setting of Section~\ref{sec:general}:
\begin{enumerate}[label=\textup{(\roman*)},leftmargin=2em,itemsep=1pt]
\item \emph{Stability of $H$.} For $\alpha$ below an explicit threshold, which depends only on $Q$, $q$,
$C_{\mathcal R}$ and $\tau_{\mathcal R}$,
\[
\inf_A\inf_\Psi\langle\Psi,H\Psi\rangle\ \ge\ -K\,M\alpha^2Q^4N\qquad\text{for every }N,
\]
where $K$ depends only on $q$ and $C_{\mathcal R}$, and $M$ is the largest mass of a non-relativistic
particle ($M=0$ if there is none). If all species are non-relativistic and spinless, no threshold is
needed. If all species have the same kinetic energy, the bounds are those of
Theorems~\ref{thm:spinless}, \ref{thm:pauli}, \ref{thm:rel} and~\ref{thm:dirac}, with $3$ replaced by
$C_{\mathcal R}$ in the electrostatic inequality and $\frac{\sqrt3}2$ by $\sqrt3\,\tau_{\mathcal R}$ in the
Zeeman bound. For mixtures, the threshold and $K$ are given in Section~\ref{sec:groups}.
\item \emph{Collapse of $H_R$.} For two particles carrying the same representation $\mathcal R$, with generators $\mathcal T^a$,
let $-\kappa_{\mathcal R}$ be the smallest eigenvalue of $\sum_a\mathcal T^a\otimes\mathcal T^a$, and suppose that $e_1e_2>0$ and $\alpha e_1e_2\kappa_{\mathcal R}>4/\pi$. Then $\inf_\Psi\langle\Psi,H_R\Psi\rangle=-\infty$ already at $A=0$, for two particles of
different species, and also for two particles of one species if the eigenspace of
$-\kappa_{\mathcal R}$ contains a vector antisymmetric in colour.
\item \emph{Collapse of $H_P$.} Consider two particles carrying the same representation
$\mathcal R$. Suppose that $\mathcal R$ has weight vectors $v_1,v_2$, with weights $w_1,w_2$, such that
$e_1e_2\,w_1\!\cdot w_2<0$, and that there is a unit vector $t\in\mathfrak t$ with
$|t\cdot w_1|=|t\cdot w_2|=\omega>0$. If $\alpha^2|e_1e_2\,w_1\!\cdot w_2|\,\omega^2>9\pi^2/4$, then
$\inf_A\inf_\Psi\langle\Psi,H_P\Psi\rangle=-\infty$ for two particles of different species, and also
for two particles of one species if $t\cdot w_1=-t\cdot w_2$.
\end{enumerate}
\end{theorem}

When all species have the same kinetic energy, the constants in Theorem~\ref{thm:groups}(i) are explicit:
\begin{gather*}
H_S\ge-11.61\,C_{\mathcal R}^2(8\pi Lq)^{2/3}M\alpha^2Q^4N;\\
H_P\ge-28.9\,C_{\mathcal R}^2q^{2/3}M\alpha^2Q^4N\quad\text{if }32\pi\tau_{\mathcal R}^2C_{\mathcal R}\alpha^2Q^2\le\kappa_P;\\
H_R\ge0\quad\text{if }C_{\mathcal R}\alpha Q^2\le\ell_0q^{-1/3};\\
H_D\ge0\quad\text{if }\lambda=C_{\mathcal R}\alpha Q^2q^{1/3}/\ell_0<1\ \text{ and }
24\pi\tau_{\mathcal R}^2\alpha qL_D\le(1-\lambda^2)^{3/2}.
\end{gather*}
For $SU(2)$, with the charge $2e$, these are Theorems~\ref{thm:spinless}, \ref{thm:pauli},
\ref{thm:rel} and~\ref{thm:dirac}. Parts (ii) and (iii) extend Theorems~\ref{thm:relcollapse}
and~\ref{thm:collapse}: they give conditions under which $H_R$ and $H_P$ collapse for large $\alpha$, and
when these conditions hold, the threshold in (i) cannot be removed. Both conditions require two
particles that attract in some colour state. Part (ii) is stated through $\kappa_{\mathcal R}$ and part (iii) through weights, because the mechanisms differ: $H_R$
collapses already at $A=0$, in the most attractive colour channel, whereas the collapse of $H_P$ needs a
magnetic field, which is taken along $t$ so that it acts on weight vectors as an abelian field. Without
attraction, for instance for a $U(1)$ field and charges of one sign, $H_R$ and $H_P$ are nonnegative for
every $\alpha$.

For $SU(n_{\rm c})$ with the fundamental representation and $X\cdot Y=2\tr(X^\dagger Y)$, as in
Section~\ref{sec:colour}, the Fierz identity gives
$\sum_a\mathcal T^a\otimes\mathcal T^a=\frac12P_{\rm swap}-\frac1{2n_{\rm c}}$, where $P_{\rm swap}$ exchanges the two colour
factors. Hence $\kappa_{\mathcal R}=\frac{n_{\rm c}+1}{2n_{\rm c}}$, attained on antisymmetric colour states, and
contracting indices gives $C_{\mathcal R}=\frac{n_{\rm c}^2-1}{2n_{\rm c}}$. For a unit vector $\zeta$, the matrix
$\sum_a\zeta_a\mathcal T^a$ is Hermitian and traceless, and the trace of its square is $\frac12$, so its eigenvalues are at
most $(\frac{n_{\rm c}-1}{2n_{\rm c}})^{1/2}$ in absolute value, with equality for multiples of $\operatorname{diag}(n_{\rm c}-1,-1,\dots,-1)$.
Hence $\tau_{\mathcal R}=(\frac{n_{\rm c}-1}{2n_{\rm c}})^{1/2}$; for $SU(3)$, $C_{\mathcal R}=\frac43$ and $\tau_{\mathcal R}=1/\sqrt3$. In (iii), the first two basis vectors
$v_1,v_2$ and $t=\frac i2\operatorname{diag}(1,-1,0,\dots,0)$ give $t\cdot t=1$ and $t\cdot w_1=-t\cdot w_2=-\frac12$, so $\omega=\frac12$, and
$w_1\!\cdot w_2=-\frac1{2n_{\rm c}}$; for $e_1e_2>0$ the hypotheses of (iii) therefore hold, also for two particles of one
species. For quarks in $SU(3)$, with generators $\lambda^a/2$ given by the Gell-Mann matrices and charge $e=1$,
parts~\textup{(ii)} and~\textup{(iii)} therefore give collapse of $H_R$ for $\alpha>6/\pi\approx1.91$ and of $H_P$
for $\alpha>\sqrt{54}\,\pi\approx23.1$. A quark and an antiquark carry
conjugate representations and attract more strongly: in their colour singlet, $H_R$ collapses already
for $\alpha>3/\pi\approx0.955$, by the argument of part~\textup{(ii)}. Their stability for small
$\alpha$ is covered by part~\textup{(i)}.

The proof, in Section~\ref{sec:groups}, adapts the stability proofs of
Sections~\ref{sec:electro}, \ref{sec:kin} and~\ref{sec:rel}. The electrostatic inequality and the Zeeman
bound hold particle by particle, the one-body bounds are applied species by species, and a
non-relativistic particle is compared with a relativistic one through $p^2/2m\ge c|p|-mc^2/2$. The
group may also be a product, with the couplings of the factors absorbed into the invariant inner
product on $\mathfrak g$.

The results above are all for fermions. For bosons stability fails, as it does for electric charges.

\begin{theorem}[Two bosonic species are unstable]\label{thm:bosons}
Consider $H_S$ at $A=0$ with two bosonic species of $N$ particles each ($2N$ in total), with masses $m_1,m_2$ and
the same charge $e\ne0$. Then $\inf_\Psi\langle\Psi,H_S\Psi\rangle\le-c\,m_{\min}\,\alpha^2e^4N^{7/5}$
for a universal $c>0$, where $m_{\min}=\min(m_1,m_2)$. In particular, no bound $-CN$ holds.
\end{theorem}

\subsection{Remarks on the results}

\begin{remark}[Optimality, and the role of the masses]\label{rem:sharp}
The energy scale $M\alpha^2Q^4$ in the bounds is that of a single colour-singlet pair. Two particles of one species in the colour singlet, with a
symmetric spatial wave function, form a hydrogen-like atom with reduced mass $m/2$ and coupling
$3\alpha e^2$, so their energy is $-\frac94m\alpha^2e^4$. Placing $\lfloor N/2\rfloor$ such pairs far
apart shows that the bounds $-C\,M\alpha^2Q^4N$ have the optimal dependence on $N$, and also on $M$ and
$Q$ when the heaviest particles carry the largest charge. In particular, the bound must grow with the heaviest mass. Unlike nuclei in the abelian theory, particles carrying colour cannot be
taken infinitely heavy.
\end{remark}

\begin{remark}[Static colour-neutral nuclei]\label{rem:nuclei}
Infinitely heavy nuclei are excluded only if they carry colour. Consider static colour-neutral nuclei,
coupled to the particles through an additional $U(1)$ field. Stability of $H_S$ then survives for all couplings, uniformly in both fields and in the positions
of the nuclei. Remark~\ref{rem:nucleiproof}, at the end of Section~\ref{sec:kin}, describes the model and sketches
the proof, which reuses the tools of that section. For $H_P$ and $H_R$ the question is open (Section~\ref{sec:open}).
\end{remark}

\begin{remark}[Thresholds]\label{rem:thresholds}
For the physical charge $e=\frac12$, whose singlet attraction $\frac34\alpha/|x|$ is close to the
abelian $\alpha/|x|$, our thresholds are of the same order as the abelian ones: $\alpha\le0.0600$
(Theorem~\ref{thm:pauli}, $q=4$) against $\alpha<0.06$ for the abelian Pauli operator~\cite{LLS95},
and $\alpha\le0.0899$ (Theorem~\ref{thm:rel}, $q=2$) against $\alpha\mathrel<1/94$ at $A=0$~\cite{LY88} and
$\alpha\le0.038$ for $Z\alpha\le1/\pi$ in a magnetic field~\cite[Theorem~8.1]{LS10}. Since the models
differ, the methods are best compared through the nearest-neighbour bound, where both apply to the
same inequality: Proposition~\ref{prop:relkin} gives the constant $0.085$ for every $A$, about two thirds of the
$0.129$ that Lieb and Yau obtain at $A=0$~\cite[Theorem~5]{LY88}.

As in the abelian case, the stability and collapse conditions are far apart. Those of
Theorems~\ref{thm:pauli} and~\ref{thm:collapse} differ by a factor of about $310$ in $\alpha Q$,
against about $110$ in $\alpha$ for the abelian Pauli operator ($0.06$ and
$6.67$~\cite{LLS95,LL86}). Those of Theorems~\ref{thm:rel} and~\ref{thm:relcollapse} differ by a factor
of about $19$ in $\alpha Q^2$, against about $250$ in $\alpha$ for relativistic matter ($1/94$ and
$128/(15\pi)$~\cite{LY88}).
\end{remark}

\begin{remark}[Optimizing the constants]\label{rem:optimal}
The proofs fix their free parameters at simple values; optimizing them gains only a few per cent.
Larger gains would need better constants in the kinetic-energy bounds for $p_A^2$: a better Lieb--Thirring
constant, such as those known for $A=0$~\cite{Ru11,FHJN21}, would improve Theorem~\ref{thm:spinless}, and a better constant
in the Cwikel--Lieb--Rozenblum bound of Lemma~\ref{lem:count} would improve Theorems~\ref{thm:pauli},
\ref{thm:rel} and~\ref{thm:dirac}.
\end{remark}

\begin{remark}[Bosons]\label{rem:bosons}
For a single species of identical bosons, stability is open. The question has no abelian
counterpart: there, identical charges repel, and a single species is trivially stable. Here the answer
is not obvious either way. Without further internal labels ($q=2$), the attractive singlet channel
of~\eqref{eq:twobody} requires a spatially antisymmetric pair wave function, and so a condensate interacts
only through the repulsive triplet channel. Particles of opposite colour in nearly orthogonal orbitals
do attract, but at a cost in kinetic energy. A proof of stability cannot use the Lieb--Thirring-type
bounds that are known without antisymmetry, since these rely on an interaction that is repulsive in
every channel~\cite{LPS15,LNP16,DN25}.

For two species, the trial states of Theorem~\ref{thm:bosons} are colour product states, not colour
singlets. Since colour singlets are the analogue of neutral systems, it is natural to ask whether the
$N^{7/5}$ behaviour persists among them. We have not examined this.
\end{remark}

\subsection{Discussion}\label{sec:phys}

The model describes fermions with colour charge in the fundamental representation and a classical
colour field. The fermions interact through the field, including the instantaneous non-abelian
Coulomb potential fixed by the Gauss law. With $e=\frac12$, the singlet pair potential is $-C_F\alpha/|x|$. The operator
$(\sigma\cdot p_A)^2/2m$ consists of the terms of order $1/m$ in the heavy-quark Hamiltonian of
non-relativistic quantum chromodynamics. Its colour-magnetic coupling has the leading-order value,
gyromagnetic factor $2$, without the higher-order corrections of quantum chromodynamics~\cite{CaL86,BBL95}.
Because the field is classical, the model does not capture confinement, asymptotic freedom or the
quantum dynamics of gluons. The couplings for which we prove stability are small, as in the abelian theory. For $SU(3)$
quarks with one flavour, Theorem~\ref{thm:groups}(i) gives stability for $\alpha$ up to about $0.018$, $0.029$
and $0.044$ for Dirac, Pauli and relativistic spinless kinetic energy. This is far below the strong
coupling of quantum chromodynamics, which is about $0.12$ at the $Z$ mass and of order one at hadronic
scales~\cite{PDG24}. The collapse thresholds are much larger. In the Pauli model, collapse is shown only
for $\alpha>\sqrt{54}\,\pi\approx23.1$, far above any value of the strong coupling. With relativistic
spinless kinetic energy, however, a quark and an antiquark collapse for $\alpha>3/\pi\approx0.955$, already
at $A=0$ (Section~\ref{sec:resgeneral}), a value of the order of the strong coupling at hadronic scales. This
collapse happens at short distances, where the coupling of quantum chromodynamics becomes small by
asymptotic freedom, which a classical field with a fixed coupling cannot reproduce. It is therefore no
prediction for quantum chromodynamics. Rather, it shows that at physical couplings the classical model
cannot settle the stability of matter made of quarks, and that the quantized theory is needed.

The results show that Fermi statistics and non-abelian
electrostatics alone suffice for the stability of spinless non-relativistic matter interacting with
a classical colour field, even without the field energy. With spin the field energy is needed as well, and with spin or relativistic kinetic
energy the coupling must also be small.

Minimizing over the field is not a formality, because a classical colour field need not simply
follow its sources: it can rearrange itself to lower the energy. For a static classical colour source, the energy with the same covariant
electric term and the magnetic energy was studied by Taubes~\cite{Ta01}. Mandula found that at large
coupling the Coulomb solution is unstable~\cite{Ma77}. For an extended source, Sikivie and Weiss
exhibited solutions of lower energy~\cite{SW78a,SW78b}: colour-magnetic dipole solutions once the
coupling is large enough, and total-screening solutions of arbitrarily low energy. Our bounds hold
uniformly over all such configurations.

Finally, Theorems~\ref{thm:dirac} and~\ref{thm:free} single out the gauge-covariant projection onto the
positive spectral subspace of $D_A$ as the appropriate description of the Dirac sea in a non-abelian field, as in the abelian case~\cite{LSS97,LLo02}. The free projection is neither gauge covariant
nor stable.

\subsection{Open problems}\label{sec:open}

\begin{enumerate}[leftmargin=2em,itemsep=2pt]
\item \emph{Rough fields.} Propositions~\ref{prop:rough} and~\ref{prop:roughpauli} cover the spinless
models for $A\in L^2_{\rm loc}$ and the Pauli model for $A\in L^2_{\rm loc}$ with $\int|B|^2<\infty$. The
Dirac model for rough fields remains open.
\item \emph{Colour-neutral nuclei.} Remarks~\ref{rem:nuclei} and~\ref{rem:nucleiproof} sketch a proof of stability of $H_S$ with static
colour-neutral nuclei. Whether the same holds for $H_P$ and $H_R$, for small couplings, is open. We
expect it to follow by combining the present methods with the abelian arguments
of~\cite{LLS95,LY88}, but we have not checked the details.
\item \emph{Thresholds.} The stability and collapse thresholds of Theorems~\ref{thm:pauli}
and~\ref{thm:collapse}, and of Theorems~\ref{thm:rel} and~\ref{thm:relcollapse}, are far apart
(Remark~\ref{rem:thresholds}). The critical couplings are not known, here or in the abelian case. For
the Pauli operator the collapse is driven by zero modes, so the question is related to how large a
field must be to support one; see~\cite{ES01,FL22}.
\item \emph{Identical bosons.} Theorem~\ref{thm:bosons} uses two species. For a single species of
identical bosons, stability of the second kind is open (Remark~\ref{rem:bosons}).
\item \emph{Quantized field.} The field here is classical. The analogous problem with a quantized
colour field, with an ultraviolet cutoff, is open. Section~\ref{sec:phys} gives a further reason to study
it: at physical couplings, the classical model cannot settle the stability of matter made of quarks. In the abelian case, the proofs for a quantized
field build on bounds that are uniform in a classical field~\cite{BFG96,FFG97,LLo02}, and our bounds
are uniform in $A$ in the same way. The new feature is that a quantized colour field itself carries
colour charge, which enters the Coulomb energy.
\end{enumerate}

\section{Preliminaries: the covariant Laplacian}\label{sec:tools}

Most of what we need about the covariant Laplacian $-\Delta_A=p_A^2$ follows from one fact: its heat
kernel is dominated by that of the ordinary Laplacian. This is the non-abelian version of the diamagnetic inequality~\cite[Theorem~4.4]{LS10}. We record it and its
consequences, followed by the two one-body inequalities from which the kinetic energy estimates are
built.

The symbol $-\Delta_A$ denotes two operators (Section~\ref{sec:colour}). On functions with values in
the internal space, $A$ acts on the colour factor by multiplication (the fundamental representation);
this operator gives the kinetic energy. On $\R^3$-valued functions, $A$ acts by $A\times\cdot$ (the
adjoint representation); this operator gives the Coulomb kernel $K_A$. Unless a lemma says otherwise, its
statement holds for both. In each case $A$ acts on $\C^q$, for some $q$, by anti-Hermitian matrices, and $-\Delta_A$
acts on $L^2(\R^3;\C^q)$. In the adjoint case the space is $\C^3$, the complexification of $\su$.

\begin{lemma}[Diamagnetic inequality]\label{lem:dom}
For all $f\in L^2(\R^3;\C^q)$ and $t\ge0$, $|e^{t\Delta_A}f|\le e^{t\Delta}|f|$ pointwise. For $t>0$ the heat kernel $e^{t\Delta_A}(x,y)$ is smooth, satisfies
$e^{t\Delta_A}(y,x)=e^{t\Delta_A}(x,y)^\dagger$, and
\[
\|e^{t\Delta_A}(x,y)\|\le(4\pi t)^{-3/2}e^{-|x-y|^2/4t},\qquad
\tr_{\C^q}e^{t\Delta_A}(x,x)\le q(4\pi t)^{-3/2}.
\]
The Green function $K_A(x,y)=\int_0^\infty e^{t\Delta_A}(x,y)\,dt$ of $-\Delta_A$ satisfies
$-\Delta_AK_A(\cdot,y)w=\delta_yw$ for $w\in\C^q$, in the sense of distributions. It is smooth off the
diagonal, satisfies $K_A(y,x)=K_A(x,y)^\dagger$, and
\begin{equation}\label{eq:Kbound}
\|K_A(x,y)\|\le\frac1{4\pi|x-y|}.
\end{equation}
In the adjoint representation $K_A$ is real, so $K_A(x,y)^\dagger$ is the transpose.
\end{lemma}

\begin{proof}
On $\R^3$, $-\Delta_A$ is essentially self-adjoint on $C_c^\infty$~\cite[Chap.~XII]{Gu17}. The diamagnetic
inequality is the semigroup form of Kato's inequality~\cite{Ka72,HSU77,Si79}, and holds for any smooth
gauge field acting by anti-Hermitian matrices, on any Riemannian manifold~\cite[Sect.~VII.3]{Gu17}. The
heat kernel is smooth by~\cite[Theorem~II.1]{Gu17}, and its symmetry follows from self-adjointness.
Applied to approximate delta functions at $y$, the diamagnetic inequality gives the bound on
$\|e^{t\Delta_A}(x,y)\|$, and the trace bound follows. Integrating in $t$, with
$\int_0^\infty(4\pi t)^{-3/2}e^{-r^2/4t}\,dt=1/(4\pi r)$, gives $K_A(y,x)=K_A(x,y)^\dagger$
and~\eqref{eq:Kbound}. For the equation $-\Delta_AK_A(\cdot,y)w=\delta_yw$, integrate
$\partial_te^{t\Delta_A}=\Delta_Ae^{t\Delta_A}$ from $0$ to $T$ and let $T\to\infty$; by the Gaussian bound,
$e^{T\Delta_A}(\cdot,y)\to0$ uniformly. Smoothness off the diagonal follows from elliptic regularity. In
the adjoint representation $-\Delta_A$ is a real operator, so $e^{t\Delta_A}$ and $K_A$ are real.
\end{proof}

\begin{lemma}[Positivity of the Coulomb energy]\label{lem:pos}
For a finite $\su$-valued measure $\varrho$ with total variation $|\varrho|$ and
$\iint|\varrho|(dx)\,|\varrho|(dy)/|x-y|<\infty$,
\[
\langle\varrho,K_A\varrho\rangle:=\iint\langle\varrho(dx),K_A(x,y)\varrho(dy)\rangle=\int_0^\infty\|e^{t\Delta_A/2}\varrho\|^2dt\ \ge\ 0.
\]
Here $\|\cdot\|$ is the norm of $L^2(\R^3;\R^3)$.
\end{lemma}

\begin{proof}
By Lemma~\ref{lem:dom}, $e^{t\Delta_A/2}\varrho(x)=\int e^{t\Delta_A/2}(x,y)\varrho(dy)$ satisfies
$|e^{t\Delta_A/2}\varrho|\le e^{t\Delta/2}|\varrho|$. By Tonelli,
$\int_0^\infty\|e^{t\Delta_A/2}\varrho\|^2dt\le\iint|\varrho|(dx)\,|\varrho|(dy)/(4\pi|x-y|)<\infty$. By the
semigroup property and $e^{s\Delta_A}(y,x)=e^{s\Delta_A}(x,y)^\dagger$,
$\|e^{t\Delta_A/2}\varrho\|^2=\iint\langle\varrho(dx),e^{t\Delta_A}(x,y)\varrho(dy)\rangle$. Integrating in
$t$, with Fubini justified by the same bound, gives the claim.
\end{proof}

\begin{lemma}[Maximum principle]\label{lem:maxp}
If $\Delta_Au=0$ on an open set, then $\Delta|u|^2=2|\nabla_Au|^2\ge0$ there. Hence $|u|$ satisfies
the maximum principle, and a bounded function $u$ on $\R^3$ with $\Delta_Au=0$ that tends to $0$ at
infinity vanishes identically.
\end{lemma}

\begin{proof}
By elliptic regularity $u$ is smooth. Since $A$ is anti-Hermitian,
$\partial_j|u|^2=2\operatorname{Re}\langle u,\nabla_{A,j}u\rangle$, and hence
$\Delta|u|^2=2\operatorname{Re}\langle u,\Delta_Au\rangle+2|\nabla_Au|^2$. So $|u|^2$ is subharmonic where $\Delta_Au=0$, and the maximum principle for subharmonic functions
applies to $|u|^2$, hence to $|u|$. For the last claim, apply it on balls of radius $R$ centred at the
origin: $\sup_{|x|<R}|u|\le\sup_{|x|=R}|u|\to0$ as $R\to\infty$.
\end{proof}

The Lieb--Thirring inequality for $p_A^2$ follows from the Cwikel--Lieb--Rozenblum (CLR) bound on the
number of negative eigenvalues~\cite{Cw77,Li80,Ro72}, as in~\cite{LLS95}. Rozenblum and Solomyak
showed that Lieb's proof of this bound, with the same constant, applies to every scalar operator whose semigroup is dominated
by a positive semigroup with suitable kernel bounds, such as $e^{t\Delta}$; this covers the magnetic
Schr\"odinger operator~\cite[Theorem~2.4 and Sect.~3.2]{RS98}. For $p_A^2$ the domination is Lemma~\ref{lem:dom}, and Appendix~\ref{sec:clr} extends
their proof to $\C^q$-valued functions. See~\cite{FLW22} for background. We write $N_{-E}(X)$ for the number of eigenvalues of a
self-adjoint operator $X$ below $-E$, counted with multiplicity. By the min-max principle,
$N_{-E}(X)\le N_{-E}(Y)$ whenever $X\ge Y$ as quadratic forms, that is, when the form domain of $X$ lies
in that of $Y$ and the forms are ordered there.

\begin{lemma}[Cwikel--Lieb--Rozenblum bound]\label{lem:count}
Let $U\ge0$ and $E\ge0$, with $(U-E)_+\in L^{3/2}$. Then
\[
N_{-E}(p_A^2-U)\le q\,L_C\!\int(U-E)_+^{3/2},\qquad L_C=0.116 .
\]
\end{lemma}

\begin{proof}
Since $p_A^2-U+E\ge p_A^2-(U-E)_+$, the min-max principle gives $N_{-E}(p_A^2-U)\le N_0(p_A^2-(U-E)_+)$,
the number of negative eigenvalues of $p_A^2-(U-E)_+$, and Proposition~\ref{prop:clr} bounds this by
$qL_C\int(U-E)_+^{3/2}$. For $A=0$ and $q=1$ this is the CLR bound with Lieb's constant $0.1156$~\cite{Li80},
which $L_C$ rounds up.
\end{proof}

\begin{lemma}[Lieb--Thirring inequality]\label{lem:LT}
For $V\ge0$ and every density matrix $0\le\gamma\le1$ on $L^2(\R^3;\C^q)$,
\begin{equation}\label{eq:LT}
\Tr\bigl((p_A^2-V)\gamma\bigr)\ge-L\,q\int V^{5/2},\qquad
L=\tfrac25L_C=0.0464 .
\end{equation}
\end{lemma}

\begin{proof}
For $0\le\gamma\le1$, $\Tr((p_A^2-V)\gamma)\ge-\Tr(p_A^2-V)_-$, where $-\Tr(p_A^2-V)_-$ is the sum of the
negative eigenvalues, and $\Tr(p_A^2-V)_-=\int_0^\infty N_{-E}(p_A^2-V)\,dE$. Insert
Lemma~\ref{lem:count} and use $\int_0^\infty(V-E)_+^{3/2}dE=\frac25V^{5/2}$.
\end{proof}

For $A=0$ and $q=1$, the best constant in~\eqref{eq:LT} is the Lieb--Thirring constant $L_{1,3}$
of~\cite[Sect.~4.1]{LS10}. It is at least the semiclassical value $1/(15\pi^2)$, so our $L$ exceeds it by
at most a factor $6.9$; better constants are discussed in Remark~\ref{rem:optimal}. Every Lieb--Thirring-type bound in this paper uses the CLR bound of Lemma~\ref{lem:count}, and through it the diamagnetic inequality of Lemma~\ref{lem:dom}.

The next lemma is the hydrogen atom in a ball with Neumann boundary conditions. This is how
we keep track of the kinetic energy near a centre; the idea goes back to Dyson and Lenard.

\begin{lemma}[Hydrogen atom in a Neumann ball; cf.~{\cite[Lemma~2a]{LD68}}]\label{lem:ball}
Let $\mathcal B_R=\mathcal B(y,R)$ be the ball of radius $R$ around $y$, and let $Z>0$. For every real
$u\in H^1(\mathcal B_R)$,
\[
\int_{\mathcal B_R}|\nabla u|^2\ge\int_{\mathcal B_R}\Bigl(\frac Z{|x-y|}-\frac{Z^2}4-\frac{3Z}{2R}\Bigr)u^2 .
\]
Hence, for $a>0$ and $\psi\in L^2(\mathcal B_R;\C^q)$ with $\nabla_A\psi\in L^2(\mathcal B_R)$,
\[
a\int_{\mathcal B_R}|\nabla_A\psi|^2-\int_{\mathcal B_R}\frac{Z|\psi|^2}{|x-y|}\ge
-\Bigl(\frac{Z^2}{4a}+\frac{3Z}{2R}\Bigr)\int_{\mathcal B_R}|\psi|^2 .
\]
\end{lemma}

\begin{proof}
Take $y=0$, let $\phi=e^{-Z|x|/2}$ (so that $-\Delta\phi=(Z/|x|-Z^2/4)\phi$), and write $u=\phi f$.
Integrating by parts on $\mathcal B_R$,
\[
\int|\nabla u|^2=\int\phi^2|\nabla f|^2+\int(Z/|x|-Z^2/4)u^2-\frac Z2\int_{\partial\mathcal B_R}u^2 .
\]
By the divergence theorem,
\[
R\int_{\partial\mathcal B_R}u^2=\int\bigl(3u^2+f^2x\cdot\nabla\phi^2+\phi^2x\cdot\nabla f^2\bigr),
\]
where $x\cdot\nabla\phi^2=-Z|x|\phi^2$ and $|x\cdot\nabla f^2|\le2|x||f||\nabla f|$. Bounding
$(Z/R)|x||f||\nabla f|$ by $|\nabla f|^2+\bigl(Z|x|/(2R)\bigr)^2f^2$ gives
\[
\frac Z2\int_{\partial\mathcal B_R}u^2\le\frac{3Z}{2R}\int u^2+\int\phi^2|\nabla f|^2+\frac{Z^2}{4R^2}\int|x|(|x|-2R)u^2 ,
\]
and the last term is not positive on $\mathcal B_R$. This proves the first claim. For the second, apply
the first to $u=|\psi|$ with $Z/a$ in place of $Z$, multiply by $a$, and use the pointwise diamagnetic
inequality $|\nabla|\psi||\le|\nabla_A\psi|$; the local argument at the start of the proof of
Lemma~\ref{lem:maxmin} gives $|\psi|\in H^1(\mathcal B_R)$ and this inequality.
\end{proof}

\begin{lemma}[Zeeman bound]\label{lem:zeeman}
$(\sigma\cdot p_A)^2\ge0$, and
\[
(\sigma\cdot p_A)^2\ge p_A^2-\tfrac{\sqrt3}2|B| .
\]
The constant $\frac{\sqrt3}2$ in the pointwise bound on the norm of the Zeeman term is sharp.
\end{lemma}

\begin{proof}
Positivity holds because $\sigma\cdot p_A$ is symmetric. By~\eqref{eq:pauli}, the second claim says
that the Zeeman term $\frac12\sum_{l,a}B_l^a\sigma_l\otimes\tau^a$ has norm at most $\frac{\sqrt3}2|B|$.
Write the $3\times3$ matrix $(B_l^a)$ as $\sum_kc_ku_kv_k^{\mathsf T}$ with orthonormal $u_k$, $v_k$ and
singular values $c_k\ge0$. Then the Zeeman term is $\frac12\sum_kc_k(u_k\cdot\sigma)\otimes(v_k\cdot\tau)$,
and each $(u_k\cdot\sigma)\otimes(v_k\cdot\tau)$ has norm $1$, so its norm is at most
$\frac12\sum_kc_k\le\frac{\sqrt3}2\bigl(\sum_kc_k^2\bigr)^{1/2}=\frac{\sqrt3}2|B|$. Equality is attained for
$B_l^a=-c\,\delta_{la}$, $c>0$, on the spin--colour singlet.
\end{proof}

\section{The electrostatic inequality}\label{sec:electro}
Recall from Section~\ref{sec:ideas} how the abelian inequality~\eqref{eq:abelianNN} is proved: each
point charge is smeared over a sphere, Newton's theorem says that this does not change the
interactions, and positivity of the total energy does the rest. The uniform surface charge works because
its potential outside the sphere is that of the point charge. For colour charges the potential
$K_A(\cdot,y)$ is in general not spherically symmetric about $y$, so the surface charge with
the same potential outside the sphere as the point charge is not uniform and has no explicit formula.
We therefore take this property as the definition: we start from the potential
$K_A(\cdot,y)$ of a point charge and ask which surface charge produces it outside the sphere.

\begin{lemma}[A non-abelian Newton theorem]\label{lem:newton}
Fix $y\in\R^3$ and $R>0$, and let $\mathcal B=\mathcal B(y,R)$. For $w\in\R^3$, let
\[
\Phi_w=K_A(\cdot,y)w\,\1_{{\mathcal B}^c}+\varphi_w\,\1_{\mathcal B},
\]
where $ \varphi_w$ is the solution of the interior Dirichlet problem
\[\begin{split}
&\Delta_A\varphi_w=0\qquad\qquad\;\;\text{on}\quad \mathcal B\,,\\
&\varphi_w=K_A(\cdot,y)w\qquad\text{ on} \quad\partial\mathcal B\,.
\end{split}\]
Then:
\begin{enumerate}[label=\textup{(\roman*)},leftmargin=2em,itemsep=1pt]
\item $\varrho_w:=-\Delta_A\Phi_w$, taken in the distributional sense, is a surface charge on the sphere $\partial\mathcal B$ with smooth
density:
\[
\varrho_w=\Bigl(n\cdot\nabla_A\varphi_w-n\cdot\nabla_A(K_A(\cdot,y)w)\Bigr)d\sigma,
\]
    where \(n\) is the outward unit normal and \(d\sigma\) is surface measure on \(\partial\mathcal B\);
\item
$\Phi_w=K_A\varrho_w$ on $\R^3$;
\item $\Phi_w(y)=\mathcal Kw$ for all $w\in\R^3$, where $\mathcal K$ is a symmetric, positive semidefinite $3\times3$
matrix, independent of $w$, with
$\|\mathcal K\|\le1/(4\pi R)$, and $\langle\varrho_w,K_A\varrho_{w'}\rangle=w\cdot\mathcal Kw'$ for all
$w,w'\in\R^3$.
\end{enumerate}
\end{lemma}

In words: $\varrho_w$ is a surface charge on $\partial\mathcal B$ whose potential outside the ball is
that of the point charge $w$ at $y$, and $\mathcal K$ is the matrix of self-energies of these surface
charges. At $A=0$, $\varrho_w$ is the uniform surface charge $w/(4\pi R^2)$, $\Phi_w$ is constant
inside the ball, and $\mathcal K=\1/(4\pi R)$: this is Newton's theorem. For general $A$, $\varrho_w$ need not be uniform or parallel to $w$, and $\mathcal K$ has no explicit
formula. The electrostatic inequality below uses the bound $\|\mathcal K\|\le1/(4\pi R)$ instead.

\begin{proof}
The boundary value problem inside $\mathcal B$ has a unique solution, smooth up to $\partial\mathcal B$. Indeed, let $g$ be a smooth extension of the boundary values to $\overline{\mathcal B}$, and write the
solution as $g+u$, where $u$ vanishes on $\partial\mathcal B$ and $\Delta_Au=-\Delta_Ag$. The diamagnetic
inequality $|\nabla|u||\le|\nabla_Au|$ and Poincar\'e's inequality for $|u|$ bound $\|u\|$ by a multiple of
$\|\nabla_Au\|$, and since $A$ is bounded on $\mathcal B$, $\|\nabla u\|\le\|\nabla_Au\|+\|A\|_\infty\|u\|$. So the
quadratic form $\int_{\mathcal B}|\nabla_Au|^2$ is coercive on $H^1_0(\mathcal B)$, the Lax--Milgram theorem
gives a unique such $u$, and elliptic regularity gives smoothness.

(i) $\Phi_w$ is continuous, since both pieces are smooth up to $\partial\mathcal B$ and agree there. It
satisfies $\Delta_A\Phi_w=0$ on both sides of $\partial\mathcal B$, so Green's formula allows us to identify $\varrho_w$. To be precise, applying the divergence theorem separately on $\mathcal B$ and $\mathcal B^c$ to the
identity
\[
\sum_j\partial_j\bigl(\nabla_{A,j}\Phi_w\cdot\chi-\Phi_w\cdot\nabla_{A,j}\chi\bigr)
=\Delta_A\Phi_w\cdot\chi-\Phi_w\cdot\Delta_A\chi
\]
for an $\R^3$-valued test function $\chi$ gives the stated formula for \(\varrho_w\) 
as the jump of the covariant
normal derivative across the sphere.

(ii) The
functions $K_A\varrho_w$ and $\Phi_w$ are bounded, tend to $0$ at infinity by~\eqref{eq:Kbound}, and
both solve $-\Delta_Au=\varrho_w$. Their difference is smooth and annihilated by $\Delta_A$, hence
zero by Lemma~\ref{lem:maxp}.

(iii) By uniqueness, $\Phi_w$ is linear in $w$, which defines $\mathcal K$. By (i), (ii) and the symmetry of
$K_A$,
$\langle\varrho_w,K_A\varrho_{w'}\rangle=\int_{\partial\mathcal B}\varrho_w(dz)\cdot\Phi_{w'}(z)
=\int_{\partial\mathcal B}\varrho_w(dz)\cdot K_A(z,y)w'=\Phi_w(y)\cdot w'$.
The left side is symmetric in $w,w'$, so $\mathcal K$ is symmetric, and positivity follows from
Lemma~\ref{lem:pos}. For the norm: for a unit vector $w$, the function $\Phi_w$ is annihilated by
$\Delta_A$ inside $\mathcal B$ and equals $K_A(\cdot,y)w$ on $\partial\mathcal B$, so
Lemma~\ref{lem:maxp} and~\eqref{eq:Kbound} give
$|\mathcal Kw|=|\Phi_w(y)|\le\max_{\partial\mathcal B}|K_A(\cdot,y)w|\le1/(4\pi R)$.
\end{proof}

\begin{proposition}[Electrostatic inequality]\label{prop:electro}
For every smooth $A$, every $N\ge2$ and all distinct $x_1,\dots,x_N$, with $d_i=\min_{j\ne i}|x_i-x_j|$,
\[
V_A(\ux)\ \ge\ -3\sum_{i=1}^N\frac{e_i^2}{d_i}\qquad\text{as matrices on }(\C^2_{\rm col})^{\otimes N}.
\]
\end{proposition}

\begin{proof}
Let $\mathcal B_i=\mathcal B(x_i,d_i/2)$, and let $\Phi_{i,w}$, $\varrho_{i,w}$ and $\mathcal K_i$ be as in
Lemma~\ref{lem:newton} for this ball.
For $i\ne j$ we have $|x_i-x_j|\ge\max(d_i,d_j)$, so $x_j\notin\overline{\mathcal B_i}$ and
$\partial\mathcal B_i\cap\mathcal B_j=\emptyset$. Hence, by Lemma~\ref{lem:newton}(i), (ii) and the symmetry of $K_A$,
for $w,w'\in\R^3$,
\begin{align*}
\langle\varrho_{i,w},K_A\varrho_{j,w'}\rangle
&=\int_{\partial\mathcal B_i}\varrho_{i,w}(dz)\cdot\Phi_{j,w'}(z)
=\int_{\partial\mathcal B_i}\varrho_{i,w}(dz)\cdot K_A(z,x_j)w'\\
&=\Phi_{i,w}(x_j)\cdot w'=w\cdot K_A(x_i,x_j)w',
\end{align*}
while $\langle\varrho_{i,w},K_A\varrho_{i,w'}\rangle=w\cdot\mathcal K_iw'$. In other words, the
$3N\times3N$ matrix $\mathcal G$ of mutual energies of the surface charges has the Coulomb kernels
$K_A(x_i,x_j)$ as its off-diagonal blocks and the self-energies $\mathcal K_i$ on the diagonal; we
write $\mathcal G_{IJ}$ for its entries, with $I=(i,a)$ and $J=(j,b)$. Since $\varrho_{i,w}$ is linear in $w$, for all $w_1,\dots,w_N\in\R^3$
\[
\sum_{I,J}\mathcal G_{IJ}\,w_i^a\,w_j^b=\Bigl\langle\sum_i\varrho_{i,w_i},K_A\sum_j\varrho_{j,w_j}\Bigr\rangle\ge0
\]
by Lemma~\ref{lem:pos}. So $\mathcal G$ is positive semidefinite, and $\mathcal G=\mathcal M^{\mathsf T}\mathcal M$ for a
real matrix $\mathcal M$. With $\hat e_I=e_i\tau^a_{(i)}$ for $I=(i,a)$, the operators
$\sum_I\mathcal M_{kI}\hat e_I$ are Hermitian on the colour space $(\C^2_{\rm col})^{\otimes N}$, since $\mathcal M$ is real, and therefore
\[
0\le\sum_k\Bigl(\sum_I\mathcal M_{kI}\hat e_I\Bigr)^2=\sum_{I,J}\mathcal G_{IJ}\hat e_I\hat e_J
=\frac1{2\pi}V_A(\ux)+\sum_ie_i^2\tr\mathcal K_i .
\]
Here the off-diagonal blocks give
$\sum_{i\ne j}e_ie_j\sum_{a,b}K_A^{ab}(x_i,x_j)\tau^a_{(i)}\tau^b_{(j)}=V_A(\ux)/(2\pi)$ by~\eqref{eq:VAintro}, since each
unordered pair appears twice, $K_A(x,y)^{\mathsf T}=K_A(y,x)$, and $\tau_{(i)}$, $\tau_{(j)}$ commute. The
diagonal blocks give $\sum_ie_i^2\tr\mathcal K_i$, since $\sum_{a,b}\mathcal K_i^{ab}\tau^a\tau^b=(\tr\mathcal
K_i)\1$ for the symmetric matrix $\mathcal K_i$. Finally, $\tr\mathcal K_i\le3\|\mathcal K_i\|\le3/(2\pi d_i)$, and the claim follows.
\end{proof}

\begin{remark}\label{rem:electrocheck}
The electrostatic inequality of Proposition~\ref{prop:electro} is sharp up to a factor $2$, as is its
abelian version~\eqref{eq:abelianNN}: for two particles of equal charge in the colour singlet at $A=0$,
\eqref{eq:twobody} gives a Coulomb energy of exactly half the bound, as do two opposite charges in the
abelian case. With fixed nuclei, the electrostatic
inequality of Lieb and Yau~\cite{LY88} uses the Voronoi cells of the nuclei instead, so that each
electron sees only its nearest nucleus; see also~\cite[Chap.~5]{LS10}. Here there are no nuclei and every
pair can attract, so the nearest neighbour among all particles takes the place of the nearest nucleus.
\end{remark}

\section{Kinetic energy inequalities}\label{sec:kin}

Proposition~\ref{prop:electro} has reduced the Coulomb energy to the nearest-neighbour potential
$\sum_id_i^{-1}$. This section bounds the kinetic energy from below by that potential, in two steps. First, splitting
the particles into centres and satellites turns the many-body potential into a one-body one. Second,
a one-body bound is applied to the satellites in the potential of the centres. For $A=0$, such lower bounds on the kinetic energy of fermions were first proved by Fefferman and
de la Llave, for $|p|$ and one spin state~\cite[Corollary~2]{FdlL86}. Lieb and Yau proved them for $q$
spin states, for both $|p|$ and $p^2$~\cite[Theorem~5]{LY88}.

\paragraph{Random splitting.} Let $0<p<1$, and let $\mathcal C\subset\{1,\dots,N\}$ be a random set,
containing each index independently with probability $p$. The indices in $\mathcal C$ are the `centres', the others the
`satellites'. For $\ell\in(0,\infty]$, let
\begin{equation}\label{eq:Wdef}
W_{\mathcal C}^\ell(x)=\frac{\1\bigl[\dist(x,\{x_k\}_{k\in\mathcal C})<\ell\bigr]}{\dist(x,\{x_k\}_{k\in\mathcal C})}
\end{equation}
be the nearest-centre potential, cut off at distance $\ell$. For $j\in\mathcal C$, let
$D_j^{\mathcal C}=\frac12\min_{k\in\mathcal C,\,k\ne j}|x_j-x_k|$ be half the distance from centre $j$ to the
nearest other centre, with $D_j^{\mathcal C}=\infty$ if $|\mathcal C|=1$. We write $\E$ for the expectation
over the random set $\mathcal C$.

\begin{lemma}[Random splitting]\label{lem:split}
For every $i$, every $\ell\in(0,\infty]$ and any one-body operators $T_i$,
\begin{equation}\label{eq:split}
d_i^{-1}\le\frac1{p(1-p)}\,\E\bigl[\1(i\notin\mathcal C)W_{\mathcal C}^\ell(x_i)\bigr]+\ell^{-1},\qquad
\sum_iT_i=\frac1{1-p}\,\E\sum_{i\notin\mathcal C}T_i ,
\end{equation}
and moreover $\E\sum_{j\in\mathcal C}(D_j^{\mathcal C})^{-1}\le2p\sum_jd_j^{-1}$.
\end{lemma}

\begin{proof}
Fix $i$ and let $j$ be a nearest neighbour of $i$. With probability $p(1-p)$, $i$ is a satellite and
$j$ a centre, and then $\dist(x_i,\{x_k\}_{k\in\mathcal C})=d_i$, so that $W_{\mathcal C}^\ell(x_i)=d_i^{-1}$ if
$d_i<\ell$. Since $W_{\mathcal C}^\ell\ge0$ always, this gives the inequality when $d_i<\ell$; when
$d_i\ge\ell$, the term $\ell^{-1}$ suffices. The identity
says that each particle is a satellite with probability $1-p$. For the last claim,
$D_j^{\mathcal C}\ge d_j/2$ and $j\in\mathcal C$ with probability $p$.
\end{proof}

The idea of splitting the particles into two groups and applying a
one-body bound to one group, with the other held fixed, goes back to Lieb and Thirring~\cite[proof of Theorem~1]{LT84}, and Lieb and Yau used it to prove their nearest-neighbour
bound~\cite[Theorem~5]{LY88}. It is sometimes
referred to as the method of L\'evy-Leblond~\cite{LL69}, who in fact used a simpler decomposition:
the Hamiltonian as a sum of terms in each of which one particle is held fixed and the others move in
its field.

Once $\mathcal C$ and the positions of the centres are fixed, the satellites form a fermionic system in an
external potential, and the next lemma lets us apply one-body bounds to them. If $\mathcal C=\emptyset$, then
$W_{\mathcal C}^\ell\equiv0$ and these bounds hold trivially, so we may assume $\mathcal C\ne\emptyset$.

\begin{lemma}[Reduction to a one-body density matrix]\label{lem:slice}
Let $\Psi$ be fermionic and fix $\mathcal C$. For a.e.\ value $\xi$ of the coordinates and internal
indices of the particles in $\mathcal C$, the function $\Psi(\cdot\,;\xi)$ of the remaining variables is antisymmetric
under exchange of any two satellites of the same species. Its one-body density matrix is the direct sum over the species of the one-body density matrices of
their satellites, $\gamma_\xi=\bigoplus_k\gamma^{(k)}_\xi$ on $L^2(\R^3;\C^q)=\bigoplus_kL^2(\R^3;\mathfrak h_k)$. The one for species $k$ has the kernel
\[
\gamma^{(k)}_\xi(x,y)=N'_k\int\Psi(x,\underline x';\xi)\,\Psi(y,\underline x';\xi)^\dagger\,d\underline x' ,
\]
where $N'_k$ is the number of satellites of species $k$, $x$ and $y$ are positions of one of them, whose
internal index is left free, so that $\Psi(x,\underline x';\xi)\in\mathfrak h_k$, and $\underline x'$ collects the positions and
internal indices of the other satellites. It satisfies $0\le\gamma_\xi\le\|\Psi(\cdot\,;\xi)\|^2$. Let $h_\xi$ be a one-body operator on
$L^2(\R^3;\C^q)$ that maps each $L^2(\R^3;\mathfrak h_k)$ into itself and depends measurably on $\xi$
(in our applications only through the positions of the centres), and let $(h_\xi)_i$ denote its action
on particle $i$, which is the same for all particles of a species. Then
\[
\Bigl\langle\Psi,\sum_{i\notin\mathcal C}(h_\xi)_i\Psi\Bigr\rangle=\int\Tr(h_\xi\gamma_\xi)\,d\xi ,
\]
and any lower bound $\Tr(h_\xi\gamma)\ge-F(\xi)$, valid for all $0\le\gamma\le1$ with
$\Tr\gamma=N-|\mathcal C|$ (the number of satellites), gives
$\bigl\langle\Psi,\sum_{i\notin\mathcal C}(h_\xi)_i\Psi\bigr\rangle\ge-\langle\Psi,F\Psi\rangle$.
\end{lemma}

\begin{proof}
Antisymmetry within each species is inherited from $\Psi$; the variables of the centres, including
those of the same species, are fixed and play no role. The bound $\gamma_\xi\le\|\Psi(\cdot;\xi)\|^2$ is the
Pauli principle for each species separately, and it passes to the direct sum. The identity is
Fubini's theorem. Apply the bound to $\gamma=\gamma_\xi/\|\Psi(\cdot;\xi)\|^2$ and integrate.
\end{proof}

The bound $\gamma\le1$ of Lemma~\ref{lem:slice} is the only place where Fermi statistics enters, and
it explains why only the total number $q$ of internal states appears in the results. The kinetic operators $p_A^2$, $(\sigma\cdot p_A)^2$ and
$|p_A|$ act on each species separately, through its colour and spin factors, so Lemma~\ref{lem:slice}
applies to them, also in the general setting of Section~\ref{sec:general}.

\subsection{Without spin}

\begin{proposition}[Nearest-neighbour bound without spin]\label{prop:kin0}
For fermionic $\Psi$ and $\kappa>0$,
\[
\sum_i\langle\Psi,p_{A,i}^2\Psi\rangle-\kappa\Bigl\langle\Psi,\sum_id_i^{-1}\Psi\Bigr\rangle
\ge-c\,(8\pi Lq)^{2/3}\kappa^2N\|\Psi\|^2 ,\qquad c=(5^5/2^4)^{1/3}<5.803 .
\]
\end{proposition}

\begin{proof}
By~\eqref{eq:split} with centre probability $\frac35$,
\[
\sum_ip_{A,i}^2-\kappa\sum_id_i^{-1}\ge\frac52\,\E\sum_{i\notin\mathcal C}\bigl(p_{A,i}^2-\tfrac53\kappa W_{\mathcal C}^\ell(x_i)\bigr)-\frac{\kappa N}\ell .
\]
Fix $\mathcal C$ and the centres. Lemmas~\ref{lem:LT} and~\ref{lem:slice} bound the sum over the satellites
below by $-Lq\int(\frac53\kappa W_{\mathcal C}^\ell)^{5/2}$. Since
$\int(W_{\mathcal C}^\ell)^{5/2}\le\sum_{k\in\mathcal C}\int_{|x-x_k|<\ell}|x-x_k|^{-5/2}=8\pi\ell^{1/2}|\mathcal C|$ and
$\E|\mathcal C|=\frac35N$, the right side is at least
$-N\bigl[\frac52(\frac53)^{3/2}\,8\pi Lq\,\kappa^{5/2}\ell^{1/2}+\kappa/\ell\bigr]$. Optimizing over $\ell$ gives the
claim; the centre probability $\frac35$ minimizes the resulting constant.
\end{proof}

\subsection{With spin}\label{sec:pauli}

The Zeeman term in~\eqref{eq:pauli} can be as negative as $-\frac{\sqrt3}2|B|$, and $|B|$ can be
large where the particles are. Lieb, Loss and Solovej~\cite{LLS95} deal with this by introducing an
energy scale $\mu$ and writing the sum of the negative eigenvalues as
$\int_0^\infty N_{-E}\,dE$.
For $E\le\mu$ the CLR bound (Lemma~\ref{lem:count}) is applied directly, with the Zeeman term
treated as part of the potential; this produces a term $\mu^{1/2}\int|B|^2$ that the field energy
pays for. For $E\ge\mu$, where the deep eigenvalues are counted, a fraction $\eta$ of the kinetic
energy is first spent on the hydrogen bound in a ball around each centre (Lemma~\ref{lem:ball}). The hydrogen
bound removes the Coulomb singularity, at the price of an error $\mu^{-3/2}Z^4\sum_jD_j^{-1}$, where
$D_j$ is half the distance from centre $j$ to the nearest other centre. This error is absorbed into the
Coulomb term in Proposition~\ref{prop:kinP}. A second parameter $\vartheta\in(0,1)$ splits the energy $E$ in the CLR bound between the magnetic and the Coulomb parts of the potential;
see~\eqref{eq:thetasplit}. We fix $\eta=\frac1{20}$ and $\vartheta=\frac25$
(see Remark~\ref{rem:optimal}).

\begin{lemma}[Fixed centres with spin]\label{lem:onebody}
Let $\{R_j\}_{j\in\mathcal C}$ be distinct points, $D_j=\frac12\min_{k\ne j}|R_j-R_k|$ (with $D_j=\infty$ if
$|\mathcal C|=1$), and $V=Z/\dist(x,\{R_j\})$ with $Z>0$. Suppose $\mu\ge Z^2/(\eta(1-\vartheta))=\frac{100}3Z^2$
and $\nu>0$. Then for every density matrix $0\le\gamma\le1$ on $L^2(\R^3;\C^q)$, with $T=(\sigma\cdot p_A)^2$,
\[
\Tr\bigl((T-V)\gamma\bigr)\ge-\nu\Tr\gamma
-q\Bigl(\tfrac{5\pi^2}4L_1Z^3\nu^{-1/2}|\mathcal C|+L_2\mu^{1/2}\!\int|B|^2+L_3\mu^{-3/2}Z^4\sum_{j\in\mathcal C}D_j^{-1}\Bigr),
\]
with the constants $L_1,L_2,L_3$ given in~\eqref{eq:Lconst} after the proof; numerically
${L_1=0.1094}$, ${L_2=0.1238}$ and ${L_3=2.625}$.
\end{lemma}

\begin{proof}
Write $V_\nu=(V-\nu)_+$ and $b=\frac{\sqrt3}2|B|$. Since $V\le V_\nu+\nu$ and
$0\le\gamma\le1$,
\begin{equation}\label{eq:start}
\Tr\bigl((T-V)\gamma\bigr)\ge-\nu\Tr\gamma-\Tr(T-V_\nu)_-,\qquad
\Tr(T-V_\nu)_-=\int_0^\infty N_{-E}(T-V_\nu)\,dE .
\end{equation}
We bound the counting function separately for $E\le\mu$ and $E\ge\mu$. In both cases the magnetic
and the Coulomb parts of the potential are separated by the elementary inequality, a consequence of Jensen's
inequality,
\begin{equation}\label{eq:thetasplit}
(u+v-t)_+^{3/2}\le2^{1/2}\bigl[(u-\vartheta t)_+^{3/2}+(v-(1-\vartheta)t)_+^{3/2}\bigr],
\qquad u,v,t\ge0 .
\end{equation}

\emph{Step 1: low energies, $E\le\mu$.} By Lemma~\ref{lem:zeeman}, $T-V_\nu\ge p_A^2-(b+V_\nu)$.
Lemma~\ref{lem:count} and~\eqref{eq:thetasplit} with $t=E$ give
\[
N_{-E}(T-V_\nu)\le2^{1/2}qL_C\!\int\Bigl[(b-\vartheta E)_+^{3/2}
+\bigl(V_\nu-(1-\vartheta)E\bigr)_+^{3/2}\Bigr].
\]
For $E\le\mu$ the magnetic term is at most ${2^{1/2}qL_C\int(b-\vartheta E^2/\mu)_+^{3/2}}$, since
$\vartheta E\ge\vartheta E^2/\mu$; the Coulomb term is kept as it is.

\emph{Step 2: high energies, $E\ge\mu$.} Since $T\ge0$ and $V_\nu\le V$, we have
$T-V_\nu\ge(\mu/E)T-V$. By Lemma~\ref{lem:zeeman} again,
\[
(\mu/E)T-V\ge(\mu/E)(1-\eta)p_A^2-(\mu/E)b+\bigl[(\mu/E)\eta\,p_A^2-V\bigr].
\]
The open balls $\mathcal B_j=\mathcal B(R_j,D_j)$ are disjoint and each lies in the Voronoi cell of its
centre (the set of points closer to $R_j$ than to any other centre), so $V=Z/|x-R_j|$ on $\mathcal B_j$.
In the bracketed term, keep only the kinetic energy inside the balls and apply Lemma~\ref{lem:ball} with $a=\eta\mu/E$ (on $\mathcal B_R$ with $R\to\infty$ if $|\mathcal C|=1$). This
gives $(\mu/E)\eta\,p_A^2-V\ge-\widetilde W_E$, where
\[
\widetilde W_E=\frac{Z^2E}{4\eta\mu}+\frac{3Z}{2D_j}\ \text{ on }\mathcal B_j,\qquad
\widetilde W_E=V\ \text{ elsewhere}.
\]
Hence $(\mu/E)T-V\ge(\mu/E)(1-\eta)\bigl(p_A^2-U_E\bigr)$ with
$U_E=(1-\eta)^{-1}(b+\frac E\mu\widetilde W_E)$, and Lemma~\ref{lem:count}, applied at the energy
$E^2/(\mu(1-\eta))$ (since $N_{-E}(cX)=N_{-E/c}(X)$ for $c>0$), gives
\[
N_{-E}(T-V_\nu)\le qL_C(1-\eta)^{-3/2}\!\int\Bigl(b+\tfrac E\mu\widetilde W_E
-\tfrac{E^2}{\mu}\Bigr)_+^{3/2}.
\]
Now use~\eqref{eq:thetasplit} with $t=E^2/\mu$ and pull the factor $(E/\mu)^{3/2}$ out of the
second term:
\[
N_{-E}(T-V_\nu)\le2^{1/2}qL_C(1-\eta)^{-3/2}\Bigl[\int\bigl(b-\tfrac{\vartheta E^2}\mu\bigr)_+^{3/2}
+\bigl(\tfrac E\mu\bigr)^{3/2}\!\int\bigl(\widetilde W_E-(1-\vartheta)E\bigr)_+^{3/2}\Bigr].
\]

\emph{Step 3: integration over $E$.} We extend each $E$-integral to $(0,\infty)$, which only
increases it, and use $(1-\eta)^{-3/2}\ge1$ to combine the magnetic terms of both steps. Three
integrals remain, and each is computed by scaling $E$.

(a) \emph{Magnetic term.} Substituting $E=(\mu b/\vartheta)^{1/2}t$,
\[
\int_0^\infty\bigl(b-\tfrac{\vartheta E^2}\mu\bigr)_+^{3/2}dE=\bigl(\tfrac\mu\vartheta\bigr)^{1/2}J\,b^2,
\qquad J=\int_0^1(1-t^2)^{3/2}dt=\tfrac{3\pi}{16} .
\]%
Since $b^2=\frac34|B|^2$, the magnetic terms contribute at most $qL_2\mu^{1/2}\int|B|^2$.

(b) \emph{Coulomb term, low energies.} Substituting $E=vt/(1-\vartheta)$,
\[
\int_0^\infty\bigl(v-(1-\vartheta)E\bigr)_+^{3/2}dE=\frac2{5(1-\vartheta)}\,v^{5/2},
\]
so this term contributes at most $qL_1\int V_\nu^{5/2}$. Since $V_\nu\le\max_j(Z/|x-R_j|-\nu)_+$,
\[
\int V_\nu^{5/2}\le|\mathcal C|\,4\pi\!\int_0^{Z/\nu}\Bigl(\frac Zr-\nu\Bigr)^{5/2}r^2dr
=4\pi\,\mathrm B(\tfrac12,\tfrac72)\,Z^3\nu^{-1/2}|\mathcal C|=\tfrac{5\pi^2}4Z^3\nu^{-1/2}|\mathcal C|,
\]
where $\mathrm B$ is Euler's Beta function.

(c) \emph{Coulomb term, high energies.} On $\mathcal B_j$,
$\widetilde W_E-(1-\vartheta)E=w_j-\eta'E$ with $w_j=\frac{3Z}{2D_j}$ and
$\eta'=1-\vartheta-\frac{Z^2}{4\eta\mu}$, and the hypothesis on $\mu$ gives
$\eta'\ge\frac34(1-\vartheta)$. Substituting $E=wt/\eta'$,
\[
\int_0^\infty\bigl(\tfrac E\mu\bigr)^{3/2}(w-\eta'E)_+^{3/2}dE=\mu^{-3/2}(\eta')^{-5/2}\mathrm B(\tfrac52,\tfrac52)\,w^4
\le\mu^{-3/2}\bigl(\tfrac4{3(1-\vartheta)}\bigr)^{5/2}\mathrm B(\tfrac52,\tfrac52)\,w^4,
\]%
and $\int_{\mathcal B_j}w_j^4=(4\pi/3)D_j^3w_j^4=(27\pi/4)Z^4D_j^{-1}$. Off the
balls $\widetilde W_E=V$, and the same computation with $1-\vartheta$ in place of $\eta'$ gives
the factor $\mu^{-3/2}(1-\vartheta)^{-5/2}\mathrm B(\frac52,\frac52)$ and
$\int_{\text{cell}_j\setminus\mathcal B_j}V^4\le4\pi Z^4\int_{D_j}^\infty r^{-2}dr=4\pi Z^4D_j^{-1}$. Together
this term contributes at most $qL_3\mu^{-3/2}Z^4\sum_jD_j^{-1}$. (If $|\mathcal C|=1$, the ball is all
of $\R^3$ and this term vanishes, since then $\widetilde W_E-(1-\vartheta)E=-\eta'E\le0$.)

Inserting (a)--(c) into~\eqref{eq:start} proves the lemma.
\end{proof}

\paragraph{Explicit constants.} Collecting the factors in (a)--(c), with
$J=\frac{3\pi}{16}$ as in (a), the constants of the lemma are
\begin{equation}\label{eq:Lconst}
\begin{aligned}
L_1&=\frac{2\sqrt2}5\,\frac{L_C}{1-\vartheta},\qquad
L_2=\tfrac34\sqrt2\,L_C(1-\eta)^{-3/2}\vartheta^{-1/2}J,\\
L_3&=\sqrt2\,L_C(1-\eta)^{-3/2}\mathrm B(\tfrac52,\tfrac52)\Bigl(\bigl(\tfrac4{3(1-\vartheta)}\bigr)^{5/2}\tfrac{27\pi}4
+(1-\vartheta)^{-5/2}4\pi\Bigr).
\end{aligned}
\end{equation}

We can now combine the splitting with the one-body bound. The hypothesis $\kappa\le\kappa_P\varepsilon$
below lets the field energy $\varepsilon\int|B|^2$ pay for the $L_2$-term, with a value of $\mu$ allowed in
Lemma~\ref{lem:onebody}, and lets the Coulomb term pay for the $L_3$-term.

\begin{proposition}[Pauli kinetic inequality]\label{prop:kinP}
Let $\kappa_P$ be as in Theorem~\ref{thm:pauli}. If $0<\kappa\le \kappa_P\,\varepsilon$, then for fermionic
$\Psi$,
\[
\Bigl\langle\Psi,\Bigl[\sum_i(\sigma\cdot p_A)^2_i-\kappa\sum_id_i^{-1}+\varepsilon\!\int|B|^2\Bigr]\Psi\Bigr\rangle
\ge-C_P\kappa^2N\|\Psi\|^2,\qquad C_P=14.43\,q^{2/3}.
\]
\end{proposition}

\begin{proof}
Write $T_i=(\sigma\cdot p_A)_i^2$ and $s=\frac54$. The idea is to prove the bound with the
slightly larger coupling $s\kappa$, and to use the excess $(s-1)\kappa\sum_id_i^{-1}$ at
the end to pay for the error from the centres.

\emph{Step 1: splitting.} Let $Z=2s\kappa$. By~\eqref{eq:split} with centre probability $\frac12$ and $\ell=\infty$,
\[
\sum_iT_i-s\kappa\sum_id_i^{-1}+\varepsilon\!\int|B|^2
\ge\E\Bigl[2\sum_{i\notin\mathcal C}\bigl(T_i-ZW_{\mathcal C}^\infty(x_i)\bigr)+\varepsilon\!\int|B|^2\Bigr].
\]

\emph{Step 2: fixed centres.} Fix $\mathcal C$ and the centres, and apply Lemma~\ref{lem:onebody} to the
one-body density matrix of the satellites (Lemma~\ref{lem:slice}), whose trace is $N-|\mathcal C|$.
Choose $\mu=(\varepsilon/(2qL_2))^2$, so that $2qL_2\mu^{1/2}\int|B|^2$ cancels the field energy. With
$Z=2s\kappa$, the hypothesis $\mu\ge Z^2/(\eta(1-\vartheta))$ of the lemma reads
\[
\frac\kappa\varepsilon\le\frac{\sqrt{\eta(1-\vartheta)}}{4qL_2s},
\]
and the right side is at least $0.279\,q^{-1}$. The bracket in Step 1 is then at least
\[
-2\nu(N-|\mathcal C|)-\tfrac{5\pi^2}2qL_1Z^3\nu^{-1/2}|\mathcal C|-2qL_3\mu^{-3/2}Z^4\sum_{j\in\mathcal C}(D_j^{\mathcal C})^{-1}.
\]

\emph{Step 3: averaging.} We have $\E|\mathcal C|=\E(N-|\mathcal C|)=N/2$ and, by
Lemma~\ref{lem:split} with centre probability $\frac12$, $\E\sum_{j\in\mathcal C}(D_j^{\mathcal C})^{-1}\le\sum_jd_j^{-1}$.
With $Z^3=8s^3\kappa^3$ and $Z^4=16s^4\kappa^4$,
\[
\sum_iT_i-s\kappa\sum_id_i^{-1}+\varepsilon\!\int|B|^2\ge-N\bigl[\nu+10\pi^2qL_1s^3\kappa^3\nu^{-1/2}\bigr]
-32qL_3\mu^{-3/2}s^4\kappa^4\sum_id_i^{-1}.
\]

\emph{Step 4: absorption.} Add $(s-1)\kappa\sum_id_i^{-1}$ to both sides. The last term is
absorbed if $32qL_3\mu^{-3/2}s^4\kappa^4\le(s-1)\kappa$, that is, since
$\mu^{3/2}=(\varepsilon/(2qL_2))^3$,
\[
\Bigl(\frac\kappa\varepsilon\Bigr)^3\le\frac{s-1}{32qL_3(2qL_2)^3s^4},
\]
and the cube root of the right side is at least $0.431\,q^{-4/3}$. Both conditions on
$\kappa/\varepsilon$ hold when $\kappa\le\kappa_P\varepsilon$. Finally, $\nu+a\nu^{-1/2}$ with
$a=10\pi^2qL_1s^3\kappa^3$ is minimal at $\nu=(a/2)^{2/3}$, where it equals
$3(a/2)^{2/3}=3(5\pi^2qL_1s^3)^{2/3}\kappa^2\le C_P\kappa^2$.
\end{proof}

\subsection{Proofs of Theorems~\ref{thm:spinless} and~\ref{thm:pauli} and of Proposition~\ref{prop:firstkind}}\label{sec:proofs12}

The proofs of Theorems~\ref{thm:spinless} and~\ref{thm:pauli} are now short. They use that $\alpha V_A\ge-3\alpha Q^2\sum_id_i^{-1}$ by
Proposition~\ref{prop:electro}, and that $\frac1{2m_i}T_i\ge\frac1{2M}T_i$ for both $T_i=p_{A,i}^2$ and $T_i=(\sigma\cdot p_A)_i^2$, since kinetic energies are
nonnegative.

\begin{proof}[Proof of Theorem~\ref{thm:spinless}]
$H_S\ge\frac1{2M}\bigl[\sum_ip_{A,i}^2-\kappa\sum_id_i^{-1}\bigr]$ with $\kappa=6M\alpha Q^2$.
Proposition~\ref{prop:kin0} gives
$H_S\ge-\frac1{2M}\,c\,(8\pi Lq)^{2/3}(6M\alpha Q^2)^2N=-18c\,(8\pi Lq)^{2/3}M\alpha^2Q^4N$, and
$18c<104.5$.
\end{proof}

\begin{proof}[Proof of Theorem~\ref{thm:pauli}]
$H_P\ge\frac1{2M}\bigl[\sum_i(\sigma\cdot p_A)_i^2-\kappa\sum_id_i^{-1}+\varepsilon\int|B|^2\bigr]$ with
$\kappa=6M\alpha Q^2$ and $\varepsilon=M/(4\pi\alpha)$. So $\kappa/\varepsilon=24\pi\alpha^2Q^2$,
and Proposition~\ref{prop:kinP} gives $H_P\ge-18C_P\,M\alpha^2Q^4N$, where
$18C_P\le260\,q^{2/3}$.
\end{proof}

\begin{proof}[Proof of Proposition~\ref{prop:firstkind}]
Let $b=\frac{\sqrt3}2|B|\in L^2$. By~\eqref{eq:pauli} and
Lemma~\ref{lem:zeeman}, $(\sigma\cdot p_A)^2\ge p_A^2-b$. For $E>0$, H\"older's inequality gives
$\|b\1(b>E)\|_{3/2}\le\|b\|_2\,|\{b>E\}|^{1/6}$, which tends to $0$ as $E\to\infty$. By Sobolev's
inequality and the diamagnetic inequality,
$\langle u,b\,u\rangle\le E\|u\|^2+C\|b\1(b>E)\|_{3/2}\|\nabla_Au\|^2$. So $b$ is form-bounded with
respect to $p_A^2$ with relative bound $0$. By~\eqref{eq:Kbound} and the singular-value argument in the proof of Lemma~\ref{lem:zeeman}, which
bounds the norm of $\sum_{a,b}K_A^{ab}(x,y)\,\tau^a\otimes\tau^b$ by $3\|K_A(x,y)\|$, we have
$|V_A|\le3Q^2\sum_{i<j}|x_i-x_j|^{-1}$. By the diamagnetic inequality and Hardy's inequality, $|V_A|$ is
therefore also form-bounded with respect to $\sum_ip_{A,i}^2$, with relative bound $0$. Hence $\sum_i(2m_i)^{-1}(\sigma\cdot p_A)^2_i+\alpha V_A$ is bounded below, with a
bound depending on $A$, $N$ and $\alpha$.
\end{proof}

\begin{remark}[Static colour-neutral nuclei: the model and a sketch of the proof]\label{rem:nucleiproof}
In the model of Remark~\ref{rem:nuclei}, the gauge group is $SU(2)\times U(1)$, with couplings $\alpha$
and $\alpha'$. The particles
are as in Section~\ref{sec:model} and carry in addition $U(1)$ charges $z_i\in\R$, the same within
each species. There are also $K$ static nuclei at $R_1,\dots,R_K$, with $U(1)$ charges $Z_k>0$
and no colour. In the kinetic energy of particle $i$, $\nabla_A$ is replaced by $\nabla+A+iz_ia$,
where $a$ is a real vector potential, so that $H_S$ becomes
\begin{gather*}
H_S=\sum_{i=1}^N\frac{(p_{A,i}+z_ia(x_i))^2}{2m_i}+\alpha V_A+\alpha'W,\\
W=\sum_{i<j}\frac{z_iz_j}{|x_i-x_j|}+\sum_{i,k}\frac{z_iZ_k}{|x_i-R_k|}+\sum_{k<l}\frac{Z_kZ_l}{|R_k-R_l|}.
\end{gather*}
Here $\alpha'W$ is the abelian Coulomb energy of all the charges $z_i$ and $Z_k$, including the repulsion
between the nuclei. Unlike $V_A$, it does not depend on the fields.

For $H_S$, stability survives: uniformly in $A$, in $a$ and in the positions of the nuclei, and
for all $\alpha$ and $\alpha'$. We sketch the argument. Split the kinetic energy into two halves. The first
half, together with $\alpha V_A$, is bounded below by $-2C_1\,M\alpha^2Q^4N$ by the proof of
Theorem~\ref{thm:spinless}. Indeed, Lemma~\ref{lem:dom} holds for the combined connection, and
since the connection now depends on $z_i$, Lemma~\ref{lem:LT} is applied species by species, through
Lemma~\ref{lem:slice}, which changes nothing because the one-body bounds are linear in the number of internal states. In
the second half, drop the kinetic energy of the particles with $z_i\ge0$ and treat them as further
static nuclei. For fixed $\Psi$ and fields, regard $\langle\Psi,W\Psi\rangle$ as a function of independent charges, one
for each particle and each nucleus; since $W$ contains no squares of charges, it is affine in each of them. As in~\cite[Proposition~3.1]{LS10}, it is therefore bounded below
by its minimum over all choices in which each negative charge is $0$ or $-z$ and each positive charge
is $0$ or $Z$, where $z$ and $Z$ are the largest magnitudes of the negative and positive charges.
Particles with charge $0$ drop out of $W$, and their kinetic energy can be dropped. After a factor
$z^2$ is taken out, what remains is the abelian problem of electrons and static nuclei, with coupling
$\alpha'z^2$ and nuclear charges $Z/z$; for fixed positions of the other particles, $\Psi$ is
antisymmetric in the electrons of each species. Its proof in~\cite[Sect.~7.1]{LS10} uses only the Lieb--Thirring inequality and electrostatic inequalities, and
Lemma~\ref{lem:LT}, applied species by species, supplies the former for the combined connection.
This gives a lower bound $-C(N+K)$, with $C$ depending on the masses, the charges, $\alpha'$
and $q$.

For $H_P$ and $H_R$ the same splitting reduces the question to the abelian problem with static
nuclei, but with the kinetic energy $(\sigma\cdot p_A)^2$ or $|p_A|$ of a non-abelian connection;
for $H_R$ the splitting also halves the admissible nuclear coupling. We expect the abelian
arguments~\cite{LLS95,LY88} to carry over with the tools of Section~\ref{sec:tools}, but we have
not checked the details.
\end{remark}

\section{Instability: proofs of Theorems~\ref{thm:collapse} and~\ref{thm:bosons}}\label{sec:collapse}

The collapse of Theorem~\ref{thm:collapse} comes from the zero modes of the Pauli operator found by Loss and Yau~\cite[Sect.~II]{LY86}.
We embed their abelian field in a $U(1)$ subgroup of $SU(2)$, put the two
particles in zero modes, choose their colours so that they attract, and dilate. The kinetic energy is
zero throughout, and a dilation by a length $\ell$ multiplies both the field energy and the
Coulomb energy by $\ell^{-1}$. For large $\alpha$ the Coulomb energy wins.

We first recall the zero mode. Let
\[
\psi_0(x)=(1+|x|^2)^{-3/2}(1+i\sigma\cdot x)\binom10,\qquad a(x)=-\frac{3}{1+|x|^2}\,h(x),\qquad
h=\frac{\psi_0^\dagger\sigma\psi_0}{|\psi_0|^2}.
\]
A direct computation gives $\sigma\cdot(p+a)\psi_0=0$. Moreover $|\psi_0|^2=(1+|x|^2)^{-2}$,
$\|\psi_0\|^2=\pi^2$, and $\curl a=-12(1+|x|^2)^{-2}h$. Here $h$ is a unit vector field (the Hopf
map), since $|\psi^\dagger\sigma\psi|=|\psi|^2$ for every spinor $\psi$, so $|\curl a|^2=144(1+|x|^2)^{-4}$ and $\int|\curl a|^2=18\pi^2$. Let $\rho=|\psi_0|^2/\pi^2$, a probability density, and
\[
I_{\LY}=\iint\frac{\rho(x)\rho(y)}{|x-y|}\,dx\,dy .
\]
By Newton's theorem, $I_{\LY}=\int_0^\infty\bigl(\tfrac2\pi F(r)\bigr)^2r^{-2}dr$, where
$\frac2\pi F(r)$ is the charge inside radius $r$, with $F(r)=\arctan r-\frac r{1+r^2}$ and
$F'=2r^2(1+r^2)^{-2}$. In fact $I_{\LY}=1/\pi$: two integrations by parts give
$\int_0^\infty F^2r^{-2}dr=4\int_0^\infty Fr(1+r^2)^{-2}dr=4\int_0^\infty r^2(1+r^2)^{-3}dr=\pi/4$.
The zero modes lie in the form domain of the Pauli operator, and $C_c^\infty$ is a form core. Since $a$
is bounded, $\|\sigma\cdot(p+a)\psi\|$ and $\|p\psi\|$ differ by at most $\|a\|_\infty\|\psi\|$, so this form
domain is $H^1$. On $H^1$ the form of $\alpha V_A$ is continuous, since $|V_A|\le3Q^2/|x_1-x_2|$ (see the proof of Proposition~\ref{prop:firstkind}) and Hardy's inequality applies. So the trial states below are admissible.

\begin{proof}[Proof of Theorem~\ref{thm:collapse}]
Now take $A_j=ia_j\tau^3$, which is smooth and $\su$-valued; as a vector in $\R^3$, $A_j=(0,0,-2a_j)$.
Since the components of $A$ commute, $B_l=(0,0,-2(\curl a)_l)$ and $|B|^2=4|\curl a|^2$. On colour-up states $p_A$ acts as $p+a$,
and on colour-down states as $p-a$. The time reversal $\tilde\psi_0=\sigma_2\overline{\psi_0}$ of $\psi_0$ is a zero
mode of $\sigma\cdot(p-a)$, since $\sigma_2\overline{\sigma_j}\sigma_2=-\sigma_j$, and
$|\tilde\psi_0|=|\psi_0|$.

If $e_1e_2>0$, put particle~1 in $\pi^{-1}\psi_0\otimes|{\uparrow}\rangle$ and particle~2 in
$\pi^{-1}\tilde\psi_0\otimes|{\downarrow}\rangle$; if the two particles belong to the same species,
take the Slater determinant of these two orbitals, which are orthogonal. If $e_1e_2<0$, the particles have
different charges, hence belong to different species, and we put both in
$\pi^{-1}\psi_0\otimes|{\uparrow}\rangle$. The kinetic energy vanishes. Because only the third components of the $A_j$ are nonzero, the
covariant Laplacian in the adjoint representation is block diagonal, so $K_A^{33}=1/(4\pi|x-y|)$ and
$K_A^{3c}=K_A^{c3}=0$ for $c\ne3$. For the product states, the colour expectations of $\tau^{1,2}$ vanish, and those of $\tau^3$ are
$\pm1$, so $\langle\alpha V_A\rangle=-\alpha|e_1e_2|I_{\LY}$.

For the Slater determinant, let $\phi_1=\pi^{-1}\psi_0\otimes|{\uparrow}\rangle$ and
$\phi_2=\pi^{-1}\tilde\psi_0\otimes|{\downarrow}\rangle$, and let $\gamma(x,y)=\sum_k\phi_k(x)\phi_k(y)^\dagger$ be the
kernel of its one-body density matrix, a $4\times4$ matrix on spin and colour. Then
\[
\begin{aligned}
\langle V_A\rangle=2\pi e_1e_2\iint\sum_{a,b}K_A^{ab}(x,y)\Bigl(&\tr\bigl[\tau^a\gamma(x,x)\bigr]\tr\bigl[\tau^b\gamma(y,y)\bigr]\\
&-\tr\bigl[\tau^a\gamma(x,y)\,\tau^b\gamma(y,x)\bigr]\Bigr)\,dx\,dy .
\end{aligned}
\]
Since $|\tilde\psi_0|=|\psi_0|$, both orbitals have density $\rho$, so $\tr[\tau^a\gamma(x,x)]=\rho(x)\tr\tau^a=0$
and the first term, the direct term, vanishes. In the second, the exchange term,
\[
\tr\bigl[\tau^a\gamma(x,y)\,\tau^b\gamma(y,x)\bigr]=\sum_{k,l}\bigl(\phi_l(x)^\dagger\tau^a\phi_k(x)\bigr)\bigl(\phi_k(y)^\dagger\tau^b\phi_l(y)\bigr).
\]
The terms with $k\ne l$ vanish, because each contains the factor $\phi_1(x)^\dagger\tau^a\phi_2(x)$ or its
complex conjugate, and
$\phi_1(x)^\dagger\tau^a\phi_2(x)=\pi^{-2}\langle{\uparrow}|\tau^a|{\downarrow}\rangle\,\psi_0(x)^\dagger\sigma_2\overline{\psi_0(x)}=0$,
since $v^{\rm T}\sigma_2v=0$ for every $v\in\C^2$, the matrix $\sigma_2$ being antisymmetric. The terms with $k=l$ give $2\delta^{a3}\delta^{b3}\rho(x)\rho(y)$, since
$\langle{\uparrow}|\tau^a|{\uparrow}\rangle=-\langle{\downarrow}|\tau^a|{\downarrow}\rangle=\delta^{a3}$. Hence
$\langle\alpha V_A\rangle=-4\pi\alpha e_1e_2\iint K_A^{33}(x,y)\rho(x)\rho(y)\,dx\,dy=-\alpha e_1e_2I_{\LY}$, as for
the product state.

The field energy is
$\frac1{8\pi\alpha}\int4|\curl a|^2=\frac{9\pi}{\alpha}$. Dilating $\psi\to\ell^{-3/2}\psi(\cdot/\ell)$
and $a\to\ell^{-1}a(\cdot/\ell)$ preserves the zero modes and multiplies both energies by $\ell^{-1}$:
\[
\langle H_P\rangle_\ell=\ell^{-1}\Bigl[\frac{9\pi}{\alpha}-\alpha|e_1e_2|I_{\LY}\Bigr]\to-\infty
\quad(\ell\to0)\ \text{ if }\ \alpha^2|e_1e_2|>\frac{9\pi}{I_{\LY}}=9\pi^2 .
\]
Without the field energy, the bracket would be $-\alpha|e_1e_2|I_{\LY}<0$, and the energy would be
unbounded below for every $\alpha>0$.
\end{proof}

Finally, we turn to bosons. The idea is that colour can imitate the
two-component charged Bose gas, whose $N^{7/5}$ instability is known.

\begin{proof}[Proof of Theorem~\ref{thm:bosons}]
Take trial states in which every particle of the first species is colour-up and every particle of
the second species is colour-down, with a spatial wave function $\Psi$ that is symmetric within each
species. This is admissible for bosons. As in the proof of Theorem~\ref{thm:collapse}, all colour-flip terms have
vanishing expectation: within a species all particles have the same colour, and particles of
different species are distinguishable. Moreover $\langle\tau^3_{(i)}\tau^3_{(j)}\rangle=\pm1$ for equal or opposite
colours. Hence $\langle H_S\rangle$ equals the energy of $\Psi$ in the abelian two-component Bose gas:
charges $\pm e$, pair potential $\pm\alpha e^2/|x|$, masses $m_1,m_2$. Replacing both masses by
$m_{\min}$ only increases the kinetic energy. Dyson~\cite{Dy67} showed that this abelian system has
ground state energy at most $-c\,m_{\min}\,\tilde e^{\,4}N^{7/5}$, where $\tilde e^{\,2}$ is the strength of
the pair potential; here $\tilde e^{\,2}=\alpha e^2$. The sharp asymptotics were later established
in~\cite{LSo04} (lower bound) and~\cite{So06} (upper bound); for unequal masses see~\cite{SS22}.
\end{proof}

\section{Relativistic kinetic energy and the Dirac operator}\label{sec:rel}

For the relativistic models we follow the same route, but two of the tools have to be replaced. The Lieb--Thirring inequality for $p_A^2$ becomes the Daubechies inequality~\cite{Da83} for
$|p_A|$ (Lemma~\ref{lem:daub}), and the hydrogen atom in a ball (Lemma~\ref{lem:ball}) becomes Lieb and Yau's relativistic Hardy inequality in a ball (Lemma~\ref{lem:relhardy}). The second replacement requires more
work. The operator $|p_A|$ is not local, so `keeping the kinetic energy inside the balls' needs a
definition. The restricted form below supplies it, and it is superadditive over disjoint sets
(Lemma~\ref{lem:superadd}). The Dirac operator is then
reduced to $|p_A|$ by the inequality of Birman, Koplienko and Solomyak (Lemma~\ref{lem:bks}), as
in~\cite{LSS97}. The pieces are used as follows. Lemmas~\ref{lem:superadd}, \ref{lem:relhardy}
and~\ref{lem:daub} give the relativistic fixed-centre bound (Lemma~\ref{lem:relonebody}), hence
Proposition~\ref{prop:relkin} and Theorem~\ref{thm:rel}. Lemmas~\ref{lem:bks} and~\ref{lem:daub}
then handle the Dirac operator and prove Theorem~\ref{thm:dirac}. Section~\ref{sec:free}, on the
free projection, is independent of both.

Throughout this section $|p_A|=(p_A^2)^{1/2}$ on $L^2(\R^3;\C^q)$. For a Borel set $\Omega$ and a
real function $v$, define the restricted form
\[
\langle v,|p|_\Omega v\rangle=\frac1{2\pi^2}\iint_{\Omega\times\Omega}\frac{|v(x)-v(y)|^2}{|x-y|^4}\,dx\,dy .
\]
For $\Omega=\R^3$ this is $\langle v,|p|v\rangle$, by the heat kernel formula in the next proof with
$A=0$.

\begin{lemma}[Superadditivity of the restricted form]\label{lem:superadd}
For $u$ in the form domain of $|p_A|$ and disjoint Borel sets $\Omega_j$,
\[
\langle u,|p_A|u\rangle\ \ge\ \sum_j\langle|u|,|p|_{\Omega_j}|u|\rangle .
\]
\end{lemma}

\begin{proof}
Use $E^{1/2}=\frac1{2\sqrt\pi}\int_0^\infty(1-e^{-tE})t^{-3/2}dt$ for $E\ge0$. Since
$\int e^{t\Delta}(x,y)\,dy=1$,
\begin{gather*}
\langle u,(1-e^{t\Delta_A})u\rangle=\iint I_t(x,y)\,dx\,dy,\\
I_t=\tfrac12\bigl[e^{t\Delta}(x,y)(|u(x)|^2+|u(y)|^2)-2\operatorname{Re}\langle u(x),e^{t\Delta_A}(x,y)u(y)\rangle\bigr].
\end{gather*}
By Lemma~\ref{lem:dom}, $I_t\ge\frac12e^{t\Delta}(x,y)(|u(x)|-|u(y)|)^2\ge0$. So we may restrict the
integration to $\bigcup_j\Omega_j\times\Omega_j$. Finally,
$\frac1{2\sqrt\pi}\int_0^\infty t^{-3/2}e^{t\Delta}(x,y)\,dt=\pi^{-2}|x-y|^{-4}$.
\end{proof}

The next lemma, due to Lieb and Yau, is the relativistic Hardy inequality $|p|\ge\frac2\pi|x|^{-1}$
of Kato and Herbst~\cite[Remark~V.5.12]{Ka76}, \cite[Theorem~2.1]{He77}, localized to a ball. It keeps the sharp constant $2/\pi$; the price of localization is an error of
order $1/R$.

\begin{lemma}[Relativistic Hardy inequality in a ball~{\cite[Theorem~7]{LY88}}]\label{lem:relhardy}
Let $\mathcal B=\mathcal B(y,R)$ and let $v\in L^2(\mathcal B)$ be real, with $\langle v,|p|_{\mathcal B}v\rangle<\infty$. Then
\[
\langle v,|p|_{\mathcal B}v\rangle\ge\frac2\pi\int_{\mathcal B}\frac{|v|^2}{|x-y|}
-\frac1R\int_{\mathcal B}\Upsilon\Bigl(\frac{|x-y|}R\Bigr)|v|^2 ,
\]
where, for $0<r<1$,
\[
\Upsilon(r)=\frac2{\pi(1+r)}+\frac{1+3r^2}{\pi(1+r^2)r}\ln(1+r)-\frac{1-r^2}{\pi(1+r^2)r}\ln(1-r)
-\frac{4r}{\pi(1+r^2)}\ln r .
\]
Each term of $\Upsilon$ is positive, and $\int_{|x|<1}\Upsilon(|x|)^4dx\le7.625$.
\end{lemma}

The integral bound is~\cite[(2.27)]{LY88}; we have also checked it numerically.

\begin{lemma}[Birman--Koplienko--Solomyak inequality]\label{lem:bks}
Let $T,S\ge0$ be self-adjoint, and let $S$ be form-bounded with respect to $T$, with relative bound less
than $1$, so that $T-S$ is a self-adjoint operator bounded below with form domain $\mathcal Q(T)$.
Let $X=(T-S)_-$, a bounded operator. Then $S^{1/2}\le(T+X)^{1/2}$. If moreover
$\Tr X^{1/2}<\infty$, then
\[
\Tr(T^{1/2}-S^{1/2})_-\le\Tr X^{1/2},
\]
where $\Tr(T^{1/2}-S^{1/2})_-:=\sup_\gamma\Tr\gamma(S^{1/2}-T^{1/2})$, the supremum being over
finite-rank $0\le\gamma\le1$ with range in $\mathcal Q(T^{1/2})$ (this is the usual trace when the
difference is self-adjoint).
\end{lemma}

\begin{proof}
Since $T-S+X=(T-S)_+\ge0$, we have $S\le T+X$ as forms. The L\"owner--Heinz inequality, which holds for unbounded positive operators~\cite[Satz~2]{He51} (see
also~\cite[Theorem~V.4.12]{Ka76}), gives $S^{1/2}\le(T+X)^{1/2}$. Since $T\le T+X\le T+\|X\|$, it also gives
$T^{1/2}\le(T+X)^{1/2}\le(T+\|X\|)^{1/2}\le T^{1/2}+\|X\|^{1/2}$, so $(T+X)^{1/2}$ has form domain
$\mathcal Q(T^{1/2})$.

Now let $0\le\gamma\le1$ have finite rank, with range in $\mathcal Q(T^{1/2})$. Then
$\Tr\gamma(T^{1/2}-S^{1/2})\ge-\Tr\gamma[(T+X)^{1/2}-T^{1/2}]$, and it remains to bound the latter
trace by $\Tr X^{1/2}$. Let $T_n=\min(T,n)$. By the case of bounded operators, due to Birman, Koplienko and Solomyak~\cite{BKS75} (see
also~\cite[Corollary~2]{An88}), applied to $T_n+X$ and $T_n$ in the trace norm,
$\|(T_n+X)^{1/2}-T_n^{1/2}\|_1\le\Tr X^{1/2}$, so
$\Tr\gamma[(T_n+X)^{1/2}-T_n^{1/2}]\le\Tr X^{1/2}$.

It remains to let $n\to\infty$. Write $\gamma=\sum_k\lambda_k|\varphi_k\rangle\langle\varphi_k|$
(a finite sum) with $\varphi_k\in\mathcal Q(T^{1/2})$. It suffices to show that, for every
$\varphi\in\mathcal Q(T^{1/2})$,
\[
\langle\varphi,T_n^{1/2}\varphi\rangle\to\langle\varphi,T^{1/2}\varphi\rangle
\quad\text{and}\quad
\langle\varphi,(T_n+X)^{1/2}\varphi\rangle\to\langle\varphi,(T+X)^{1/2}\varphi\rangle .
\]
The first holds by monotone convergence in the spectral representation of $T$. For the second,
the forms of $T_n+X$ increase to that of $T+X$, so $T_n+X\to T+X$ in the strong resolvent sense~\cite[Theorem~3.1]{Si78}, \cite[\S VIII.3.4]{Ka76}. Since $(x^{1/2}+i)^{-1}$ is bounded and continuous on $[0,\infty)$, also
$(T_n+X)^{1/2}\to(T+X)^{1/2}$ in the strong resolvent sense. On the other hand, the operators
$(T_n+X)^{1/2}$ increase with $n$, by operator monotonicity, so by the monotone convergence theorem for forms~\cite[Theorem~3.1]{Si78}, \cite[\S VIII.3.4]{Ka76} their forms increase to a closed form whose operator is their strong
resolvent limit, namely $(T+X)^{1/2}$. This gives the second convergence on
$\mathcal Q((T+X)^{1/2})=\mathcal Q(T^{1/2})$. Hence
$\Tr\gamma[(T+X)^{1/2}-T^{1/2}]\le\Tr X^{1/2}$, and taking the supremum over $\gamma$ gives the
claim.
\end{proof}

By spectral truncation of $\gamma$, the bound $\Tr\gamma(S^{1/2}-T^{1/2})\le\Tr X^{1/2}$ extends to every
trace-class $0\le\gamma\le1$ with $\Tr\gamma T^{1/2}<\infty$, the form in which the lemma is used below. The lemma is due in substance to Birman, Koplienko and Solomyak~\cite{BKS75}, and is used in this
form in~\cite{LSS97}. We have included the proof to make its validity for unbounded operators
explicit.

\begin{corollary}[Form domains]\label{cor:domains}
If $\int|B|^2<\infty$, then $\theta|p_A|\le|D_A|+c$ for every $\theta<1$ and some constant
$c=c(A,\theta)$. In particular $\mathcal Q(|D_A|)\subset\mathcal Q(|p_A|)$.
\end{corollary}

\begin{proof}
Apply Lemma~\ref{lem:bks} with $T=D_A^2$ and $S=\theta^2p_A^2$. By~\eqref{eq:D2} and
Lemma~\ref{lem:zeeman}, whose proof applies verbatim to $\Sigma_l$, $T-S\ge(1-\theta^2)p_A^2-\frac{\sqrt3}2|B|$, and
$|B|\in L^2\subset L^{3/2}+L^\infty$ is infinitesimally form-bounded with respect to $p_A^2$ (by the
diamagnetic and Sobolev inequalities). Hence $T\ge(1-\varepsilon)p_A^2-C_\varepsilon$ for every
$\varepsilon>0$, so $S$ is $T$-form-bounded with relative bound $\theta^2/(1-\varepsilon)<1$ for
small $\varepsilon$, and $T-S$ is bounded below. So $X$ is bounded, and
$\theta|p_A|=S^{1/2}\le(D_A^2+\|X\|)^{1/2}\le|D_A|+\|X\|^{1/2}$.
\end{proof}

\begin{lemma}[Daubechies-type bounds; cf.~\cite{Da83}]\label{lem:daub}
For $a>0$ and $U\ge0$, whenever the right-hand sides are finite,
\[
\Tr(a|p_A|-U)_-\le qL_D\,a^{-3}\!\int U^4,\qquad
\Tr(ap_A^2-U)_-^{1/2}\le qL_D\,a^{-3/2}\!\int U^2,
\]
with $L_D=\frac{3\pi}{16}L_C<0.0684$.
\end{lemma}

\begin{proof}
For the second bound, $\Tr Y_-^{1/2}=\frac12\int_0^\infty E^{-1/2}N_{-E}(Y)\,dE$ for self-adjoint $Y$. Apply
Lemma~\ref{lem:count} and the identity $\int_0^\infty E^{-1/2}(U-E)_+^{3/2}dE=\mathrm B(\frac12,\frac52)U^2=\frac{3\pi}8U^2$; scaling gives $\Tr(ap_A^2-U)_-^{1/2}=a^{1/2}\Tr(p_A^2-U/a)_-^{1/2}$. The
first bound follows from the second by Lemma~\ref{lem:bks}, applied with $T=a^2p_A^2$ and $S=U^2$.
Here $U^2$ is infinitesimally form-bounded with respect to $T$ because $U^2\in L^2\subset L^{3/2}+L^\infty$,
and $\Tr(T-S)_-^{1/2}<\infty$ by the second bound.
\end{proof}

Here $L_D$ is a bound on the constant $L_{1/2,3}$ of~\cite[Sect.~4.1]{LS10}, whose semiclassical
value is $1/(32\pi)\approx0.0099$. For $A=0$, Daubechies obtained the constant $0.0258$ in the first
bound~\cite{Da83} (see~\cite[Theorem~8]{LY88}).

We can now prove the relativistic analogue of Lemma~\ref{lem:onebody}. It is simpler, since there
is no Zeeman term, but the structure is the same: a fraction $\pi Z/2$ of the kinetic energy is spent on
Lemma~\ref{lem:relhardy} near each centre, and the rest on the Daubechies bound of Lemma~\ref{lem:daub}.

\begin{lemma}[Relativistic fixed centres]\label{lem:relonebody}
With $\{R_j\}_{j\in\mathcal C}$, $D_j$ and $V=Z/\dist(x,\{R_j\})$ as in Lemma~\ref{lem:onebody}, suppose $0<Z<2/\pi$. Then for every density matrix $0\le\gamma\le1$ on $L^2(\R^3;\C^q)$,
\[
\Tr\bigl((|p_A|-V)\gamma\bigr)\ge-c'q\,Z^4(1-\pi Z/2)^{-3}\sum_{j\in\mathcal C}D_j^{-1},\qquad
c'=L_D\bigl(7.625\,\pi^4/16+4\pi\bigr).
\]
\end{lemma}

\begin{proof}
Let $a=1-\pi Z/2>0$ and $\mathcal B_j=\mathcal B(R_j,D_j)$. By Lemmas~\ref{lem:superadd}
and~\ref{lem:relhardy}, applied to $|\psi|$,
\[
(1-a)\langle\psi,|p_A|\psi\rangle\ge\sum_j\int_{\mathcal B_j}\Bigl(\frac{Z}{|x-R_j|}
-\frac{\pi Z}{2D_j}\,\Upsilon\Bigl(\frac{|x-R_j|}{D_j}\Bigr)\Bigr)|\psi|^2 .
\]
(If $|\mathcal C|=1$, use the relativistic Hardy inequality $|p|\ge\frac2\pi|x|^{-1}$ on $\R^3$ instead, and drop
the $\Upsilon$ term.) Hence $|p_A|-V\ge a|p_A|-U$, where $U=\frac{\pi Z}{2D_j}\Upsilon(|x-R_j|/D_j)$ on $\mathcal B_j$ and $U=V$ off the balls. By
Lemma~\ref{lem:daub}, $\Tr(a|p_A|-U)_-\le qL_Da^{-3}\int U^4$. Finally,
$\int_{\mathcal B_j}U^4\le7.625\,(\pi/2)^4Z^4D_j^{-1}$ and
$\int_{\text{cell}_j\setminus\mathcal B_j}V^4\le4\pi Z^4D_j^{-1}$.
\end{proof}

\begin{proposition}[Relativistic nearest-neighbour bound]\label{prop:relkin}
For fermionic $\Psi$,
\[
\sum_i\langle\Psi,|p_A|_i\Psi\rangle\ \ge\ \ell_0\,q^{-1/3}\Bigl\langle\Psi,\sum_id_i^{-1}\Psi\Bigr\rangle,
\qquad\ell_0=0.085.
\]
\end{proposition}

\begin{proof}
The constant $c'$ of Lemma~\ref{lem:relonebody} satisfies $c'\le4.04$.
Let $0<\kappa<3/(2\pi)$.

By~\eqref{eq:split} with centre probability $\frac34$ and $\ell=\infty$,
$\sum_i|p_A|_i-\kappa\sum_id_i^{-1}\ge4\,\E\sum_{i\notin\mathcal C}\bigl(|p_A|_i-\frac43\kappa W_{\mathcal C}^\infty(x_i)\bigr)$.
Apply Lemma~\ref{lem:relonebody} with $Z=\frac43\kappa$ to the one-body density matrix of the satellites
(Lemma~\ref{lem:slice}), and take the expectation using the last claim of Lemma~\ref{lem:split}.
This gives
\[
\sum_i|p_A|_i\ge\bigl[\kappa-\tfrac{512}{27}\,c'q\,\kappa^4(1-\tfrac{2\pi}3\kappa)^{-3}\bigr]\sum_id_i^{-1}.
\]
The error from the centres is of order $\kappa^4$, so it is absorbed for small $\kappa$. This is the
mechanism of Lieb and Yau~\cite{LY88}, and it is where the smallness of
$\alpha$ in Theorems~\ref{thm:rel} and~\ref{thm:dirac} comes from.

Now choose $\kappa=\kappa_0\,q^{-1/3}$ with $\kappa_0=\frac19$. Since $q\ge1$, we have
$\kappa\le\kappa_0$, and the bracket is at least
$q^{-1/3}\bigl[\kappa_0-\frac{512}{27}c'\kappa_0^4(1-\frac{2\pi}3\kappa_0)^{-3}\bigr]\ge0.0852\,q^{-1/3}$.
\end{proof}

\begin{proof}[Proof of Theorem~\ref{thm:rel}]
We have $\sqrt{p_A^2+m^2}\ge|p_A|$, and $\alpha V_A\ge-3\alpha Q^2\sum_id_i^{-1}$ by
Proposition~\ref{prop:electro}. By Proposition~\ref{prop:relkin},
$H_R\ge(\ell_0q^{-1/3}-3\alpha Q^2)\sum_id_i^{-1}\ge0$.
\end{proof}

\begin{proof}[Proof of Theorem~\ref{thm:relcollapse}]
At $A=0$, put the two particles in the colour singlet, where by~\eqref{eq:twobody}
$\alpha V_0=-3\alpha e_1e_2/|x_1-x_2|$. Take $\Psi=f(x_1-x_2)\,g((x_1+x_2)/2)$ (with the other
internal indices fixed), with $f$ radial and $f,g\in C_c^\infty$ normalized. This is antisymmetric in colour and symmetric in space, so
it is admissible also when both particles belong to one species. Let $p$ and $P$ be the relative and
the total momentum. In Fourier space $|p_1|+|p_2|\le2|p|+|P|$, and $\sqrt{p^2+m^2}\le|p|+m$. Hence
\[
\langle H_R\rangle\le2\Bigl\langle f,\Bigl(|p|-\frac{3\alpha e_1e_2/2}{|x|}\Bigr)f\Bigr\rangle
+\langle g,|P|g\rangle+m_1+m_2 .
\]
If $\frac32\alpha e_1e_2>\frac2\pi$, there is a radial $f\in C_c^\infty$ for which the first term is negative,
since the constant $2/\pi$ in the relativistic Hardy inequality is sharp~\cite[Theorem~2.1]{He77},
already for radial functions. Replacing $f$ by $\ell^{-3/2}f(\cdot/\ell)$ multiplies this term by $\ell^{-1}$,
so it tends to $-\infty$ as $\ell\to0$. The second term can be made arbitrarily small by spreading $g$.
\end{proof}

\paragraph{The form of $H_D$.} We now prove the claim of Section~\ref{sec:domains}. Under the
hypotheses of Theorem~\ref{thm:dirac}, Propositions~\ref{prop:electro} and~\ref{prop:relkin} give
$\alpha V_A\ge-\lambda\sum_i|p_A|_i$ on fermionic states, with $\lambda<1$ as in
Theorem~\ref{thm:dirac}, and Corollary~\ref{cor:domains}, applied with any $\theta\in(\lambda,1)$, gives
$\lambda|p_A|\le(\lambda/\theta)(|D_A|+c)$. Since
$D_A=|D_A|$ on the range of $\Lambda_A^+$, this bounds the quadratic form of $H_D$ below by a positive
multiple of that of $\sum_i|D_A|_i$, up to an additive constant. In the other direction,
$|V_A|\le3Q^2\sum_{i<j}|x_i-x_j|^{-1}$ (see the proof of Proposition~\ref{prop:firstkind}), which is form-bounded with respect to
$\sum_i|p_A|_i$ by Lemma~\ref{lem:superadd} with $\Omega=\R^3$ and the relativistic Hardy inequality, and hence with respect to $\sum_i|D_A|_i$ by
Corollary~\ref{cor:domains}; for details see the proof of Proposition~\ref{prop:rough}.

\begin{proof}[Proof of Theorem~\ref{thm:dirac}]
We may assume $\int|B|^2<\infty$; otherwise the field energy is infinite. By
Corollary~\ref{cor:domains}, states of finite Dirac kinetic energy have finite $\sum_i|p_A|_i$-energy,
so Proposition~\ref{prop:relkin} applies to them. On the range of $\Lambda_A^+$ we have $D_A=|D_A|$.
Since $D_A^2=(\balpha\cdot p_A)^2+m^2$, the L\"owner--Heinz inequality gives
$|D_A|\ge|\balpha\cdot p_A|$. The right side does not depend on the mass, so it is the same one-body operator for every
species. Propositions~\ref{prop:electro} and~\ref{prop:relkin} give
$\alpha V_A\ge-3\alpha Q^2\sum_id_i^{-1}\ge-\lambda\sum_i|p_A|_i$, so
\[
H_D\ge\sum_i(|\balpha\cdot p_A|-\lambda|p_A|)_i+\frac1{8\pi\alpha}\int|B|^2 .
\]
The one-body density matrix of $\Psi$ (as in Lemma~\ref{lem:slice}, with $\mathcal C=\emptyset$) satisfies
$0\le\gamma\le1$ on $L^2(\R^3;\C^q)$, so the first term is at least
$-\Tr(|\balpha\cdot p_A|-\lambda|p_A|)_-$. By Lemma~\ref{lem:bks}, whose hypotheses hold as in the
proof of Corollary~\ref{cor:domains}, and by~\eqref{eq:D2} with $m=0$,
\begin{align*}
\Tr(|\balpha\cdot p_A|-\lambda|p_A|)_-&\le\Tr\bigl((\balpha\cdot p_A)^2-\lambda^2p_A^2\bigr)_-^{1/2}
\le\Tr\bigl((1-\lambda^2)p_A^2-\tfrac{\sqrt3}2|B|\bigr)_-^{1/2}\\
&\le\tfrac34\,qL_D(1-\lambda^2)^{-3/2}\!\int|B|^2,
\end{align*}
using Lemma~\ref{lem:zeeman} for $\Sigma_l$ and Lemma~\ref{lem:daub}.
Comparing with the field energy proves the claim.
\end{proof}

\subsection{The free Dirac projection}\label{sec:free}

The choice of projection matters for two reasons, one for each part of Theorem~\ref{thm:free}.
Part (a) is a one-particle statement. A pure gauge $A=U^*dU$ has no field energy, and the unitary $U$ conjugates $D_A$ to $D_0$
but does not preserve the free positive subspace. So one can lower the energy of a free
positive-energy state without limit. Part (b) is the many-particle mechanism of~\cite{LSS97} (see
also~\cite{GT99}). Free
positive-energy states carry a current that couples linearly to the gauge field, while a gauge field that
is constant over the system costs a field energy that grows only like the diameter of the system, of order
$N^{1/3}$. Optimizing the amplitude of the field gains an energy of order $N^{5/3}$, which beats the kinetic energy of order $N$. Colour
plays two roles. Colour-up and colour-down particles see opposite fields, so with time-reversed
orbitals for the second species the two currents add rather than cancel. The two colour charge
densities, by contrast, cancel, so the Coulomb energy is negligible.

Two facts are used. First, for $\psi$ in the range of $\Lambda_0^+$,
$\langle\psi,\alpha_j\psi\rangle=\langle\psi,(p_j/E(p))\psi\rangle$ with $E(p)=\sqrt{p^2+m^2}$,
because $\Lambda_0^+(p)\alpha_j\Lambda_0^+(p)=(p_j/E(p))\Lambda_0^+(p)$ pointwise in momentum space~\cite[Sect.~1.6]{Th92}. Second, time reversal $\Theta\psi=i\alpha_1\alpha_3\overline\psi$ (in the Dirac
representation) is antiunitary and commutes with $D_0$. Hence
$\Theta$ preserves the range of $\Lambda_0^+$, preserves densities $\psi^\dagger\psi$, and reverses
currents $\psi^\dagger\balpha\psi$. We write $\zhat$ for the unit vector in the $x_3$-direction.

\begin{proof}[Proof of Theorem~\ref{thm:free}(a)]
Let $f$ be smooth with compact support, $U=\operatorname{diag}(e^{if},e^{-if})\in SU(2)$
and $A=U^*dU=i\tau^3df$. Then $F=0$, and $D_A=U^*D_0U$. Take $\psi=\phi\otimes|{\uparrow}\rangle$
with $\phi\in\Lambda_0^+L^2(\R^3;\C^4)$, normalized, and $c:=\langle\phi,(p_3/E)\phi\rangle>0$
(take $\hat\phi$ with compact support in $\{p_3>0\}$). Then
\[
\langle\psi,D_A\psi\rangle=\langle\phi,D_0\phi\rangle+\langle\phi,(\balpha\cdot\nabla f)\phi\rangle .
\]
Now take $f=-t\,x_3\chi(x/R)$ with $t>0$, where $\chi\in C_c^\infty(\R^3)$ satisfies $\chi=1$ on $\mathcal B(0,1)$ and $\operatorname{supp}\chi\subset\mathcal B(0,2)$. Then $\nabla(x_3\chi(x/R))\to\zhat$ pointwise
as $R\to\infty$, with $|\nabla(x_3\chi(x/R))|\le\|\chi\|_\infty+2\|\nabla\chi\|_\infty$ for all $R$, so by dominated convergence
$\langle\phi,(\balpha\cdot\nabla(x_3\chi(x/R)))\phi\rangle\to\langle\phi,\alpha_3\phi\rangle=c$. Fix $R$
with this quantity $\ge c/2$. Then
$\langle\psi,D_{A}\psi\rangle\le\langle\phi,D_0\phi\rangle-tc/2\to-\infty$ as $t\to\infty$.
\end{proof}

\begin{proof}[Proof of Theorem~\ref{thm:free}(b)]
The two species have masses $m_1,m_2$ and $N=n^3$ particles each. Three parameters are chosen: first
$N$ is fixed, large enough (Step~6); then, for each lattice spacing $\ell$ of the orbitals, the
amplitude $t$ of the field is optimized; finally $\ell$ is sent to $0$, the ultrarelativistic limit in which the
masses become negligible. We also fix $k_0=20$, and write $\Lambda^+_{0,m}$ for the free projection for mass $m$.

\emph{Step 1: orbitals.} Let $\varphi\ge0$ be smooth with $\int\varphi=1$ and
$\operatorname{supp}\varphi\subset\mathcal B(0,\frac1{10})$, and let $f\ge0$ with
$f^2=(2\pi)^{-3}\1_{[-\pi,\pi]^3}*\varphi$. Then $\sum_{\nu\in\mathbb Z^3}f(\zeta+2\pi\nu)^2\equiv(2\pi)^{-3}$
and $\int f^2=1$. Fix a unit vector $u_0$ with $\alpha_3u_0=u_0$, and for $m\in\{m_1,m_2\}$ put
$u_m(p)=\Lambda^+_{0,m}(p)u_0/|\Lambda^+_{0,m}(p)u_0|$. This is well defined for $p_3\ge0$, since
$|\Lambda^+_{0,m}(p)u_0|^2=\frac12(1+p_3/E_m(p))$. Let $\{y_k\}$ be the $n^3$ points of
$\ell\mathbb Z^3$ (for odd $n$) or of $\ell(\mathbb Z^3+\frac12(1,1,1))$ (for even $n$) in a cubic block of
side $\ell(n-1)$ centred at $0$, so that
$|y_k|\le\frac{\sqrt3}2\ell N^{1/3}$, and set
\[
\hat\phi^m_k(p)=\ell^{3/2}f(\ell p-k_0\zhat)\,e^{-ip\cdot y_k}u_m(p).
\]
Then $\phi^m_k$ lies in the range of $\Lambda^+_{0,m}$, and the periodization identity gives
$\langle\phi^m_k,\phi^m_l\rangle=\delta_{kl}$. On the support, $|p|\le(k_0+6)/\ell$ and
$p_3\ge(k_0-4)/\ell$.

\emph{Step 2: the massless limit.} Substituting $p=(\zeta+k_0\zhat)/\ell$ gives
$|\phi^m_k(x)|=\ell^{-3/2}|\eta_{m,\ell}((x-y_k)/\ell)|$, where $\eta_{m,\ell}=\mathcal F^{-1}[f\,\tilde u_{m,\ell}]$
and $\tilde u_{m,\ell}(\zeta)=u_m((\zeta+k_0\zhat)/\ell)$. As $\ell\to0$, $\tilde u_{m,\ell}\to\tilde u$
uniformly on $\operatorname{supp}f$, with error $O(m\ell)$; here $\tilde u$ is the massless spinor,
the same for both masses. Since $f$ has compact support, Plancherel gives $\eta_{m,\ell}\to\eta_0$ in $H^{1/2}(\R^3)$.

\emph{Step 3: trial state and kinetic energy.} The trial state is
\[
\Psi=\bigwedge_{k=1}^N\bigl(\phi^{m_1}_k\otimes|{\uparrow}\rangle\bigr)\otimes\bigwedge_{k=1}^N\bigl(\Theta\phi^{m_2}_k\otimes|{\downarrow}\rangle\bigr),
\]
where $\bigwedge$ denotes the normalized antisymmetric product, that is, the Slater determinant:
species~1 occupies the orbitals $\phi^{m_1}_k$ with colour up, and species~2 the time-reversed orbitals
$\Theta\phi^{m_2}_k$ with colour down.
Time reversal preserves the range of $\Lambda^+_{0,m_2}$, densities and kinetic energy, and reverses
currents. Hence the densities are $\rho_1=\sum_k|\phi^{m_1}_k|^2$ and $\rho_2=\sum_k|\phi^{m_2}_k|^2$,
and the currents are $j_1$ and $-j_2$, where $j_1=\sum_k(\phi^{m_1}_k)^\dagger\balpha\,\phi^{m_1}_k$
and $j_2$ likewise. Each orbital has
$\langle D_0\rangle\le(k_0+6)/\ell+m$. For $m\ell\le1$ it also has
$\langle\phi^m_k,\alpha_3\phi^m_k\rangle=\int|\hat\phi^m_k|^2p_3/E_m\ge\frac{k_0-4}{k_0+7}>\frac12$.

\emph{Step 4: magnetic energy.} Let $R=\sqrt3\,\ell N^{1/3}$ and
$w=\curl\bigl(\chi(x/R)\tfrac12\zhat\times x\bigr)$, with $\chi$ as in the proof of (a). Then $\nabla\cdot w=0$, $w=\zhat$ on
$\mathcal B(0,R)$, $|w|\le c_\chi$, and $\int|\curl w|^2=RY_1$, where $c_\chi$ and $Y_1$ depend only on
$\chi$. Take $A=ia\tau^3$ with $a=-tw$, $t>0$;
this is in Coulomb gauge. Colour-up sees $p+a$ and colour-down sees $p-a$, so the magnetic terms of
the kinetic energy are
\[
\int a\cdot j_1-\int a\cdot(-j_2)=-t(X_1+X_2),\qquad X_1=\int w\cdot j_1,\quad X_2=\int w\cdot j_2 .
\]
Since $|w-\zhat|\le(1+c_\chi)\1_{|x|>R}$, $|\psi^\dagger\balpha\psi|\le|\psi|^2$ pointwise, and $|x|>R$ implies $|x-y_k|>\frac{\sqrt3}2\ell N^{1/3}$,
\begin{gather*}
X_1\ge N\bigl[\tfrac12-(1+c_\chi)\,\epsilon_{m_1}(\ell)\bigr]\quad\text{and likewise }X_2,\\
\epsilon_m(\ell)=\int_{|z|>\sqrt3N^{1/3}/2}|\eta_{m,\ell}|^2\ \xrightarrow[\ell\to0]{}\
\epsilon_N:=\int_{|z|>\sqrt3N^{1/3}/2}|\eta_0|^2 .
\end{gather*}
The field energy is $\frac1{8\pi\alpha}\int4|\curl a|^2=t^2RY_1/(2\pi\alpha)$.

\emph{Step 5: Coulomb energy.} As in Section~\ref{sec:collapse}, $K_A^{33}=1/(4\pi|x-y|)$ and every
colour-flip term has vanishing expectation: within a species all orbitals have the same colour, and
between the species there are no exchange terms, since particles of different species are
distinguishable. So the Coulomb energy is that of an abelian system with
charges $\pm e$. For Slater determinants, exchange terms only lower it, so
$\langle\alpha V_A\rangle\le\alpha e^2D(\rho_1-\rho_2)$, where
$D(h)=\frac12\iint h(x)h(y)|x-y|^{-1}dx\,dy$. By scaling,
$D(\rho_1-\rho_2)=\ell^{-1}D(\tilde\rho_{1,\ell}-\tilde\rho_{2,\ell})$, with
$\tilde\rho_{1,\ell}(z)=\sum_k|\eta_{m_1,\ell}(z-y_k/\ell)|^2$ and $\tilde\rho_{2,\ell}$ likewise. By
Step~2 (and $H^{1/2}\subset L^3$), both converge to the same limit in
$L^1\cap L^{3/2}\subset L^{6/5}$. Hence $D(\tilde\rho_{1,\ell}-\tilde\rho_{2,\ell})\to0$ by the Hardy--Littlewood--Sobolev inequality~\cite{LL01}.

\emph{Step 6: conclusion.} Choose $N$ so large that $\epsilon_N\le1/(8(1+c_\chi))$ and
$N^{2/3}>16\sqrt3(k_0+6)Y_1/(\pi\alpha)$. For $\ell$ small, $\epsilon_{m_1}(\ell)$ and $\epsilon_{m_2}(\ell)$ are at most
$1/(4(1+c_\chi))$, so $X_1,X_2\ge N/4$ and $X_1+X_2\ge N/2$. The energy of the trial state is at most the
kinetic energy $2(k_0+6)N/\ell+N(m_1+m_2)$ of Step~3, plus $-t(X_1+X_2)+t^2RY_1/(2\pi\alpha)$ from
Step~4, plus the Coulomb energy of Step~5. The minimum over $t$ of $-tN/2+t^2RY_1/(2\pi\alpha)$ is
$-\pi\alpha N^2/(8RY_1)=-\pi\alpha N^{5/3}/(8\sqrt3\,\ell Y_1)$. Hence
\[
\langle H_D^0\rangle\le\frac1\ell\Bigl[2(k_0+6)N-\frac{\pi\alpha}{8\sqrt3\,Y_1}N^{5/3}
+\alpha e^2D(\tilde\rho_{1,\ell}-\tilde\rho_{2,\ell})\Bigr]+N(m_1+m_2).
\]
By Step~5 and the choice of $N$, the bracket tends to a negative number as $\ell\to0$, so
$\langle H_D^0\rangle\to-\infty$.
\end{proof}

\begin{remark}\label{rem:abelian}
Part (a) holds equally in the abelian case. There it is excluded by the Coulomb gauge used
in~\cite{LSS97}, where the instability of the free projection involves many particles, as in
part (b). In the non-abelian case the Coulomb gauge does not remove the ambiguity: it admits Gribov
copies~\cite{Gr78,Sn78}, gauge-equivalent fields that all satisfy the gauge condition, so even within
the Coulomb gauge the free-projection model may depend on the choice of representative. The projection
$\Lambda_A^+$ of Theorem~\ref{thm:dirac} is gauge covariant and has no such ambiguity.
\end{remark}

\section{The general setting: proof of Theorem~\ref{thm:groups}}\label{sec:groups}

We prove Theorem~\ref{thm:groups} in steps. For one representation and one kinetic energy,
nothing in the stability proofs is specific to $SU(2)$ except two constants: the factor $3$ in the
electrostatic inequality, which came from $\sum_{ab}\mathcal K^{ab}\tau^a\tau^b=(\tr\mathcal K)\1$, and the
constant $\frac{\sqrt3}2$ in the Zeeman bound. We identify their replacements first, then admit
several representations, mixtures of kinetic energies and product groups, and finally prove the two
collapse results.

We use the setting of Section~\ref{sec:general}, with the orthonormal basis $\{t_a\}$ chosen so that
its first elements span the Cartan subalgebra $\mathfrak t$ and the remaining ones span $\mathfrak t^\perp$. We first take all particles in one unitary representation
$\mathcal R$ of $G$ on $\C^n$, with generators $\mathcal T^a=i\mathcal R(t_a)$; the maxima in the definitions of
$C_{\mathcal R}$ and $\tau_{\mathcal R}$ are then over this representation alone. The derivation of
Appendix~\ref{app:derivation} gives $e_i=1$. A weight vector $v$ with weight $w$ satisfies $\mathcal T^av=w^av$ for $t_a\in\mathfrak t$, where
$w^a=t_a\cdot w$ are the components of $w$. For $t_a\in\mathfrak t^\perp$ one has $\langle v,\mathcal T^av\rangle=0$, because
such $\mathcal T^a$ map a weight vector to a combination of weight vectors with other weights.

\begin{proof}[Proof of Theorem~\ref{thm:groups}]
\emph{(i) Stability: one representation, one kinetic energy.} Lemmas~\ref{lem:dom}, \ref{lem:pos}, \ref{lem:maxp} and~\ref{lem:newton} hold
for any $\mathfrak g$-valued $A$ (with $\R^3$ replaced by $\mathfrak g$), in both representations, since $\mathcal R(A_j)$ and $[A_j,\cdot\,]$
are anti-Hermitian. In Proposition~\ref{prop:electro}, the only $SU(2)$-specific step is
$\sum_{ab}\mathcal K^{ab}\tau^a\tau^b=(\tr\mathcal K)\1$. Since $\|\mathcal K\|\1-\mathcal K$ is symmetric and
positive semidefinite, it can be written as $\mathcal M^{\mathsf T}\mathcal M$, and then
$\sum_{ab}(\|\mathcal K\|\delta^{ab}-\mathcal K^{ab})\mathcal T^a\mathcal T^b=\sum_k(\sum_a\mathcal M_{ka}\mathcal T^a)^2\ge0$. Hence
$\sum_{ab}\mathcal K^{ab}\mathcal T^a\mathcal T^b\le\|\mathcal K\|C_{\mathcal R}$, and the proof gives
$V_A\ge-C_{\mathcal R}\sum_ie_i^2/d_i$. The Zeeman term is $-\sum_{l,a}B_l^a\sigma_l\otimes\mathcal T^a$, of
norm at most $\tau_{\mathcal R}\sum_l|B_l|\le\sqrt3\,\tau_{\mathcal R}|B|$; the same holds with
$\Sigma_l$ in place of $\sigma_l$. All kinetic estimates use only the total number $q$ of internal
states, now counted as in Section~\ref{sec:general}. The constants change as follows. The
Coulomb coefficient becomes $\kappa=2MC_{\mathcal R}\alpha Q^2$, so $\kappa/\varepsilon=8\pi C_{\mathcal R}\alpha^2Q^2$,
and the constants of Theorems~\ref{thm:spinless} and~\ref{thm:pauli} acquire the factor
$(2C_{\mathcal R}/6)^2$. The Zeeman constant enters only through $b^2=3\tau_{\mathcal R}^2|B|^2$, in $L_2$ of
Lemma~\ref{lem:onebody} and in $\int b^2$ in the proof of Theorem~\ref{thm:dirac}; so $L_2$ is
multiplied by $4\tau_{\mathcal R}^2$, both branches of $\kappa_P$ by $1/(4\tau_{\mathcal R}^2)$, and $C_P$ is
unchanged. This gives the constants displayed after Theorem~\ref{thm:groups}.

\emph{Several representations.} If the particles carry different representations, the same argument bounds the self-energy of
particle $i$ by $\|\mathcal K_i\|C_{\mathcal R_i}$, so that $V_A\ge-\sum_iC_{\mathcal R_i}e_i^2/d_i$, and bounds its
Zeeman term by $\sqrt3\,\tau_{\mathcal R_i}|B|$; the adjoint kernel $K_A$ is the same for all
particles, since it is fixed by the group. The kinetic operators now depend on the representation
through $\mathcal R_i(A)$, so the one-body bounds are applied species by species, as in
Remark~\ref{rem:nucleiproof}. Since every one-body bound, including the field term of
Lemma~\ref{lem:onebody}, is linear in the number of internal states, only the total $q$ enters.
Replacing $C_{\mathcal R_i}$ and $\tau_{\mathcal R_i}$ by their maxima gives \textup{(i)} when all species
have the same kinetic energy.

\emph{Mixtures of kinetic energies.} If the species have different kinetic energies, the one-body
bounds are again applied species by species. Let $\lambda=C_{\mathcal R}\alpha Q^2q^{1/3}/\ell_0$. If all species are
non-relativistic, Lemma~\ref{lem:onebody} holds also for the spinless ones, since its proof uses only
$T\ge0$ and $T\ge p_A^2-b$; so Proposition~\ref{prop:kinP} and the proof of Theorem~\ref{thm:pauli}
apply. If all species are relativistic, the proof of Theorem~\ref{thm:dirac} applies, since a species
with kinetic energy $\sqrt{p_A^2+m^2}$ contributes at least $(1-\lambda)|p_A|\ge0$; if no species is a Dirac
particle, the field energy is not needed, and $\lambda\le1$ suffices, as in Theorem~\ref{thm:rel}. If both kinds are present,
suppose that $\lambda\le\frac12$, and let $c=2\lambda$. By the spectral theorem, a non-relativistic particle has
$p_A^2/2m\ge c|p_A|-mc^2/2$, or $(\sigma\cdot p_A)^2/2m\ge c|\sigma\cdot p_A|-mc^2/2$; a relativistic one has
$\sqrt{p_A^2+m^2}\ge|p_A|$, or $D_A\ge|\balpha\cdot p_A|$ on the range of $\Lambda_A^+$. Moreover,
Propositions~\ref{prop:electro} and~\ref{prop:relkin} give $\alpha V_A\ge-\lambda\sum_i|p_A|_i$. A spinless
particle then contributes at least $\lambda|p_A|$ or $(1-\lambda)|p_A|$, both nonnegative. A particle with spin
contributes $c(|\sigma\cdot p_A|-\frac12|p_A|)$ or $|\balpha\cdot p_A|-\lambda|p_A|$. These terms are bounded as in
the proof of Theorem~\ref{thm:dirac}, since $(\balpha\cdot p_A)^2$ consists of two copies of
$(\sigma\cdot p_A)^2$; as $c\le1$ and $\lambda\le\frac12$, the field energy pays for them if
$24\pi\tau_{\mathcal R}^2\alpha qL_D\le(\frac34)^{3/2}$. What remains is $-\frac{c^2}2\sum_im_i$, summed over the
non-relativistic particles, which is at least $-(2C_{\mathcal R}^2q^{2/3}/\ell_0^2)\,M\alpha^2Q^4N$.
So for a mixture of both kinds the threshold is $\lambda\le\frac12$, together with
$24\pi\tau_{\mathcal R}^2\alpha qL_D\le(\frac34)^{3/2}$ if some species has spin, and
$K=2C_{\mathcal R}^2q^{2/3}/\ell_0^2$. When all species are non-relativistic, the threshold and $K$ are those
of Theorem~\ref{thm:pauli}, or of Theorem~\ref{thm:spinless}, with no threshold, if all are spinless; when
all are relativistic, the threshold is that of Theorem~\ref{thm:dirac}, or of Theorem~\ref{thm:rel} if none is a Dirac
particle. In all cases the constants are replaced as in the display after
Theorem~\ref{thm:groups}.

\emph{Product groups.} For $G=G_1\times\dots\times G_r$ with couplings $\alpha_1,\dots,\alpha_r$, take the invariant
inner product $\sum_j(\alpha/\alpha_j)\langle\cdot,\cdot\rangle_j$, where $\langle\cdot,\cdot\rangle_j$ is an invariant inner
product on the Lie algebra of $G_j$. The field energy then splits into $\sum_j\frac1{8\pi\alpha_j}\int|B_j|^2$,
and the Coulomb terms of factor $j$ carry the coupling $\alpha_j$.

\emph{(ii) Relativistic collapse.} At $A=0$ the kernel is $K_0^{ab}=\delta^{ab}/(4\pi|x-y|)$, whatever the
representations of the particles. Put the two particles in an eigenvector of $\sum_a\mathcal T^a\otimes\mathcal T^a$ with eigenvalue
$-\kappa_{\mathcal R}$, antisymmetric in colour if the particles belong to one species; the spatial
part of the trial state is symmetric. There $\alpha V_0=-\alpha e_1e_2\kappa_{\mathcal R}/|x_1-x_2|$, and the trial
states of Theorem~\ref{thm:relcollapse} give collapse if $\alpha e_1e_2\kappa_{\mathcal R}>4/\pi$, which is the stated condition.
For two particles in different representations $\mathcal R_1$ and $\mathcal R_2$, let
$-\kappa_{\mathcal R_1\mathcal R_2}$ be the smallest eigenvalue of $\sum_a\mathcal T^a_{(1)}\otimes\mathcal T^a_{(2)}$. The same
trial state then gives collapse if $\alpha e_1e_2\kappa_{\mathcal R_1\mathcal R_2}>4/\pi$; for a quark and an antiquark,
the colour singlet gives $\kappa_{\mathcal R_1\mathcal R_2}=C_{\mathcal R}$.

\emph{(iii) Pauli collapse.} Fix $v_1,v_2$ and $t$ as in the statement. As in the proof of Theorem~\ref{thm:collapse}, a field along $t$ acts on
each $v_j$ as an abelian field, so both particles can be put in zero modes. Let $(a,\psi_0)$ be the Loss--Yau pair of Section~\ref{sec:collapse}, so that
$\psi_0$ and $\tilde\psi_0=\sigma_2\overline{\psi_0}$ are zero modes of $\sigma\cdot(p+a)$ and $\sigma\cdot(p-a)$, respectively.
Take $A_j=\omega^{-1}a_jt$. Since $\mathcal R(A_j)=\omega^{-1}a_j\mathcal R(t)$ and $\mathcal R(t)v_j=-i(t\cdot w_j)v_j$, the operator $\sigma\cdot p_A$ maps a function
$f\otimes v_j$, with $f$ spinor-valued, to $\bigl(\sigma\cdot(p-(t\cdot w_j/\omega)a)f\bigr)\otimes v_j$. So it annihilates
$\phi_j=\pi^{-1}\psi_0\otimes v_j$ if $t\cdot w_j=-\omega$, and $\phi_j=\pi^{-1}\tilde\psi_0\otimes v_j$ if $t\cdot w_j=\omega$.

Particles of different species are distinguishable (Section~\ref{sec:hamiltonians}), and for them we take the
product state $\phi_1\otimes\phi_2$. Two particles of one species need an antisymmetric state. If
$t\cdot w_1=-t\cdot w_2$, the orbitals $\phi_1$ and $\phi_2$ have the spinor parts $\psi_0$ and $\tilde\psi_0$, so they are
orthogonal, and we take their Slater determinant. In both cases the kinetic energy vanishes.

The operator $[t,\cdot\,]$ vanishes on $\mathfrak t$ and preserves $\mathfrak t^\perp$, so $K_A$ is block diagonal, with
$K_A^{ab}=\delta^{ab}/(4\pi|x-y|)$ for $t_a,t_b\in\mathfrak t$. Moreover, $\phi_j^\dagger\mathcal T^a\phi_j$ equals $\rho\,w_j^a$ for
$t_a\in\mathfrak t$ and vanishes for $t_a\in\mathfrak t^\perp$. For the product state this gives
$\langle\alpha V_A\rangle=\alpha e_1e_2(w_1\!\cdot w_2)I_{\LY}$. For the Slater determinant, use the formula for
$\langle V_A\rangle$ in the proof of Theorem~\ref{thm:collapse}, with $\tau^a$ replaced by $\mathcal T^a$. The direct term
now gives $\rho(x)\rho(y)(w_1+w_2)^a(w_1+w_2)^b$. In the exchange term, the terms with $k\ne l$ vanish as there,
since the spinor parts are $\psi_0$ and $\tilde\psi_0$, and the terms with $k=l$ give
$\rho(x)\rho(y)(w_1^aw_1^b+w_2^aw_2^b)$. The difference is $\rho(x)\rho(y)(w_1^aw_2^b+w_2^aw_1^b)$, so
$\langle\alpha V_A\rangle$ is again $\alpha e_1e_2(w_1\!\cdot w_2)I_{\LY}$. Since the components of $A$ commute and $|t|=1$, $|B|^2=\omega^{-2}|\curl a|^2$, so the
field energy is $18\pi^2/(8\pi\alpha\omega^2)$. Dilation gives collapse if
$\alpha|e_1e_2\,w_1\!\cdot w_2|/\pi>18\pi^2/(8\pi\alpha\omega^2)$, which is the stated condition. For
$SU(2)$ with $\mathcal T^a=\tau^a/2$ this reproduces Theorem~\ref{thm:collapse}.
\end{proof}

\appendix
\section{Derivation of the model}\label{app:derivation}

We derive the Coulomb potential~\eqref{eq:VAintro} and the Hamiltonians of Section~\ref{sec:model} from
the classical Yang--Mills Lagrangian coupled to matter. The coupling constant $g$ is scaled into the
gauge field: $A_\mu$ is $\su$-valued, the matter field sees $\partial_\mu+A_\mu$, and the field energy carries the factor $g^{-2}=(4\pi\alpha)^{-1}$. The derivation has four steps:
the Gauss law, the energy, the elimination of the longitudinal electric field, and the quantization of the
matter field.

Let the temporal component $A_0$ and $A=(A_1,A_2,A_3)$ be $\su$-valued, and let $E=\nabla_AA_0-\partial_tA$, where
$\nabla_A$ acts on $A_0$ in the adjoint representation. For a Schr\"odinger field $\psi$ of mass $m$
with values in the internal space of Section~\ref{sec:model}, the Lagrangian is
\begin{equation}\label{eq:lagr}
\mathcal L=\frac1{2g^2}\int\bigl(|E|^2-|B|^2\bigr)
+\int\Bigl(\operatorname{Re}\bigl\langle\psi,i(\partial_t+A_0)\psi\bigr\rangle_{\C}
-\frac1{2m}|\nabla_A\psi|^2\Bigr).
\end{equation}
Since $A_0$ is anti-Hermitian on $\C^2_{\rm col}$, the term $\langle\psi,iA_0\psi\rangle_{\C}$ is real.

\emph{The Gauss law.} $\mathcal L$ contains no time derivative of $A_0$, so $A_0$ is a Lagrange multiplier.
Vary $A_0$ by a compactly supported $\delta A_0$. This changes $E$ by $\nabla_A\delta A_0$, so
\[
\delta\,\frac1{2g^2}\int|E|^2=\frac1{g^2}\int\langle E,\nabla_A\delta A_0\rangle
=-\frac1{g^2}\int\langle\nabla_A\cdot E,\delta A_0\rangle .
\]
The second equality is integration by parts. It holds for $\nabla_A$ because $[A_j,\cdot\,]=A_j\times\cdot$
is antisymmetric for the inner product, so that
$\partial_j\langle u,v\rangle=\langle\nabla_{A,j}u,v\rangle+\langle u,\nabla_{A,j}v\rangle$.
Since $i\,\delta A_0=\frac12\sum_a\delta A_0^a\tau^a$, the matter term changes by
$\int\operatorname{Re}\langle\psi,i\,\delta A_0\psi\rangle_{\C}=\int\sum_a\delta A_0^a\,\tfrac12\psi^\dagger\tau^a\psi
=\int\langle\delta A_0,\rho\rangle$, where $\rho^a=\frac12\psi^\dagger\tau^a\psi$. The two changes must
add up to zero for every $\delta A_0$, which gives the Gauss law
\begin{equation}\label{eq:gauss}
\nabla_A\cdot E=g^2\rho=4\pi\alpha\rho .
\end{equation}

\emph{Energy.} The momentum conjugate to $A$ is $-E/g^2$. In the Legendre transform the terms containing $A_0$ combine to
$\int\langle A_0,g^{-2}\nabla_A\cdot E-\rho\rangle$, which vanishes by~\eqref{eq:gauss}. This gives the energy
\begin{equation}\label{eq:energy}
\mathcal E(\psi,A,E)=\frac1{2g^2}\int\bigl(|E|^2+|B|^2\bigr)+\frac1{2m}\int|\nabla_A\psi|^2 .
\end{equation}

\emph{Longitudinal field.} Write $E=E_\perp-\nabla_A\varphi$ with $\nabla_A\cdot E_\perp=0$ and
$\varphi$ decaying at infinity. The two parts are orthogonal in $L^2$, because
$\int\langle E_\perp,\nabla_A\varphi\rangle=-\int\langle\nabla_A\cdot E_\perp,\varphi\rangle=0$.
By~\eqref{eq:gauss}, $-\Delta_A\varphi=g^2\rho$, so $\varphi=g^2K_A\rho$, and
\[
\frac1{2g^2}\int|\nabla_A\varphi|^2=\frac1{2g^2}\langle\varphi,-\Delta_A\varphi\rangle
=\frac{g^2}2\sum_{a,b}\iint\rho^a(x)K_A^{ab}(x,y)\rho^b(y)\,dx\,dy ,
\]
which, since $g^2=4\pi\alpha$, is the right side of~\eqref{eq:coulombenergy}. The energy thus splits
into the transverse electric energy $\frac1{8\pi\alpha}\int|E_\perp|^2\ge0$, the magnetic energy
$\frac1{8\pi\alpha}\int|B|^2$, the matter kinetic energy, and this instantaneous Coulomb energy.

\emph{Quantization.} Only the matter field is quantized; the gauge field remains classical.
Replace $\psi$ by fermionic field operators, one for each species, and normal-order the Coulomb
energy. This removes the self-interaction of each particle. With several species, $\rho^a$ is the
sum over the species of $\frac12\psi^\dagger\tau^a\psi$. In the sector with $N$ particles, the colour
charge density is $\rho^a(x)=\frac12\sum_i\delta(x-x_i)\tau^a_{(i)}$, and the Coulomb energy becomes
\[
\frac{g^2}2\sum_{i\ne j}\sum_{a,b}K_A^{ab}(x_i,x_j)\,\tfrac14\tau^a_{(i)}\tau^b_{(j)}
=4\pi\alpha\sum_{i<j}\tfrac14\sum_{a,b}K_A^{ab}(x_i,x_j)\,\tau^a_{(i)}\tau^b_{(j)} ,
\]
since $K_A^{ab}(x,y)=K_A^{ba}(y,x)$. This is $\alpha V_A$ with $e_i=\frac12$: the factor $\frac12$ is that of the generators $\tau^a/2$ of
the fundamental representation, since $V_A$ is written with the Pauli matrices $\tau^a$.

\emph{Other kinetic energies.} The matter term in~\eqref{eq:lagr} enters the Gauss law only through
$\operatorname{Re}\langle\psi,iA_0\psi\rangle_{\C}$. Replacing $(2m)^{-1}|\nabla_A\psi|^2$ by the Pauli
energy $(2m)^{-1}|(\sigma\cdot\nabla_A)\psi|^2$, or by the Dirac energy
$\operatorname{Re}\langle\psi,D_A\psi\rangle_{\C}$, therefore leaves~\eqref{eq:gauss} and the Coulomb
energy unchanged. The Pauli energy gives $H_P$, and the Dirac energy gives $H_D$ once the particles are
restricted to the positive spectral subspace of $D_A$, as in Section~\ref{sec:relmodels}. The operator
$\sqrt{p_A^2+m^2}$ in $H_R$ is not derived from a local Lagrangian; it is the standard model for
relativistic effects, as in~\cite{LY88} and~\cite[Chap.~8]{LS10}.

\emph{The Hamiltonians.} Altogether, for $e_i=\frac12$ and the Schr\"odinger, Pauli and Dirac kinetic energies, this gives the
Hamiltonian $H_{\rm tot}$ of~\eqref{eq:Htot}, and the ground state energy $\mathcal E_N$ of
Section~\ref{sec:intro}. The transverse electric energy
$\frac1{8\pi\alpha}\int|E_\perp|^2$ is needed for the time-dependent equations, but not for the ground
state energy: since $E_\perp$ enters the energy only through this nonnegative term, the energy is
smallest when $E_\perp=0$. The reduction to $H_S$, $H_P$, $H_R$ and $H_D$ is explained in
Sections~\ref{sec:hamiltonians} and~\ref{sec:relmodels}.

For classical charge densities, the same argument shows that the Coulomb
energy~\eqref{eq:coulombenergy} is the smallest electric energy compatible with the Gauss law; in
particular, it is at most the electric energy in the Coulomb gauge. We do not claim this for the normal-ordered $N$-particle Hamiltonians: the comparison includes the
self-energy terms, and normal ordering removes them, so the inequality need not survive.

\section{Rough fields}\label{sec:rough}

The stability theorems are proved first for smooth $A$, but the constants do not depend on $A$ at
all.
One therefore expects them to hold for every $A$ for which the Hamiltonians make sense. This appendix shows that the bounds for $H_S$ and $H_R$ hold for every $A\in L^2_{\rm loc}$, and that the bound for $H_P$ holds when, in addition, $\int|B|^2<\infty$. This is the natural class: the form $\int|\nabla_A\psi|^2$ is finite on $C_c^\infty$
exactly when $A\in L^2_{\rm loc}$, and then the quadratic term $[A_j,A_k]$ of the curvature is locally
integrable. The one piece of work is to show that the covariant Laplacian, its heat kernel
and the Coulomb kernel behave well when $A$ is approximated by smooth fields. The plan is as
follows. Lemma~\ref{lem:maxmin} identifies the form domain of $p_A^2$, which is needed to recognize
limits. Lemma~\ref{lem:srs} shows that resolvents and heat kernels converge when $A$ is approximated
by smooth fields. Proposition~\ref{prop:rough} then transfers the diamagnetic inequality, the bound
on $K_A$ and the stability bounds for $H_S$ and $H_R$, and Proposition~\ref{prop:roughpauli} does the
same for $H_P$.

Throughout this appendix $A=(A_1,A_2,A_3)$ is a measurable vector field on $\R^3$ whose components
are anti-Hermitian $n\times n$ matrices with $A_j\in L^2_{\rm loc}$, and $\nabla_A=\nabla+A$ on
$L^2(\R^3;\C^n)$. This covers both representations of an $\su$-valued $A\in L^2_{\rm loc}$: the
fundamental one, and the adjoint one $A\mapsto[A,\cdot\,]$ after complexification. We define $p_A^2$
as the operator of the closure of the form $\int|\nabla_A\psi|^2$ on $C_c^\infty$. The form is
closable: if $\psi_k\to0$ and $\nabla_A\psi_k\to w$ in $L^2$, then for $\varphi\in C_c^\infty$,
$\int\langle\varphi,w_j\rangle=\lim_k\int\langle-\partial_j\varphi-A_j\varphi,\psi_k\rangle=0$,
because $A_j\varphi\in L^2$.

The first lemma says that the minimal form (closure from $C_c^\infty$) and the maximal form (all $u$
with $\nabla_Au\in L^2$) coincide; in the abelian case this is due to Simon~\cite[Theorem~2.2]{Si79b}.

\begin{lemma}[Maximal and minimal forms]\label{lem:maxmin}
Let $u\in L^2(\R^3;\C^n)$ be such that the distribution $\nabla_Au=\nabla u+Au$ lies in $L^2$. Then
there are $u_k\in C_c^\infty$ with $u_k\to u$ and $\nabla_Au_k\to\nabla_Au$ in $L^2$. In particular,
$u$ lies in the form domain of $p_A^2$, and $|u|\in H^1$ with $|\nabla|u||\le|\nabla_Au|$ a.e.
\end{lemma}

\begin{proof}
\emph{Pointwise diamagnetic inequality.} Since $Au\in L^1_{\rm loc}$, $\nabla u=\nabla_Au-Au\in L^1_{\rm loc}$, so
$u\in W^{1,1}_{\rm loc}$. For $\varepsilon>0$ the chain rule gives
$\nabla(|u|^2+\varepsilon^2)^{1/2}=\operatorname{Re}\langle u,\nabla u\rangle(|u|^2+\varepsilon^2)^{-1/2}$.
Since $A$ is anti-Hermitian, $\operatorname{Re}\langle u,Au\rangle=0$, so we may replace $\nabla u$ by
$\nabla_Au$. Letting $\varepsilon\to0$ gives $|u|\in H^1$ and $|\nabla|u||\le|\nabla_Au|$.

\emph{Truncation.} Fix $\phi\in C^\infty([0,\infty))$ with $0\le\phi\le1$, $\phi=1$ on $[0,1]$ and
$\phi=0$ on $[2,\infty)$, and let $u_\lambda=\phi(|u|/\lambda)u$. Then $|u_\lambda|\le2\lambda$, and
\[
\nabla_Au_\lambda=\phi(|u|/\lambda)\nabla_Au+\lambda^{-1}\phi'(|u|/\lambda)(\nabla|u|)\,u .
\]
The last term is bounded by $2\|\phi'\|_\infty\1(|u|\ge\lambda)|\nabla_Au|$, which tends to $0$ in
$L^2$ as $\lambda\to\infty$. Hence $u_\lambda\to u$ and $\nabla_Au_\lambda\to\nabla_Au$ in $L^2$ by
dominated convergence.

\emph{Cutoff.} For $\chi\in C_c^\infty$ with $\chi=1$ near $0$ and $\chi_R=\chi(\cdot/R)$,
$\nabla_A(\chi_Ru)=\chi_R\nabla_Au+(\nabla\chi_R)u\to\nabla_Au$ as $R\to\infty$. So we may assume that
$u$ is bounded with compact support.

\emph{Mollification.} Then $Au\in L^2$, so $\nabla u=\nabla_Au-Au\in L^2$ and $u\in H^1$. Let
$u_k=j_{1/k}*u\in C_c^\infty$ for a mollifier $j$. Then $u_k\to u$ in $H^1$, the $u_k$ are bounded by
$\|u\|_\infty$ and supported in a fixed compact set, and $u_k\to u$ a.e.\ along a subsequence. By
dominated convergence $Au_k\to Au$ in $L^2$, hence $\nabla_Au_k\to\nabla_Au$ in $L^2$.
\end{proof}

\begin{lemma}[Approximation]\label{lem:srs}
Let $A_k$ be smooth with $A_k\to A$ in $L^2_{\rm loc}$. Then
$(p_{A_k}^2+1)^{-1}\to(p_A^2+1)^{-1}$ and $e^{t\Delta_{A_k}}\to e^{t\Delta_A}$ strongly, for every
$t\ge0$.
\end{lemma}

\begin{proof}
Fix $f\in L^2$ and let $u_k=(p_{A_k}^2+1)^{-1}f$. The energy identity
\begin{equation}\label{eq:energyid}
\|\nabla_{A_k}u_k\|^2+\|u_k\|^2=\operatorname{Re}\langle f,u_k\rangle
\end{equation}
gives $\|u_k\|\le\|f\|$ and $\|\nabla_{A_k}u_k\|\le\|f\|$. Pass to a subsequence with
$u_k\rightharpoonup u$ and $\nabla_{A_k}u_k\rightharpoonup w$ weakly in $L^2$.

\emph{Identification of $w$.} For $\varphi\in C_c^\infty$,
$\int\langle\varphi,\nabla_{A_k,j}u_k\rangle=\int\langle-\partial_j\varphi-A_{k,j}\varphi,u_k\rangle$,
first for $u_k\in C_c^\infty$ and then by density. Since $A_{k,j}\varphi\to A_j\varphi$ strongly in
$L^2$, the right side tends to $\int\langle-\partial_j\varphi-A_j\varphi,u\rangle$. Hence
$w=\nabla_Au$ as distributions. By Lemma~\ref{lem:maxmin}, $u$ lies in the form domain of $p_A^2$,
and $\|\nabla_Au\|\le\liminf\|\nabla_{A_k}u_k\|$.

\emph{Identification of $u$.} For $\varphi\in C_c^\infty$,
$\langle\nabla_{A_k}u_k,\nabla_{A_k}\varphi\rangle+\langle u_k,\varphi\rangle=\langle f,\varphi\rangle$.
Since $\nabla_{A_k}\varphi\to\nabla_A\varphi$ strongly, the limit is
$\langle\nabla_Au,\nabla_A\varphi\rangle+\langle u,\varphi\rangle=\langle f,\varphi\rangle$. As
$C_c^\infty$ is a form core, $u=(p_A^2+1)^{-1}f$. The limit does not depend on the subsequence, so
the whole sequence converges weakly.

\emph{Strong convergence.} By~\eqref{eq:energyid},
$\|\nabla_{A_k}u_k\|^2+\|u_k\|^2\to\operatorname{Re}\langle f,u\rangle=\|\nabla_Au\|^2+\|u\|^2$.
Together with the lower semicontinuity of both terms this forces $\|u_k\|\to\|u\|$, so $u_k\to u$
in $L^2$.

\emph{Semigroups.} Let $G_k=(p_{A_k}^2+1)^{-1}$ and $G=(p_A^2+1)^{-1}$, so that $0\le G_k,G\le1$.
Polynomials in $G_k$ converge strongly to the same polynomials in $G$, and $e^{t\Delta_{A_k}}=g_t(G_k)$,
where $g_t(x)=e^{-t(1/x-1)}$ for $x>0$ and $g_t(0)=0$ is continuous on $[0,1]$ for $t>0$. By the
Weierstrass theorem, $g_t(G_k)\to g_t(G)$ strongly.
\end{proof}

\begin{proposition}[Rough fields]\label{prop:rough}
Let $A\in L^2_{\rm loc}$ be $\su$-valued, with $p_A^2$ defined as above in both representations. Then
the diamagnetic inequality of Lemma~\ref{lem:dom} holds for $A$, the kernel $K_A$ exists and
satisfies~\eqref{eq:Kbound}, Proposition~\ref{prop:electro} holds for almost every configuration,
and the bounds of Theorems~\ref{thm:spinless} and~\ref{thm:rel} hold for $A$ with the same constants. In particular,
$H_S$ is stable of the first and of the second kind for every $A\in L^2_{\rm loc}$ and every $\alpha$.
\end{proposition}

\begin{proof}
Let $A_k=j_{1/k}*A$. These are smooth, $\su$-valued, and $A_k\to A$ in $L^2_{\rm loc}$, in both
representations. By Lemma~\ref{lem:srs}, $e^{t\Delta_{A_k}}\to e^{t\Delta_A}$ strongly in both
representations.

\emph{Diamagnetic inequality.} Along a subsequence $e^{t\Delta_{A_k}}f\to e^{t\Delta_A}f$ a.e., so the
pointwise inequality $|e^{t\Delta_A}f|\le e^{t\Delta}|f|$ of Lemma~\ref{lem:dom} passes to the limit.
It gives $\|e^{t\Delta_A}f\|_\infty\le(4\pi t)^{-3/2}\|f\|_1$ for $f\in L^1\cap L^2$, so each matrix entry of $e^{t\Delta_A}$ extends to a bounded operator from $L^1$ to $L^\infty$. Such an operator is an integral operator with a bounded
kernel~\cite[Theorem~1.3 and its proof]{AB94}, and by duality this kernel is bounded a.e.\ by the Gaussian kernel of $e^{t\Delta}$. The kernels can be chosen jointly measurable in $(t,x,y)$, for instance as a.e.\ limits of the
continuous functions $\langle j_\varepsilon(\cdot-x)w,e^{t\Delta_A}j_\varepsilon(\cdot-y)w'\rangle$ with
$w,w'\in\C^n$. Integrating in $t$, $K_A$ exists as an a.e.\ defined kernel satisfying~\eqref{eq:Kbound}.

\emph{Coulomb kernels.} For $f,g\in L^1\cap L^\infty$, dominated convergence in $t$ gives
$\langle f,K_{A_k}g\rangle\to\langle f,K_Ag\rangle$. Since $|K_{A_k}|\le(4\pi|x-y|)^{-1}\in L^1_{\rm loc}$,
on every compact set the sequence is dominated by an integrable function. By the
Dunford--Pettis theorem~\cite[Sect.~19.5]{RF10}, every subsequence therefore has a further
subsequence that converges weakly in $L^1_{\rm loc}(\R^6)$, and every such limit is $K_A$, since a
function in $L^1_{\rm loc}$ is determined by its integrals against products $f(x)g(y)$. Hence
$K_{A_k}\to K_A$ weakly in $L^1_{\rm loc}(\R^6)$. Consequently $\langle\Psi,V_{A_k}\Psi\rangle\to\langle\Psi,V_A\Psi\rangle$ for
$\Psi\in C_c^\infty$. Applying Proposition~\ref{prop:electro} to $A_k$ and passing to the limit gives
$\langle\Psi,(V_A+3\sum_ie_i^2/d_i)\Psi\rangle\ge0$ for all $\Psi\in C_c^\infty$. Hence the matrix
$V_A(\ux)+3\sum_ie_i^2/d_i$ is positive semidefinite for a.e.\ $\ux$.

\emph{Kinetic energies.} $\langle\psi,p_{A_k}^2\psi\rangle\to\langle\psi,p_A^2\psi\rangle$ for
$\psi\in C_c^\infty$, since $A_k\psi\to A\psi$ in $L^2$. For $|p_A|$, write
$\langle\psi,|p_{A_k}|\psi\rangle=\frac1{2\sqrt\pi}\int_0^\infty(\|\psi\|^2-\langle\psi,e^{t\Delta_{A_k}}\psi\rangle)
t^{-3/2}dt$ and use dominated convergence, with majorant $\min(2\|\psi\|^2,tC)\,t^{-3/2}$.

\emph{Conclusion.} The bound of Theorem~\ref{thm:spinless} for $A_k$ passes to the limit on
$C_c^\infty$. For Theorem~\ref{thm:rel} we pass to the limit in the stronger inequality
$\sum_i|p_{A_k}|_i+\alpha V_{A_k}\ge0$ from its proof, and then use $\sqrt{p_A^2+m^2}\ge|p_A|$.
By~\eqref{eq:Kbound}, as in the proof of Proposition~\ref{prop:firstkind}, $|V_A|\le3Q^2\sum_{i<j}|x_i-x_j|^{-1}$. By the pointwise diamagnetic inequality (Lemma~\ref{lem:maxmin}) and Hardy's inequality, this is form-bounded with respect to
$\sum_ip_{A,i}^2$ with relative bound $0$. The diamagnetic inequality for the heat kernels and the
formula above give $\langle u,|p_A|u\rangle\ge\langle|u|,|p|\,|u|\rangle$, so by the relativistic Hardy inequality $V_A$ is also form-bounded with respect to $\sum_i|p_A|_i$, with a constant depending on
$N$. Hence the energy is continuous in the form norm of the kinetic energy, and the bounds extend from
the core $C_c^\infty$ to the whole form domain. (By construction $C_c^\infty$ is a form core for $p_A^2$, and
hence also for $|p_A|$, since a core for $(p_A^2)^{1/2}$ is a core for $(p_A^2)^{1/4}$.)

\emph{Definition of $H_S$ and $H_R$.} For $H_S$ the relative bound is $0$, so $H_S$ is
defined by the KLMN theorem~\cite{RS75}. It remains to define $H_R$ for rough $A$. The proof of
Proposition~\ref{prop:relkin} gives the constant $0.0852>\ell_0$, and that bound passes
to rough $A$ in the same way. So under the hypothesis of Theorem~\ref{thm:rel},
$\sum_i|p_A|_i+\alpha V_A\ge c\sum_i|p_A|_i$ with $c=1-0.085/0.0852>0$. Since $\sqrt{p^2+m^2}-|p|\ge0$,
the form $\sum_i\sqrt{p_{A,i}^2+m_i^2}+\alpha V_A$ is at least $c$ times the kinetic form, and by the relativistic Hardy bound above it is at most a constant multiple of it. Its form norm is therefore
equivalent to that of the kinetic energy, so it is closed on the form domain of the kinetic energy,
and $H_R$ is the associated operator.
\end{proof}

The only property of $A$ used in Lemmas~\ref{lem:maxmin} and~\ref{lem:srs} is that it takes
anti-Hermitian values, so Proposition~\ref{prop:rough} extends verbatim to the compact groups of
Theorem~\ref{thm:groups}.

\begin{proposition}[Rough fields with spin]\label{prop:roughpauli}
Let $A\in L^2_{\rm loc}$ be $\su$-valued, and suppose that the distribution
$F_{jk}=\partial_jA_k-\partial_kA_j+[A_j,A_k]$ is a function with $\int|B|^2<\infty$.
Define $(\sigma\cdot p_A)^2$ as the operator of the closure of the form
$\int|(\sigma\cdot\nabla_A)\psi|^2$ on $C_c^\infty$. Then~\eqref{eq:pauli} holds in the sense of forms,
and the bound of Theorem~\ref{thm:pauli}, with the same constants, and
Proposition~\ref{prop:firstkind} hold for $A$.
\end{proposition}

\begin{proof}
\emph{The Zeeman identity.} For $\psi\in C_c^\infty$ we have $\nabla_{A,k}\psi\in L^2$. For $j\ne k$, expand
$\int\bigl(\langle\nabla_{A,j}\psi,\sigma_j\sigma_k\nabla_{A,k}\psi\rangle-\langle\nabla_{A,k}\psi,\sigma_j\sigma_k\nabla_{A,j}\psi\rangle\bigr)$
using $\nabla_A=\nabla+A$. The terms with two derivatives of $\psi$ cancel after integration by parts.
The terms linear in $A$ give the distribution $\partial_jA_k-\partial_kA_j$ tested against the smooth,
compactly supported matrix $\psi\psi^\dagger$, and the terms quadratic in $A$ give $[A_j,A_k]\in L^1_{\rm loc}$.
Together they give $-\langle\psi,\sigma_j\sigma_kF_{jk}\psi\rangle$, since $F_{jk}\in L^2_{\rm loc}$. With
$\sigma_j\sigma_k=\delta_{jk}+i\sum_l\varepsilon_{jkl}\sigma_l$ this gives
\[
\int|(\sigma\cdot\nabla_A)\psi|^2=\int|\nabla_A\psi|^2-\sum_{j<k}\langle\psi,\sigma_j\sigma_kF_{jk}\psi\rangle
=\int|\nabla_A\psi|^2-\Bigl\langle\psi,i\sum_l\sigma_l\otimes B_l\,\psi\Bigr\rangle .
\]
This is~\eqref{eq:pauli} on $C_c^\infty$. With Lemma~\ref{lem:zeeman} it gives
$(\sigma\cdot p_A)^2\ge p_A^2-b$, where $b=\frac{\sqrt3}2|B|$. As in the proof of
Proposition~\ref{prop:firstkind}, $b$ is form-bounded with respect to $p_A^2$ with relative bound $0$. By
the KLMN theorem the form is closable, and its closure has the same domain as $p_A^2$, with core
$C_c^\infty$.

\emph{Theorem~\ref{thm:pauli}.} Its proof uses exactly four properties of $A$: the diamagnetic
inequality for the heat kernel, through Lemma~\ref{lem:count}; \eqref{eq:Kbound} and Proposition~\ref{prop:electro};
the pointwise diamagnetic inequality, through Lemma~\ref{lem:ball}; and the inequalities $(\sigma\cdot p_A)^2\ge p_A^2-b$
and $(\sigma\cdot p_A)^2\ge0$. The first two hold by Proposition~\ref{prop:rough}. The third holds by Lemma~\ref{lem:maxmin}, and the fourth was
just shown. So the proofs of Lemma~\ref{lem:onebody}, Proposition~\ref{prop:kinP} and
Theorem~\ref{thm:pauli} apply without change on $C_c^\infty$. Since $b$ and $V_A$ are form-bounded
with respect to the kinetic energy, the energy is continuous in its form norm, and the bound extends
to the whole form domain. The same argument applies to the proof of
Proposition~\ref{prop:firstkind}, and to the compact groups of Theorem~\ref{thm:groups}, with the constant
of the Zeeman bound replaced as there.
\end{proof}

Extending Theorem~\ref{thm:dirac} to rough fields would require a self-adjoint realization of $D_A$
for rough $A$ with the properties used in Section~\ref{sec:rel}. We do not pursue this.

\section{The Cwikel--Lieb--Rozenblum bound for \texorpdfstring{$p_A^2$}{p\_A\textasciicircum 2}}\label{sec:clr}

Lemma~\ref{lem:count} rests on the following extension of the CLR bound to $\C^q$-valued functions. The
idea is simple. In Lieb's proof of the CLR bound, as given by Rozenblum and Solomyak, the semigroup $e^{-tp_A^2}$
enters the constant only through integrals of products of its integral kernels against nonnegative weights. For scalar
functions, kernels dominated by those of $e^{t\Delta}$ can therefore be replaced by them, which gives the same
constant as without a field. For $\C^q$-valued functions the kernels are $q\times q$ matrices, and their products
appear under a trace. Since the trace of a product of $q\times q$ matrices is bounded in absolute value by $q$ times the
product of their norms, and Lemma~\ref{lem:dom} bounds these norms by the kernels of $e^{t\Delta}$, the same replacement costs only
a factor~$q$.

\begin{proposition}[CLR bound for $p_A^2$]\label{prop:clr}
Let $A$ be smooth, or as in Proposition~\ref{prop:rough}, and let $0\le W\in L^{3/2}$. Then the number
$N_0(p_A^2-W)$ of negative eigenvalues of $p_A^2-W$ satisfies $N_0(p_A^2-W)\le q\,L_C\int W^{3/2}$.
\end{proposition}

\begin{proof}
For $q=1$ the bound is proved by Rozenblum and Solomyak~\cite[Theorem~2.4 and Sect.~3.2]{RS98}, with
Lieb's constant~\cite[Sect.~3.1]{RS98}. We follow their proof~\cite[Sects.~4 and~5]{RS98} and point out where the
matrix structure enters. As there, we first take $W\in L^1\cap L^\infty$ and fix $r>0$. Let $P^A_t(x,y)$ and $P^0_t(x,y)$
be the integral kernels of $e^{-t(p_A^2+r)}$ and $e^{-t(-\Delta+r)}$. By Lemma~\ref{lem:dom},
$\|P^A_t(x,y)\|\le P^0_t(x,y)$. By the Trotter product formula, the same bound holds when a nonnegative multiple of $W$ is
added to both operators, since multiplication by $e^{-sW}$, $s\ge0$, is scalar and lies between $0$ and $1$.

Apart from general operator theory, the proof uses the domination in two places. First, in trace-class bounds
that only ensure finiteness~\cite[Lemmas~4.2 and~4.3]{RS98}; these follow from the kernel bounds, with constants
depending on $q$. Second, in the trace estimate~\cite[Proposition~5.1]{RS98}, which decides the constant. In its proof, a
scalar regularization and the averaged Trotter formula~\cite[(5.11) and Lemma~4.4]{RS98} express the traces as
limits of integrals~\cite[(5.12)]{RS98}, which for $\C^q$-valued functions take the form
\[
\int\1_\Omega(y)\,P^0_{\varepsilon/2}(y,x_0)\,\tr_{\C^q}\Bigl[\,\prod_{k=1}^nP^A_{t/n}(x_{k-1},x_k)\Bigr]\,P^0_{\varepsilon/2}(x_n,y)\,
\mathcal W(x_0,\dots,x_n)\,dy\,dx_0\cdots dx_n,
\]
with $\Omega$ a set of finite measure, $\varepsilon>0$, and a weight $\mathcal W\ge0$
that does not depend on $A$. In the scalar case, Rozenblum and Solomyak bound this integral by replacing each $P^A_{t/n}$
by $P^0_{t/n}$. For $q>1$ the replacement costs a factor $q$, since
\[
\Bigl|\tr_{\C^q}\prod_{k=1}^nP^A_{t/n}(x_{k-1},x_k)\Bigr|\le q\prod_{k=1}^n\|P^A_{t/n}(x_{k-1},x_k)\|\le q\prod_{k=1}^nP^0_{t/n}(x_{k-1},x_k).
\]

The rest of the proof, including the limit $r\to0$ and the extension to general $W$, involves only the kernels $P^0_t$ and operator
inequalities~\cite[Sects.~5.2 and~5.3]{RS98}. The result is $N_0(p_A^2-W)\le0.1156\,q\int W^{3/2}$,
and $L_C=0.116$ rounds Lieb's constant up. Since only a.e.\ bounds on the kernels were used, and these hold for
rough $A$ by Proposition~\ref{prop:rough}, the proof applies to rough $A$ as well.
\end{proof}

\section*{Statements and Declarations}
\addcontentsline{toc}{section}{Statements and Declarations}

\paragraph{Acknowledgements and funding.} JPS and RR were partially supported by the European
Union's Seventh Framework Programme under the ERC Grant Agreement No.\ 321029. RR was partially
supported by the Royal Society through a Newton International Fellowship (NF160313) and by Darwin
College Cambridge through a Research Fellowship. JPS was partially supported by the Villum Foundation through the Villum Centre of Excellence
for the Mathematics of Quantum Theory (QMATH), Grant No.\ 10059, and the Villum Investigator
grant VIL73411. DMAS has been partially
supported by STFC consolidated grant ST/P000681/1 and St John's College, Cambridge. RR would like to
thank the Isaac Newton Institute for Mathematical Sciences, Cambridge, for support and hospitality
during the programme Mathematics of Many-Body Entanglement, where work on this paper was
undertaken. This work was supported by EPSRC grant EP/Z000580/1. RR was partially supported by a grant from the Simons Foundation, through a Simons Foundation Fellowship of the Newton Institute. RR
is a member of the Gruppo Nazionale per la Fisica Matematica (GNFM) of the Istituto Nazionale di
Alta Matematica Francesco Severi (INdAM).



\paragraph{Use of large language models.} This work made substantial use of Claude Fable 5.1 and Opus 5.5 (Anthropic), which the authors used interactively.
\begin{itemize}[leftmargin=1.5em,itemsep=1pt]
\item \emph{By the authors.} They posed the problem, formulated the model and its Hamiltonian
description (including the analysis at $A=0$), and proposed the strategy for the kinetic estimates:
a splitting of the particles combined with Lieb--Thirring-type bounds. They also suggested using
the harmonic extension in the electrostatic estimate.
\item \emph{With the model, under the direction of the authors:}
\begin{itemize}[leftmargin=1.5em,itemsep=0pt]
\item the non-abelian Newton theorem (Lemma~\ref{lem:newton}) and the argument with the matrix of
mutual energies in Proposition~\ref{prop:electro};
\item the form of the random splitting used in Section~\ref{sec:kin};
\item the superadditivity of the restricted form (Lemma~\ref{lem:superadd}), and the relativistic
and Dirac results built on it;
\item Theorems~\ref{thm:free}, \ref{thm:groups} and~\ref{thm:bosons}, Propositions~\ref{prop:rough}
and~\ref{prop:roughpauli} with Lemmas~\ref{lem:maxmin} and~\ref{lem:srs}, the write-up of the
derivation in Appendix~\ref{app:derivation}, the proof of Lemma~\ref{lem:bks} for unbounded
operators, and the extension of the Cwikel--Lieb--Rozenblum bound to $p_A^2$ in Appendix~\ref{sec:clr};

\item the explicit constants;
\item the first draft of the text and much of the subsequent editing.
\end{itemize}
\end{itemize}
The authors checked every argument in this paper. They take full responsibility for its
content.

\end{document}